\documentclass[article,11pt]{article}
\usepackage[utf8]{inputenc}
\usepackage{float,graphicx,verbatim,fullpage,hyperref,amssymb,amsmath,amsthm,enumerate,multicol, xspace,xcolor,mathtools,thmtools,thm-restate,cleveref,xspace,tikz,caption,subcaption,wasysym, enumitem}
\usepackage[margin=1in]{geometry}
\usepackage{algorithm, cite}
\usepackage[noend]{algpseudocode}
\usepackage{enumitem}
\usepackage{comment}
\usetikzlibrary{calc}

\usepackage{mleftright}
\usepackage{hyperref}
    \hypersetup{ colorlinks=true, linkcolor=blue, filecolor=magenta, urlcolor=blue, citecolor=red}

\usetikzlibrary{arrows.meta}
\tikzset{>={Latex[width=1.5mm,length=1.5mm]}}

\newfloat{procedure}{htbp}{loa}
\floatname{procedure}{Procedure}

\def\R{\mathbb{R}}

\newcommand{\cP}{\mathcal{P}}

\newcommand{\opt}{\textsf{OPT}}
\newcommand{\lp}{\textsf{LP}}

\def\ep{\varepsilon}
\def\tO{\tilde{O}}

\newtheorem{theorem}{Theorem}[section]

\newtheorem{lemma}[theorem]{Lemma}
\newtheorem{claim}[theorem]{Claim}

\theoremstyle{definition}
\newtheorem{definition}[theorem]{Definition}

\usepackage{environ}

\newcommand{\tuple}[1]{\left(#1\right)} 
\newcommand{\eps}{\varepsilon}

\def\*#1{\mathbf{#1}}
\def\+#1{\mathcal{#1}}

\NewEnviron{problem}[1]{
	\begin{center}\fbox{\parbox{6in}{
				{\centering\scshape #1\par}
				\parskip=1ex
				\everypar{\hangindent=1em}
				\BODY
}}\end{center}}

\newcommand{\poly}{\ensuremath{\mathsf{poly}}}
\newcommand{\polylog}{\ensuremath{\mathsf{polylog}}}

\makeatletter
\newcommand*{\inlineequation}[2][]{
  \begingroup
    \refstepcounter{equation}
    \ifx\\#1\\
    \else
      \label{#1}
    \fi
    \relpenalty=10000 
    \binoppenalty=10000 
    \ensuremath{
      #2
    }
    ~\@eqnnum
  \endgroup
}
\makeatother

\begin{document}

\title{{Online Directed Spanners and Steiner Forests}\thanks{A preliminary version of this work appeared in the Proceedings of APPROX 2021}}

\author{Elena Grigorescu\thanks{Purdue University, Email: \{elena-g, krq\}@purdue.edu} \and Young-San Lin\thanks{Melbourne Business School, Email: y.lin@mbs.edu} \and Kent Quanrud\footnotemark[2] \thanks{E.G and Y.L. were supported in part by NSF CCF-1910659 and NSF CCF-1910411.}}
\date{\today}

\maketitle

\begin{abstract}

We present online algorithms for directed spanners and directed Steiner forests. These are well-studied  network connectivity problems that fall under the unifying framework of online covering and packing linear programming formulations. This framework was developed in the seminal work of Buchbinder  and  Naor (Mathematics of Operations Research, 34, 2009) and is based on primal-dual techniques. 
Specifically, our results include the following:

\begin{itemize}

\item For the {\em pairwise spanner} problem, in which the pairs of vertices to be spanned arrive online, we present an  efficient randomized algorithm with competitive ratio $\tilde{O}(n^{4/5})$ for graphs with general edge lengths, where $n$ is the number of vertices of the given graph. For graphs with uniform edge lengths, we give efficient randomized algorithms with competitive ratio $\tilde{O}(n^{2/3+\ep})$ and $\tilde{O}(k^{1/2+\ep})$, where $k$ is the number of terminal pairs. To the best of our knowledge, these are the first online algorithms for directed spanners. 
In the offline version, the current best approximation ratio for uniform edge lengths is $\tilde{O}(n^{3/5 + \ep})$, due to Chlamt{\'a}{\v{c}},  Dinitz, Kortsarz, and  Laekhanukit (SODA 2017, TALG 2020).

\item For the \emph{directed Steiner forest} problem with uniform costs, in which the pairs of vertices to be connected arrive online,
we present an efficient randomized algorithm with competitive ratio $\tO(n^{2/3 + \ep})$. The state-of-the-art online algorithm for general costs is due to Chakrabarty, Ene, Krishnaswamy, and Panigrahi (SICOMP 2018) and is $\tO(k^{1/2 + \ep})$-competitive. In the offline version, the current best approximation ratio with uniform costs is $\tilde{O}(n^{4/7 + \ep})$, due to Abboud and Bodwin (SODA 2018).

\end{itemize}

To obtain {\em efficient} and {\em competitive} online algorithms, we observe that a small modification of the online covering and packing framework by Buchbinder and Naor implies a polynomial-time implementation of the primal-dual approach with separation oracles, which a priori might perform exponentially many calls to the oracle. We convert the online spanner problem into an online covering problem and complete the rounding-step analysis in a problem-specific fashion.

\end{abstract}

\thispagestyle{empty}
\newpage

\section{Introduction}

\providecommand{\parof}[1]{\mleft(#1\mright)}
\providecommand{\nnreals}{\mathbb{R}_{\geq 0}}
\providecommand{\sizeof}[1]{\mleft| #1 \mright|}
\providecommand{\setof}[1]{\mleft\{ #1 \mright\}}
\providecommand{\reals}{\mathbb{R}}
\providecommand{\rip}[2]{\mleft\langle #1, #2 \mright\rangle}

We study online variants of  directed network optimization problems. In an online problem, the input is presented sequentially, one item at a time, and the algorithm is forced to make irrevocable decisions in each step, without knowledge of the remaining part of the input. The performance of the algorithm is measured by its {\em competitive ratio}, which is the ratio between the value of the online solution and that of an optimal offline solution.

Our main results focus on {\em directed spanners}, which are sparse subgraphs that approximately preserve pairwise distances between vertices. Spanners are fundamental combinatorial objects with a wide range of applications, such as distributed computation \cite{Awerbuch, PelegS89}, data structures \cite{alon1987optimalpreprocessing,Yao1982SpacetimeTF}, routing schemes \cite{CowenW04,PachockiRSTW18,PelegU89a,RodittyTZ08},  approximate shorthest paths \cite{BaswanaK10,DorHZ00, Elkin05}, distance oracles \cite{BaswanaK10, Chechik15, PatrascuR14}, and property testing \cite{AwasthiJMR16,bhattacharyya2012transitive}. For a comprehensive account of the literature, we refer the reader to the excellent survey  \cite{ahmed2019graph}.

We also study related network \emph{connectivity} problems, and in particular \emph{directed Steiner forests}, which are sparse subgraphs that maintain connectivity between target terminal vertex pairs. Steiner forests are broadly used in different areas, such as multicommodity network design \cite{fleischer2006simple,gupta2003approximation}, mechanism design and games \cite{chawla2006optimal,konemann2005primal,konemann2008group,roughgarden2007optimal}, computational biology \cite{khurana2017genome,pirhaji2016revealing}, and computational geometry \cite{bateni2012euclidean,borradaile2015polynomial}.

Our approaches are based on covering and packing linear programming (LP) formulations that fall into the unifying framework developed by Buchbinder and  Naor \cite{buchbinder2009online}, using the powerful primal-dual technique \cite{GoemansW95}.
This unifying framework extends across widely different domains, and hence provides a general abstraction that captures the algorithmic essence of all online covering and packing formulations. 
In our case, to obtain {\em efficient} competitive algorithms for solving the LPs online, we observe that the algorithms in \cite{buchbinder2009online} can be slightly modified to significantly speed up the setting of our applications, in which the algorithm might otherwise make exponentially many calls to a separation oracle. This component is not tailored to the applications studied here and may be of independent interest. In particular, previous approaches solving online covering and packing problems either focus on the competitiveness of the algorithm \cite{aaabn-set-cover,awerbuch1993throughput,bansal2012primal}, or manage to leverage the specific structure of the problem for better time efficiency in a somewhat ad-hoc manner \cite{alon2006general,awerbuch2004line,bansal2008randomized,berman1997line,buchbinder2006improved,imase1991dynamic}. Here, the solution may be viewed as a more {\em unified} framework that is also {\em efficient}.

\subsection{Our contributions}

\subsubsection{Directed spanners}
Let $G = (V,E)$ be a directed simple graph with $n$ vertices. Each edge is associated with its \emph{length} $\ell: E \to \R_{\geq 0}$. The edge lengths are \emph{uniform} if $\ell(e)=1$ for all $e \in E$. In spanner problems, the goal is to
compute a minimum cardinality (number of edges) subgraph in which the distance between
terminals is preserved up to some prescribed factor. In the most
well-studied setting, called the \emph{directed $s$-spanner} problem,
there is a fixed value $s \geq 1$ called the \emph{stretch}, and the goal is 
to find a minimum cardinality subgraph in which \emph{every} pair of
vertices has \emph{distance} within a factor of $s$ in the original
graph. Here, we use the term distance to denote the sum of the edge lengths in a shortest path of two terminal vertices. For low stretch spanners, when $s=2$, there is a tight $\Theta(\log n)$-approximation algorithm \cite{elkin1999client, Kortsarz2001OnTH}; when $s=3,4$ both with uniform edge lengths, there are $\tO(n^{1/3})$-approximation algorithms \cite{berman2013approximation,dinitz2016approximating}.
When $s > 4$, the best known approximation factor is $\tO(n^{1/2})$ \cite{berman2013approximation}.
The problem is hard to approximate within an $O(2^{{\log^{1-\eps} n}})$ factor for $3 \leq s = O(n^{1-\delta})$ and any $\eps, \delta \in (0,1)$, unless $NP\subseteq  DTIME(n^{\operatorname{polylog} n})$ \cite{ElkinP07}.

A more general setting, called the \emph{pairwise
  spanner} problem \cite{chlamtavc2020approximating}, and the {\em client-server} model \cite{bhattacharyya2012transitive,elkin1999client}, considers an arbitrary set of
terminals $D = \{(s_i,t_i) \mid i \in [k]\} \subseteq V \times V$.  Each
terminal pair $(s_i,t_i)$ has its own target distance $d_i$. The goal is
to compute a minimum cardinality subgraph in which for each $i$, the
distance from $s_i$ to $t_i$ is at most $d_i$.  For the pairwise
  spanner problem with uniform edge lengths, \cite{chlamtavc2020approximating} obtains an $\tO(n^{3/5 + \ep})$-approximation.

In the online
version, the graph is known ahead of time, and the terminal pairs and the corresponding target distances are
received one by one in an online fashion.
The distance requirement of the arriving terminal pair is satisfied by irrevocably including edges.
There are no online algorithms for the pairwise spanner problem that we are aware of, even in the simpler and long-studied
case of stretch $s$ or graphs with uniform edge lengths.

For graphs with uniform edge lengths, we prove the following theorem in Section \ref{sec:us}.

\begin{restatable}{theorem}{thmsqrtkps} \label{thm:sqrt-k-ps}
For the online pairwise spanner problem with uniform edge lengths, there exists a randomized polynomial-time algorithm with competitive ratio $\tO(k^{1/2 + \delta})$ for any constant $\delta > 0$.
\end{restatable}

Next, we turn to graphs with general edge lengths and derive online algorithms with competitive ratios in terms of $n$. We present a generic algorithm (Algorithm~\ref{alg:dks}) used for Theorems~\ref{theorem:spanner}, \ref{thm:us}, \ref{theorem:steiner}, and other variants of spanners introduced in Section \ref{subsubsec:online-ps}. For graphs with general edge lengths, we show the following in Section~\ref{subsec:general-ps}.

\begin{restatable}{theorem}{thmpdks}\label{theorem:spanner}
  For the online pairwise spanner
  problem, there is a randomized polynomial-time algorithm with
  competitive ratio $\tO(n^{4/5})$.
\end{restatable}

For graphs with uniform edge lengths without further assumptions, we use Theorem~\ref{thm:sqrt-k-ps} and the generic algorithm to prove the following theorem in Section~\ref{subsec:us}.

\begin{restatable}{theorem}{thmus}\label{thm:us}
  For the online pairwise spanner
  problem with uniform edge lengths, there is a randomized polynomial-time algorithm with
  competitive ratio $\tO(n^{2/3+\ep})$ for any constant $\ep \in (0,1/3)$.
\end{restatable}

\subsubsection{Directed Steiner forests}
In the \emph{directed Steiner forest} problem, we are given a directed graph $G=(V,E)$ with edge costs $w: E \to \R_{\geq 0}$, and a set of terminals $D = \{(s_i,t_i) \mid i \in [k]\} \subseteq V \times V$. The goal is to find a subgraph $H=(V,E')$ which includes an $s_i \leadsto t_i$ path for each terminal pair $(s_i, t_i)$, and the total cost $\sum_{e \in E'}w(e)$ is minimized. The costs are \emph{uniform} when $w(e)=1$ for all $e \in E$.

In the online version, the graph is known ahead of time, and the terminal pairs arrive one by one in an online fashion.
The connectivity requirement of the arriving terminal pair is satisfied by irrevocably including edges.

In the offline setting with general costs, the best known approximations are $O ({k^{1/2 + \ep}})$ by Chekuri et al.\
\cite{chekuri2011set} and $O(n^{2/3 + \ep})$ by Berman et al.\ \cite{berman2013approximation}. For the special case of uniform costs, there is an improved approximation factor of
$\tO(n^{4/7 + \ep})$ by Abboud and Bodwin \cite{abboud2018reachability}.
In the online setting, Chakrabarty et al.\ \cite{cekp} give an
$\tO(k^{1/2 + \ep})$ approximation for general costs. Their algorithm also extends to the more general buy-at-bulk version. We prove the following in Section~\ref{sec:dsf}.

\begin{restatable}{theorem}{thmsteiner}\label{theorem:steiner}
   For the online directed Steiner forest problem with uniform costs, there is a randomized polynomial-time algorithm with competitive ratio $\tO(n^{2/3 + \ep})$ for any constant $\ep \in (0,1/3)$.
\end{restatable}

We essentially improve the competitive ratio when the number of terminal pairs is $\omega(n^{4/3})$.

\subsubsection{Summary}

We summarize our main results for online pairwise spanners and directed Steiner forests in Table \ref{hardness} by listing the competitive ratios and contrast them with the corresponding known competitive and approximation ratios. We note that offline $\tO(n^{4/5})$-approximate pairwise spanners for graphs with general edge lengths and offline $\tO(k^{1/2+\ep})$-approximate pairwise spanners for graphs with uniform edge lengths can be obtained by our online algorithms.

\begin{table}[H]
\begin{center}
\def\arraystretch{1.2}
\begin{tabular}{|*3{l|}}
\hline
\textbf{Setting} & \textbf{Offline} & \textbf{Online}  \\
\hline
\cline{1-3}
Pairwise & $\tO(n^{4/5})$ (implied by Thm~\ref{theorem:spanner}) & $\tO(n^{4/5})$ (Thm~\ref{theorem:spanner}) \\
Spanners & {\color{gray}$\tO(n^{3/5 + \ep})$ (uniform lengths) \cite{chlamtavc2020approximating}} & $\tO(n^{2/3+\ep})$ (uniform lengths, Thm~\ref{thm:us})\\
& $\tO(k^{1/2+\ep})$ (uniform lengths, implied by Thm~\ref{thm:sqrt-k-ps}) & $\tO(k^{1/2+\ep})$ (uniform lengths, Thm~\ref{thm:sqrt-k-ps})\\
\hline
\cline{1-3}
Directed & {\color{gray}$\tO(n^{4/7 + \ep})$  (uniform costs) \cite{abboud2018reachability}} & {\color{gray}$\tO(k^{1/2 + \ep})$ \cite{cekp}} \\
Steiner & {\color{gray}$O(n^{2/3 + \ep})$ \cite{berman2013approximation}} & $\tO(n^{2/3+\ep})$ (uniform costs, Thm~\ref{theorem:steiner})\\
Forests & {\color{gray}$O(k^{1/2 + \ep})$ \cite{chekuri2011set}} &\\
\hline
\end{tabular}
\caption{Summary of the competitive and approximation ratios. Here, $n$ refers to the number of vertices and $k$ refers to the number of terminal pairs. We include the known results for comparison. The text in gray refers to known results while the text in black refers to our contributions.}
\label{hardness}
\end{center}
\end{table}

\subsection{An efficient online covering and packing framework}

Before presenting our modification to the unified framework in \cite{buchbinder2009online} to obtain efficient online covering and packing LP solvers, we give an overview of the well-known primal-dual framework for approximating covering and packing LP's online. This framework is the main engine of our applications and it is important to establish some context before getting into the application for spanners and Steiner forests. We also introduce a discussion of certain technical nuances that arise for our application, and the small modification we propose to address it. A more formal description, including proofs and fully parameterized theorem statements, is fairly technical and therefore deferred to Section~\ref{sec:covering} \emph{after} we have used these tools in the context of spanners and Steiner forests.

The primal-dual framework was first developed for the online set cover
problem in the seminal work of \cite{aaabn-set-cover}. The approach
was extended to network optimization problems in undirected graphs in
\cite{alon2006general}, then abstracted and generalized to a broad LP-based
primal-dual framework in \cite{buchbinder2009online}. Our discussion
primarily centers around the abstract framework in
\cite{buchbinder2009online}. A number of previous results in online
algorithms, such as ski rental \cite{kmmo} and paging \cite{bansal2012primal}, can be recovered from this approach and many new important
applications have since been developed, such as the $k$-server problem \cite{young1994thek}. We refer the
reader to the excellent survey by Buchbinder and Naor \cite{buchbinder2009design}.

These works develop a clean two-step approach to online algorithms
based on 1) solving the LP online, and 2) rounding the LP
online. Solving the LP online can be done in a generic fashion, while
rounding tends to be problem-specific. The setting for the \emph{covering} LP
is the following.
\begin{align}
  \begin{aligned}
    \text{minimize } & \rip{\mathbf{c}}{x} 
    \text{ over } x \in \nnreals^n 
    \text{ s.t.\ } A x \geq \mathbf{b}.
  \end{aligned}
\end{align}

Here, $A \in \R_{\geq 0}^{m \times n}$ consists of $m$ covering
constraints, $\mathbf{b} \in \R_{> 0}^n$ is a positive lower bound of the covering constraints, and $\mathbf{c} \in \R_{> 0}^m$ denotes the positive coefficients of the linear cost function. Each constraint can be normalized, so we focus on covering LP's in the following form.
\begin{align}
  \begin{aligned}
    \text{minimize } & \rip{\mathbf{c}}{x} 
    \text{ over } x \in \nnreals^n 
    \text{ s.t.\ } A x \geq \mathbf{1}
  \end{aligned}
                        \label{equation:covering}
\end{align}
where $\mathbf{1}$ is a vector of all ones.

In the online covering problem, the cost vector $\mathbf{c}$ is given offline, and each of these covering constraints is presented one by one in an online fashion, that is, $m$ can be unknown. The goal is to update $x$ in a non-decreasing manner such that all the covering constraints are satisfied and the objective value $\rip{\mathbf{c}}{x}$ is approximately optimal. An important idea in this
line of work is to simultaneously consider the dual {\em packing} problem:
\begin{align}
  \begin{aligned}
    \text{maximize } & \rip{\mathbf{1}}{y} 
    \text{ over } y \in \nnreals^m 
    \text{ s.t.\ } A^T y \leq \mathbf{c}
  \end{aligned}
                        \label{equation:packing}
\end{align}
where $A^T$ consists of $n$ packing constraints with an upper bound $\mathbf{c}$ given offline.

In the online packing problem, the \emph{columns} of $A^T$ and the corresponding variables are presented online taking initial value zero; one can either let the arriving variable remain zero, or irrevocably assign a positive value to the arriving variable. The goal is to approximately maximize the objective value $\rip{\mathbf{1}}{y}$ with each constraint approximately satisfied.

\paragraph{Separation oracles in the online setting.}

The primal-dual framework in \cite{buchbinder2009online} simultaneously solves
both LP \eqref{equation:covering} and LP \eqref{equation:packing}, and crucially
uses LP-duality and strong connections between the two solutions to
argue that they are both nearly optimal. Here we give a sketch of the LP solving framework for reference in the subsequent
discussion. We maintain solutions $x$ and $y$ for
LP \eqref{equation:covering} and LP \eqref{equation:packing}, respectively, in an online fashion. The covering solution $x$ is a function of the
packing solution $y$. In particular, each coordinate $x_j$ is exponential in the \emph{load} of the corresponding packing constraint in LP \eqref{equation:packing}. Both $x$ and $y$ are monotonically increasing. The algorithm runs in phases, where each phase corresponds to an estimate for $\opt$ revised over time. Within a phase we have the following.
If the new covering constraint $i \in [m]$, presented online, is already satisfied, then there is nothing to be
done. Otherwise, increase the corresponding coordinate $y_i$, which
simultaneously increases the $x_j$'s based on the magnitude of the
coordinate $a_{ij}$, where $a_{ij}$ is the $i$-th row $j$-th column entry of $A$. The framework in \cite{buchbinder2009online} increases $y_i$ until the increased $x_j$'s satisfy the new constraint. This naturally extends to the setting when the problem relies on a \emph{separation oracle} to retrieve an unsatisfied covering constraint where the number of constraints can be unbounded \cite{buchbinder2009online}. However, while this approach will fix all violating constraints, each individual fix may require a diminishingly small adjustment that cannot be charged off from a global perspective. Consequently the algorithm may have to address exponentially many constraints. 

\paragraph{A primal-dual bound on separation oracles.} 
Our goal is to adjust the framework to ensure that we only address a polynomial number of constraints (per phase).
For many concrete problems in the literature, this issue can be addressed directly based on the problem at hand (discussed in greater detail in Section~\ref{sec:related}). In our setting, we start with a combinatorially defined LP that is 
not a pure covering problem, and convert it to a covering LP.
While having a covering LP is conducive to the online LP framework, the machinery generates a large number of covering constraints that are very unstructured. For example, we
have little control over the coefficients of these constraints. This motivates us to develop a more generic argument to bound the number of queries to the separation oracle, based on the online LP framework, more so than the exact problem at hand. Here, when addressing a violated constraint $i$, we instead increase the dual variable $y_i$ until the increased primal variables $x$ (over-)satisfy the new constraint \emph{by a factor of 2}. This forces at least one $x_j$ to be doubled -- and in the dual, this means we used up a substantial amount of the corresponding packing constraint. Since the packing solution is already guaranteed to be feasible in each phase by the overall framework, this leads us to conclude that we only ever encounter polynomially many violating constraints. 

For our modified online covering and packing framework, we show that 1) the approximation guarantees are identical to those in \cite{buchbinder2009online}, 2) the framework only encounters polynomially many violating constraints for the online covering problem, and 3) only polynomially many updates are needed for the online packing problem.

\begin{theorem} \label{thm:inf-covering}
(Informal) There exists an $O(\log n)$-competitive online algorithm for the covering LP \eqref{equation:covering} which encounters polynomially many violating constraints.
\end{theorem}

\begin{theorem} \label{thm:inf-packing}
(Informal) Given any parameter $B>0$, there exists a $1/B$-competitive online algorithm for the packing LP \eqref{equation:packing} which updates $y$ polynomially many times, and each constraint is violated within an $O(f(A)/B)$ factor ($f(A)$ is a logarithmic function that depends on the entries in $A$).
\end{theorem}

We note that the competitive ratios given in \cite{buchbinder2009online} are tight, which also implies the tightness of the modified framework. The number of violating constraints depends not only on the number of covering variables and packing constraints $n$, but also on the number of bits used to present the entries in $A$ and $\mathbf{c}$. The formal proof for Theorem~\ref{thm:inf-covering} is provided in Section~\ref{sec:covering}, while the formal proof for Theorem~\ref{thm:inf-packing} provided in Appendix~\ref{sec:packing} is not directly relevant to this work, but may be of independent interest.

\subsection{High-level technical overview for online network optimization problems}

\paragraph{Online pairwise spanners.}

For this problem, a natural starting point is the flow-based LP approach for offline $s$-stretch directed spanners, introduced in \cite{dinitz2011directed}. 
The results of \cite{chlamtavc2020approximating} adopt a slight tweak for this approach to achieve an $\tO(n/\sqrt{\opt})$-approximation, where $\opt$ is the size of the optimal solution. With additional ideas, the $\tO(n/\sqrt{\opt})$-approximation is converted into an $\tO(n^{3/5 + \ep})$-approximation for pairwise spanners. 
One technical obstacle in the online setting is the lack of a useful lower bound for $\opt$. Another challenge is solving the LP for the spanner problem and rounding the solution in an online fashion, particularly as the natural LP is not a covering LP. We address these technical obstacles as discussed below in Section~\ref{subsubsec:online-ps}. Ultimately we obtain an $\tO(n^{4/5})$ competitive ratio for the online setting. The strategy here is to convert the LP for spanners into a covering LP, where the constraints are generated by an \emph{internal} LP. The covering LP previously appeared in \cite{dinitz2011directed} implicitly, and in \cite{dinitz2019lasserre} explicitly.

\paragraph{Online pairwise spanners with uniform edge lengths.} 

For the special case of uniform edge lengths, \cite{chlamtavc2020approximating} obtains an improved bound of $\tO(n^{3/5 + \ep})$. It is natural to ask if the online bound of $\tO(n^{4/5})$ mentioned above can be improved as well.
Indeed, we obtain an improved bound of $\tO(n^{2/3+ \ep})$ by replacing the greedy approach in the small $\opt$ regime by using the $\tO(k^{1/2+ \ep})$-competitive online algorithm discussed in Section~\ref{sec:us}. This algorithm leverages ideas from \cite{chlamtavc2020approximating} in reducing to label cover problems with ideas from the online network design algorithms of \cite{cekp}. Some additional ideas are required to combine the existing tools and among others we had to formulate a new pure covering LP that can be solved online, to facilitate the transition.

\paragraph{Online Steiner forests with uniform costs.} This problem is a special case of the online pairwise spanner problem where the distance requirement for each terminal pair is infinity and the edge lengths are uniform. The online algorithm for this problem has a similar structure to the one for pairwise spanners and similar obstacles to overcome.

\subsection{Additional background and related work} \label{sec:related}

\paragraph{Streaming, dynamic, and distributed algorithms for spanners.} A model related to online algorithms is that of streaming algorithms. In the streaming model an input is also revealed sequentially, but the algorithm is only allowed to use some small amount of space,  which is sublinear in the length of the stream, and is supposed to maintain an approximate solution. For this model, several papers consider spanner variants, such as undirected or weighted graphs, and additive or multiplicative stretch approximations, and the aim is to build spanners with small size or distortion \cite{Baswana_streamingalgorithm, KapralovW14, filtser2020graph}. 
In a related direction, spanners have also been studied in the setting of dynamic data structures, where the edges of a graph are inserted or removed one at a time and the goal is to maintain an approximate solution with small update time and space \cite{Elkin11,BodwinW16}. 
A relevant model is that of distributed computation where nodes in the network communicate efficiently to build a solution \cite{DerbelGP07,DerbelGPV08, FernandezW020}. 
As mentioned earlier, the survey by Ahmed et al.\ \cite{ahmed2019graph} gives a comprehensive account of the vast literature on spanners, and we refer the reader to the references within. 

\paragraph{Connections to buy-at-bulk formulations.}  In the buy-at-bulk network design problem \cite{awerbuch1997buy}, each edge is associated with a sub-additive cost function of its \emph{load}. Given a set of terminal demands, the goal is to route integral flows from each source to each sink concurrently to minimize the total cost of the routing. This problem is a generalization of various single-source or multicommodity network connectivity problems, including Steiner trees and Steiner forests, in which the cost of each edge is fixed once allocated. 
While most problems admit efficient polylogarithmic approximations in either the online or offline setting for undirected networks \cite{chekuri2011set,awerbuch2004line,berman1997line,gupta2017last}, the problems are much harder for directed networks. In the offline setting, the current best approximation ratio is $O(k^{\ep})$ in polynomial time and $\polylog(n)$ in quasi-polynomial time for the directed Steiner tree problem \cite{zelikovsky1997series,charikar1999approximation}, $O(\min\{k^{1/2 + \ep}, n^{2/3 + \ep}\})$ in polynomial time for the directed Steiner forest problem \cite{berman2013approximation,chekuri2011set}, and $O(\min\{k^{1/2 + \ep}, n^{4/5 + \ep}\})$ in polynomial time for the directed buy-at-bulk problem \cite{antonakopoulos2010approximating}. 
In the online setting for directed networks, \cite{cekp} showed that compared to offline, it suffices to pay an extra polylogarithmic factor, where the polylogarithmic term was later improved by \cite{shen2020online}. The main contribution of \cite{cekp} is essentially bringing the \emph{junction-tree-based} approach into the online setting for connectivity problems. This is the main ingredient that improves the competitive ratio of our online algorithm for pairwise spanners from $\tO(n^{4/5 + \ep})$ when edge lengths are general to $\tO(n^{2/3 + \ep})$ when edge lengths are uniform. Our approach for online pairwise spanners with uniform edge lengths combines this ingredient and the ideas of the offline pairwise spanner framework \cite{chlamtavc2020approximating} which tackles hard distance requirements.

\paragraph{Online LP's and separation oracles.} As previously mentioned, generating separating constraints with an oracle in the online setting is not new.
For example, this arises implicitly in early work on network optimization \cite{alon2006general} and the oracle is discussed explicitly in \cite{buchbinder2009online}.
As a recent example, \cite{gupta2014changing} develops online algorithms for the multistage matroid maintenance problem, which requires solving a covering LP with box constraints online. \cite{gupta2014changing} adjusts the separation oracle to only identify constraints that are violated by at least some constant. Because of the $\{0,1\}$-incidence structure of their LP, the sum of primal variables has to increase by a constant to satisfy such a constraint. Meanwhile the box constraints limit the total sum of primal variables to $O(n)$. This leads to an $O(n)$ bound on the number of separating constraints. While there are strong similarities to our approach, one difference is the use of the $\{0,1\}$-structure and box constraints to obtain their bound. Our comparably unstructured setting required us to develop an argument independent of concrete features such as these.

\paragraph{Other variants of online covering and packing problems.} Beyond linear objectives, there are other variants of online covering and packing problems, which focus on different objectives with linear constraints. This includes optimizing convex objectives \cite{azar2016online} and $\ell_q$-norm objectives \cite{shen2020online}. Other online problem-dependent variants include for instance  mixed covering and packing programs  \cite{azar2013online}, and sparse integer programs \cite{gupta2014approximating}. All these frameworks utilize the primal-dual technique,  which updates the covering and packing solutions simultaneously with some judiciously selected growth rate, to guarantee nice competitive ratio. Instead, our modified framework focuses on the efficiency of online algorithms for fundamental covering and packing problems, which is amenable to applications with exponential or unbounded number of constraints, where a violating one can be searched by an efficient separation oracle.

\subsection{Organization} 
Since the proof of Theorem \ref{thm:sqrt-k-ps} is the most involved contribution of this work, we start by presenting it in Section~\ref{sec:us}.
In Section~\ref{subsubsec:online-ps}, we prove Theorems~\ref{theorem:spanner}, \ref{thm:us}, and \ref{theorem:steiner} by designing and analyzing specific variants of a generic online algorithm.
We show the modified online covering framework in Section~\ref{sec:covering}, while the modified online packing framework is presented in Appendix~\ref{sec:packing}.

\section{Online Pairwise Spanners with Uniform Edge Lengths} \label{sec:us}

In this section, we prove Theorem \ref{thm:sqrt-k-ps}. Namely, we design a randomized online algorithm for the pairwise spanner problem with uniform edge lengths with competitive ratio $\tO(k^{1/2 + \delta})$ for any constant $\delta > 0$. We recall that in the {\em pairwise spanner} problem, we are given a directed graph $G=(V, E)$ with edge length $\ell: E \to \R_{\geq 0}$, a general set of $k$ terminals $D = \{(s_i,t_i) \mid i \in [k]\} \subseteq V \times V$, and a target distance $d_i$ for each terminal pair $(s_i,t_i)$, the goal is to output a subgraph $H=(V, E')$ of $G$ such that for every pair $(s_i, t_i) \in D$ it is the case that $d_H(s_i, t_i)\leq d_i$, i.e., the distance of a shortest $s_i \leadsto t_i$ path is at most $d_i$ in the subgraph $H$, and we want to minimize the number of edges in $E'$. The edge lengths are uniform if $\ell(e)=1$ for all $e \in E$.
In the online setting, the directed graph $G$ is given offline, while the vertex pairs in $D \subseteq V \times V$ arrive online one at a time. In the beginning, $E' = \emptyset$. Suppose $(s_i,t_i)$ and its target distance $d_i$ arrive in round $i$, we select some edges from $E$ and irrevocably add them to $E'$, such that in the subgraph $H=(V,E')$, $d_{H}(s_i,t_i) \leq d_i$.

\subsection{Outline of the proof of Theorem~\ref{thm:sqrt-k-ps}}

We start by describing the high-level approach of our proof of  Theorem~\ref{thm:sqrt-k-ps}. While the proof combines ideas of the online buy-at-bulk framework in  \cite{cekp} and of the reduction from the pairwise spanner problem to a connectivity problem in \cite{chlamtavc2020approximating}, implementing the  details require several new ideas. Specifically, we introduce a useful extension of the Steiner problem, called the {\em Steiner label cover} problem,  and our main contribution is an online covering LP formulation for this problem. This approach allows us to not only capture the global approximation property in an online setting, as in \cite{cekp}, but also to handle distance constraints, as in \cite{chlamtavc2020approximating}. The entire proof consists of three main ingredients: 
\begin{enumerate}
\item We first show that there exists an $O(\sqrt{k})$-approximate solution consisting of \textit{junction trees}. A junction tree is a subgraph consisting of an in-arborescence and out-arborescence rooted at the same vertex (see also Definition \ref{def:junc-tree}).
\item We then show a reduction from the online pairwise spanner problem to the online Steiner label cover problem on a forest with a loss of an
$O(k^{1/2+\delta})$ factor. More precisely, an $O(\sqrt{k})$ factor comes  from the junction tree approximation and an extra $O(k^{\delta})$ factor comes from the \textit{height reduction} technique introduced in \cite{chekuri2011set,helvig2001improved}. The height reduction technique allows us to focus on low-cost trees of height $O(1/\delta)$ in order to  recover a junction tree approximation.

\item Finally, we show a reduction from the online Steiner label cover problem to the online \emph{undirected group Steiner forest} problem, with a loss of a $\polylog(n)$ factor. More precisely, we first formulate an online covering LP for the online Steiner label cover instance, then construct an online undirected group Steiner forest instance from the LP solution, with a loss of a factor of $2$. By \cite{cekp}, the online undirected group Steiner forest problem on the forest that we construct can be solved with competitive ratio $\polylog(n)$.
\end{enumerate}

Combining these three ingredients results in an $\tO(k^{1/2+\delta})$-competitive algorithm. We provide further intuition below. The detailed description of the first, second, and third ingredients are in Sections~\ref{sec:jt-apx}, \ref{sec:lcp-hr}, and \ref{sec:online-slc}, respectively.

\paragraph{Junction tree approximation.}
Many connectivity problems, including Steiner forests, buy-at-bulk, and spanner problems, are usually solved using \emph{junction trees} introduced in \cite{chekuri2010approximation}. 

\begin{definition}\label{def:junc-tree}
A \emph{junction tree} rooted at $r \in V$ is a directed graph $G=(V,E)$, by taking the union of an in-arborescence rooted at $r$ and an out-arborescence rooted at $r$.\footnote{A junction tree does not necessarily have a tree structure in directed graphs, i.e., an edge may be used twice, once in the in-arborescence and once in the out-arborescence. Nevertheless, we continue using this term because of historical reasons. A similar notion can also be used for undirected graphs, where a junction tree is indeed a tree.} A \emph{junction tree solution} is a collection of junction trees rooted at different vertices, that satisfies all the terminal distance constraints.
\end{definition}

\begin{restatable}{lemma}{lemjtapx}\label{lem:jt-apx}
There exists an $O(\sqrt{k})$-approximate junction tree solution for pairwise spanners.
\end{restatable}

In Section~\ref{sec:jt-apx}, we present the proof of Lemma~\ref{lem:jt-apx}.
At a high level, the  proof follows by a standard \emph{density} argument. A \emph{partial solution} is a subgraph that connects a subset of the terminal pairs within the required distances. The density of a partial solution is the ratio between the number of edges used and the number of terminal pairs connected within the required distances. This argument is used for solving offline problems including the Steiner forest problem \cite{chekuri2011set,feldman2012improved,berman2013approximation}, the buy-at-bulk problem \cite{antonakopoulos2010approximating}, the Client-Server $s$-spanner \cite{bhattacharyya2012transitive} problem, and the pairwise spanner problem \cite{chlamtavc2020approximating}, by greedily removing low-density partial solutions in an iterative manner. Fortunately, this iterative approach also guarantees a nice \emph{global} approximation that consists of junction trees rooted at different vertices, which is amenable in the online setting.

\paragraph{Reduction to Steiner label cover.}

In Section~\ref{sec:lcp-hr}, we reduce the pairwise spanner problem to the following extension of Steiner problem termed \emph{Steiner label cover}.

\begin{restatable}{definition}{defslc}\label{def:slc}
In the \emph{Steiner label cover} problem, we are given a (directed or undirected) graph $G=(V,E)$, non-negative edge costs $w: E \to \R_{\geq 0}$, and a collection of $k$ disjoint vertex subset pairs $(S_i, T_i)$ for $i \in [k]$ where $S_i, T_i \subseteq V$ and $S_i \cap T_i = \emptyset$. Each pair is associated with a relation (set of permissible pairs) $R_i \subseteq S_i \times T_i$. The goal is to find a subgraph $F=(V,E')$ of $G$, such that 1) for each $i \in [k]$, there exists $(s,t) \in R_i$ such that there is an $s \leadsto t$ path in $F$, and 2) the cost $\sum_{e \in E'}w(e)$ is minimized.
\end{restatable}

For the online Steiner label cover problem, $(S_i, T_i)$ and $R_i$ arrive online, and the goal is to irrevocably select edges to satisfy the first requirement and also approximately minimize the cost.

To reduce to the online Steiner label cover problem, we construct a directed graph $G'$ that consists of disjoint layered graphs from the given graph $G=(V,E)$. Each vertex in $G'$ is labelled by the distance to (from) the root of a junction tree. This allows us to capture distance constraints by a Steiner label cover instance with distance-based relations. From $G'$, we further construct an undirected graph $H$ which is a forest by the \emph{height reduction} technique \cite{chekuri2011set,helvig2001improved}. In $H$, we define the corresponding Steiner label cover instance, where the solution is guaranteed to be a forest. The Steiner label cover instance on the forest $H$ has a nice property. For each tree in $H$, the terminal vertices can be ordered in a way such that if an \emph{interval} belongs to the relation, then any \emph{subinterval} also belongs to the relation.

\begin{restatable}{definition}{defslcord}\label{def:slc-ord}
The \emph{ordered Steiner label cover problem on a forest} is defined as a special case of the Steiner label cover problem (see Definition~\ref{def:slc}) with the following properties.
\begin{enumerate}
    \item $G$ is an undirected graph consisting of disjoint union of trees $H_1, H_2, \dots, H_n$ each of which has a distinguished root vertex $r_j$ where $j \in [n]$.
    \item For each $(S_i,T_i)$ and $R_i$ for $i \in [k]$ where $S_i$ and $T_i$ are disjoint subsets of the leaf vertices of $G$, and each tree $H_j$, the input also includes the total orderings $\prec_{i,j}$ such that:
    \begin{enumerate}
    \item For $S^j_i := S_i \cap V(H_j)$ and $T^j_i := T_i \cap V(H_j)$, the ordering $\prec_{i,j}$ is defined on $S^j_i \cup T^j_i$.
    \item The root $r_j$ separates $S^j_i$ from $T^j_i$.
    \item If $s \in S^j_i$ and $t \in T^j_i$ are such that $(s,t) \in R_i$, then for any $s' \in S^j_i$ and $t' \in T^j_i$ such that $s \preceq_{i,j} s' \prec_{i,j} t' \preceq_{i,j} t$, we have that $(s',t') \in R_i$.
    \end{enumerate}
\end{enumerate}
\end{restatable}

We note that for the online ordered Steiner label cover problem on a forest, besides $(S_i, T_i)$ and $R_i$, the orderings $\{\succ_{i,j}\}_{j \in [n]}$ also arrive online.

We employ a well-defined mapping between junction trees in $G$ and forests in $H$ by paying an $\tO(k^{1/2+\delta})$ factor for competitive online solutions. A crucial step for showing Theorem~\ref{thm:sqrt-k-ps} is the following theorem.

\begin{restatable}{theorem}{thmoslcps} \label{thm:oslc-ps}
    For any constant $\delta > 0$, an $\alpha$-competitive polynomial-time algorithm for online ordered Steiner label cover 
    on a forest implies an $O(\alpha k^{1/2 + \delta})$-competitive polynomial-time algorithm for the online pairwise spanner problem on a directed graph with uniform edge lengths.
\end{restatable}

At a high level, the online pairwise spanner problem on a directed graph $G=(V,E)$ with uniform edge lengths reduces to an instance of online Steiner label cover on the forest $H$ with the following properties.
    \begin{enumerate}
        \item $H$ consists of disjoint trees $H_r$ for each vertex $r \in V$.
        \item $|V(H)| = n^{O(1/\delta)}$, $E(H)=n^{O(1/\delta)}$, and each tree $H_r$ has depth $O(1/\delta)$.
        \item For each arriving terminal pair $(s_i,t_i)$ with distance requirement $d_i$, there is a corresponding pair of terminal sets $(\hat{S}_i,\hat{T}_i)$ and relation $\hat{R}_i$ with  $|\hat{R}_i|=n^{O(1/\delta)}$, where $\hat{S}_i$ and $\hat{T}_i$ are disjoint subsets of vertices in $H$. Furthermore, we can generate total orderings $\prec_{i,r}$ based on the distance-based relations $\hat{R}_i$ such that the Steiner label cover instance is an ordered instance on the forest $H$.
    \end{enumerate}

This technique closely follows the one for solving offline pairwise spanners in \cite{chlamtavc2020approximating}. The intermediary problem considered in \cite{chlamtavc2020approximating} is the \emph{minimum density Steiner label cover problem}. In this framework, the solution is obtained by selecting the partial solution with the lowest density among the junction trees rooted at different vertices and repeat. In the online setting, to capture the global approximation for pairwise spanners, we construct a forest $H$ and consider all the possible roots simultaneously.

The constructions of $G'$, $H$, $\hat{S}_i$, $\hat{T}_i$, $\hat{R}_i$, and the orderings $\prec_{i,r}$, the proofs of Theorem~\ref{thm:oslc-ps} and Theorem~\ref{thm:sqrt-k-ps}, and the reduction from the online pairwise spanner problem on the directed graph $G$ to the online ordered Steiner label cover problem on the forest $H$, are presented in Section~\ref{sec:lcp-hr}. 

\paragraph{An online algorithm for Steiner label cover on $H$.}

In Section~\ref{sec:online-slc}, the goal is to prove the following lemma.
\begin{restatable}{lemma}{lemonlineslc} \label{lem:online-slc}
For the online ordered Steiner label cover problem on a forest (see Definition \ref{def:slc-ord}), there is a randomized polynomial-time algorithm with competitive ratio $\polylog(n)$.
\end{restatable}

We derive an LP formulation for the Steiner label cover instance on $H$. At a high level, the LP minimizes the total edge weight by selecting edges that cover paths with endpoint pairs which belong to the distance-based relation. We show that the LP for Steiner label cover can be converted into an online covering problem, which is efficiently solvable by Theorem~\ref{thm:inf-covering}.

The online rounding is based on the online LP solution for Steiner label cover. We extract the representative vertex sets $\tilde{S}_i$ and $\tilde{T}_i$ from the terminal sets $\hat{S}_i$ and $\hat{T}_i$, respectively, according to orderings $\prec_{i,r}$ and the contribution of the terminal vertex to the objective of the Steiner label cover LP. We show that the union of cross-products over partitions of $\tilde{S}_i$ and $\tilde{T}_i$ (based on the trees in $H$) 
is a subset of the distance-based relation $\hat{R}_i$. This allows us to reduce the online ordered Steiner label cover problem to the \emph{online undirected group Steiner forest} problem defined below.

\begin{restatable}{definition}{defgpst}\label{def:gpst}
In the \emph{online undirected group Steiner forest} problem, we are given an undirected graph $G=(V,E)$ with non-negative edge costs $w: E \to \R_0$, and a collection of $k$ terminal vertex set pairs $(S_i, T_i) \subseteq V \times V$ for $i \in [k]$ that arrives online. The goal is to irrevocably pick edges from $E$ to form a subgraph $F=(V,E')$ of $G$, such that 1) upon the arrival of pair $i \in [k]$, there exists $s \in S_i$ and $t \in T_i$ such that there is an $s \leadsto t$ path in $F$, and 2) the cost $\sum_{e \in E'}w(e)$ is approximately minimized.
\end{restatable}

This technique closely follows the one for solving offline pairwise spanners in \cite{chlamtavc2020approximating}. The main difference is that in the offline pairwise spanner framework, the LP formulation is density-based and considers only one (fractional) junction tree. To globally approximate the online pairwise spanner solution, our LP formulation is based on the forest $H$ and its objective is the total weight of a (fractional) forest.

The LP for the undirected group Steiner forest problem is roughly in the following form.
\begin{equation} \label{opt:inf-dsf} 
\begin{aligned}
& \min_{x} & & \sum_{e \in E(H)}{w'(e) x_e} \\
& \text{subject to}
& & \text{$x$ supports an $\tilde{S}_i$-$\tilde{T}_i$ flow of value 1} & \forall i \in [k],\\
& & & x_e \geq 0 & \forall e \in E(H).\\
\end{aligned}
\end{equation}
Here $w'$ denotes the edge weights in $H$. We show that a solution of the undirected group Steiner forest LP~\eqref{opt:inf-dsf} recovers a solution for the Steiner label cover LP by a factor of 2. The integrality gap of the undirected group Steiner forest LP is $\polylog(n)$ because the instance can be decomposed into single-source group Steiner tree instances by the structure of $H$ \cite{cekp,garg2000polylogarithmic}. This implies that the online rounding for the Steiner label cover LP can be naturally done via solving the undirected group Steiner forest instance online by losing a $\polylog(n)$ factor \cite{cekp}.

\paragraph{Putting it all together.} In Section~\ref{sec:us-sum}, we summarize the overall $\tO(k^{1/2 + \delta})$-competitive algorithm for online pairwise spanners when the given graph has uniform edge lengths. The reduction strategy  is as follows: 
\begin{enumerate}
\item Reduce the online pairwise spanner problem of the original graph $G$ to the online Steiner label cover problem of the directed graph $G'$ which consists of disjoint layered graphs.
\item Reduce the online Steiner label cover problem on $G'$ to an online ordered Steiner label cover problem on $H$, where $H$ is a forest.
\item In the forest $H$, reduce the online ordered Steiner label cover problem to the online undirected group Steiner forest problem.
\end{enumerate}

We note that $G'$ and $H$ are constructed offline. The pairwise spanner in $G$ is $O(\sqrt{k})$-approximated by junction trees according to Lemma~\ref{lem:jt-apx}. 
The graph $G'$ preserves at most twice the cost of the pairwise spanner (junction tree solution) in $G$. The solution of the ordered Steiner label cover problem in graph $H$ is a forest. One can map a forest in $H$ to junction trees in $G'$, via the height reduction technique by losing an $O(k^{\delta})$ factor. Finally, in the forest $H$, we solve the undirected group Steiner forest instance online and recover an ordered Steiner label cover solution by losing a $\polylog(n)$ factor. The overall competitive ratio is therefore $\tO(k^{1/2 + \delta})$.

\subsection{The junction-tree approximation} \label{sec:jt-apx}

We can now proceed to the proof of Lemma~\ref{lem:jt-apx}. Our goal is to approximate a pairwise spanner solution by a collection of junction trees rooted at different vertices.

\lemjtapx*

\begin{proof}
We use a density argument to show a greedy procedure that implies a $O(\sqrt{k})$-approximate junction tree solution. The \emph{density} of a partial solution is defined as follows.

\begin{definition}
Let $J$ denote a junction tree in $G$ and $D(J)$ denote the set of source-sink pairs $(s_i, t_i)$ connected by $J$ such that $d_J(s_i, t_i) \le d_i$, then the density of $J$ is $\rho(J):=|E(J)| / |D(J)|$, i.e., the number of edges used in $J$ divided by the number of terminal pairs connected by $J$ within required distances.
\end{definition}

Intuitively, we are interested in finding low-density junction trees. Let $\opt$ be the number of edges in an optimal pairwise spanner solution. We show that there always exists a junction tree with density at most a $\sqrt{k}$ factor of the optimal density. The proof of Lemma~\ref{lem:sqrt-k-den} closely follows the one for the directed Steiner network problem in \cite{chekuri2011set} by considering whether there is a \emph{heavy} vertex that lies in $s_i \leadsto t_i$ paths for distinct $i$ or there is a simple path with low density. The case analysis also holds when there is a distance constraint $d_i$ for each $(s_i,t_i)$. We provide the proof in Appendix~\ref{pf:lem:sqrt-k-den} for the sake of completeness.

\begin{restatable}{lemma}{lemsqrtkden} \label{lem:sqrt-k-den}
There exists a junction tree $J$ such that $\rho(J) \leq \opt / \sqrt{k}$.
\end{restatable}

Now we are ready to prove Lemma~\ref{lem:jt-apx}. Consider the procedure that finds a minimum density junction tree in each iteration, and continues on the remaining disconnected terminal pairs. Suppose there are $t$ iterations, and after iteration $j \in [t]$, there are $n_j$ disconnected terminal pairs. Let $n_0 = k$ and $n_t = 0$. After each iteration, the minimum number of edges used for connecting the remaining terminal pairs in the remaining graph is at most $\opt$, so the number of edges used by this procedure is upper-bounded by
\[
\sum_{j=1}^t \frac{(n_{j-1} - n_j)\opt}{\sqrt{n_{j-1}}} \leq \sum_{i=1}^k \frac{\opt}{\sqrt{i}} \leq \int_1^{k+1} \frac{\opt}{\sqrt{x}} dx = 2 \opt (\sqrt{k+1} - 1) = O(\sqrt{k}) \opt\]
where the first inequality uses the upper bound by considering the worst case when only one terminal pair is removed in each iteration of the procedure.
\end{proof}

By the proof of Lemma~\ref{lem:sqrt-k-den}, Lemma~\ref{lem:jt-apx} also accounts for  \emph{duplicate} edges from junction trees rooted at different vertices.

\subsection{The Steiner label cover problem and height reduction} \label{sec:lcp-hr}

In this section, we give the construction of the graph $G'$ that consists of disjoint layered graphs, the undirected graph $H$ which is a forest, the terminal set pairs $(\hat{S}_i,\hat{T}_i)$, and the relations $\hat{R}_i$.

We reduce the pairwise spanner problem to the \emph{Steiner label cover} problem. A similar technique is introduced in \cite{chlamtavc2020approximating} by a density argument which is amenable for the offline setting. For our purpose, we instead focus on the global problem that is tractable in the online setting. In what follows we construct a relevant Steiner label cover instance (recall Definition~\ref{def:slc}).

\paragraph{Constructing the layered graph $G'$.} We reduce the pairwise spanner problem to the Steiner label cover problem as follows. Given $G=(V,E)$, for each vertex $r \in V$, we construct $G_r=(V_r, E_r)$. Let
\[V_r:= ((V \setminus \{r\}) \times \bigcup_{j \in [n-1]}\{-j, j\}) \cup \{(r,0)\}\]
and
\[E_r:= \{(u,j) \to (v,j+1) \mid (u,j), (v,j+1) \in V_r, (u,v) \in E\}.\]
For each edge $e \in E_r$, let the edge weight be 1. For each vertex $(u,-j)$ (respectively $(u,j)$) in $V_r$ where $j \in [n-1] \cup \{0\}$, add a vertex $(u^-,-j)$ (respectively $(u^+,j)$), and create an edge $(u^-,-j) \to (u,-j)$ (respectively $(u,j) \to (u^+,j)$) with weight 0. This concludes the construction of $G_r$ (see Figure~\ref{fig:lg} for an illustration).

Let $G'$ be the disjoint union of $G_r$ for $r \in V$. Given a pairwise spanner instance on $G$, we can construct a Steiner label cover instance on $G'$. 
For explicitness, we denote the vertex $(u,j)$ in $V_r$ by $(u,j)_r$. For each $(s_i, t_i)$, let $S_i = \{(s^-_i, -j)_r \mid j > 0, r \in V \setminus \{s_i\}\} \cup (s^-_i,0)_{s_i}$, $T_i = \{(t^+_i, j)_r \mid j > 0, r \in V \setminus \{t_i\}\} \cup (t^+_i,0)_{t_i}$, and $R_i = \{((s^-_i, -j_s)_r,(t^+_i, j_t)_r) \mid r \in V, j_s+j_t \leq d_i\}$ where we recall that $d_i$ is the distance requirement for pair $(s_i,t_i)$. Intuitively, a copy of $s_i$ and a copy of $t_i$ belongs to the relation $R_i$ if 1) they are connected by a junction tree with the same root $r$, and 2) the distance between them is at most $d_i$ in this junction tree.

\begin{figure}[H]
\centering
\begin{subfigure}{.2\textwidth}
\begin{tikzpicture}[scale=0.5]
    \node[fill,circle, inner sep=0pt, minimum size=0.2cm] (a) at (-4,2) {};
    \node[fill,circle, inner sep=0pt, minimum size=0.2cm] (c) at (-4,-2) {};
    \node[fill,circle, inner sep=0pt, minimum size=0.2cm] (b) at (-2,2) {};
    \node[fill,circle, inner sep=0pt, minimum size=0.2cm] (d) at (-2,-2) {};
    \node[fill,circle, inner sep=0pt, minimum size=0.2cm] (r) at (0,0) {};
    \node at (-4,2.7) {$a$};
    \node at (-4,-2.7) {$c$};
    \node at (-2,2.7) {$b$};
    \node at (-2,-2.7) {$d$};
    \node at (0,0.7) {$r$};
    \path[->]
        (a) edge (b)
        (b) edge (c)
        (b) edge (r)
        (c) edge (a)
        (c) edge (d)
        (d) edge (b)
        (d) edge (r)
        (r) edge (a)
        (r) edge (c);
\end{tikzpicture}
\subcaption{$G$}
\end{subfigure}
\begin{subfigure}{.7\textwidth}
\begin{tikzpicture}[scale=0.65]
    \node at (-9,1.5) {$a$};
    \node at (-9,0.5) {$b$};
    \node at (-9,-0.5) {$c$};
    \node at (-9,-1.5) {$d$};
    \node at (0,-1) {$(r,0)$};
    \node[fill,circle, inner sep=0pt, minimum size=0.2cm] (r0) at (0,0) {};
    \foreach \x in {1,...,4}
    \foreach \y in {1,...,4}{
        \node[fill,circle, inner sep=0pt, minimum size=0.2cm] (l\x\y) at (-2*\x,-2.5 + \y) {};
        \node[fill,circle, inner sep=0pt, minimum size=0.2cm] (r\x\y) at (2*\x,-2.5 + \y) {};
    }
    \foreach \x in {1,...,4}{
        \node[draw,circle, inner sep=0pt, minimum size=0.2cm] (lo\x4) at (-2*\x-1, 3) {};
        \node[draw,circle, inner sep=0pt, minimum size=0.2cm] (lo\x3) at (-2*\x-1, 2.3) {};
        \node[draw,circle, inner sep=0pt, minimum size=0.2cm] (lo\x2) at (-2*\x-1, -2.3) {};
        \node[draw,circle, inner sep=0pt, minimum size=0.2cm] (lo\x1) at (-2*\x-1, -3) {};
        \path[->]
            (lo\x4) edge[dashed] (l\x4)
            (lo\x3) edge[dashed] (l\x3)
            (lo\x2) edge[dashed] (l\x2)
            (lo\x1) edge[dashed] (l\x1);
        \node[draw,circle, inner sep=0pt, minimum size=0.2cm] (ro\x4) at (2*\x+1, 3) {};
        \node[draw,circle, inner sep=0pt, minimum size=0.2cm] (ro\x3) at (2*\x+1, 2.3) {};
        \node[draw,circle, inner sep=0pt, minimum size=0.2cm] (ro\x2) at (2*\x+1, -2.3) {};
        \node[draw,circle, inner sep=0pt, minimum size=0.2cm] (ro\x1) at (2*\x+1, -3) {};
        \path[->]
            (r\x4) edge[dashed] (ro\x4)
            (r\x3) edge[dashed] (ro\x3)
            (r\x2) edge[dashed] (ro\x2)
            (r\x1) edge[dashed] (ro\x1);
    }
    \foreach \x in {1,...,3}{
        \pgfmathtruncatemacro{\xnext}{\x+1}
        \draw[->] (l\xnext1) -- (l\x3);
        \draw[->] (l\xnext2) -- (l\x1);
        \draw[->] (l\xnext3) -- (l\x2);
        \draw[->] (l\xnext2) -- (l\x4);
        \draw[->] (l\xnext4) -- (l\x3);
        \draw[->] (r\x1) -- (r\xnext3);
        \draw[->] (r\x2) -- (r\xnext1);
        \draw[->] (r\x3) -- (r\xnext2);
        \draw[->] (r\x2) -- (r\xnext4);
        \draw[->] (r\x4) -- (r\xnext3);
    }
    \draw[->] (l11) -- (r0);
    \draw[->] (l13) -- (r0);
    \draw[->] (r0) -- (r12);
    \draw[->] (r0) -- (r14);
    \node[draw,circle, inner sep=0pt, minimum size=0.2cm] (lr0) at (-0.5, 2) {};
    \node[draw,circle, inner sep=0pt, minimum size=0.2cm] (rr0) at (0.5, 2) {};
    \path[->]
        (lr0) edge[dashed] (r0)
        (r0) edge[dashed] (rr0);
\end{tikzpicture}
\subcaption{$G_r$}
\end{subfigure}
\caption{Construction of $G_r$ given $G=(V,E)$ and $r \in V$}
\label{fig:lg}
\end{figure}

We can recover a pairwise spanner solution in $G$ with at most the cost of the corresponding Steiner label cover solution in $G'$. Suppose we have the Steiner label cover solution in $G'$, for each selected edge $(u,j) \to (v,j+1)$ in $G_r$, we select $u \to v$ for the pairwise spanner solution in $G$.

\begin{claim} \label{cl:layer-ps}
There exists a junction tree solution for the pairwise spanner problem in $G$ with at most the cost of a Steiner label cover solution in $G'$.
\end{claim}

We observe that any junction tree solution as a pairwise spanner in $G$ corresponds to a Steiner label cover solution in $G'$ and vice versa (see Figure \ref{fig:lg-ex} for an illustration). Let $\opt_{junc}$ be the optimal value of a junction tree solution that counts duplicate edges from junction trees rooted at different vertices for the pairwise spanner problem in $G$. Then the optimum of the Steiner label cover solution in $G'$ is at most $2\opt_{junc}$. In each junction tree, each edge is used at most once in the in-arborescence and at most once in the out-arborescence.

\begin{claim} \label{cl:layer-ps-opt}
The optimal value of the Steiner label cover problem in $G'$ is at most $2\opt_{junc}$.
\end{claim}

\begin{figure}[H]
\centering
\begin{subfigure}{.2\textwidth}
\begin{tikzpicture}[scale=0.5]
    \node[fill,circle, inner sep=0pt, minimum size=0.2cm] (a) at (-4,2) {};
    \node[fill,circle, inner sep=0pt, minimum size=0.2cm] (c) at (-4,-2) {};
    \node[fill,circle, inner sep=0pt, minimum size=0.2cm] (b) at (-2,2) {};
    \node[fill,circle, inner sep=0pt, minimum size=0.2cm] (d) at (-2,-2) {};
    \node[fill,circle, inner sep=0pt, minimum size=0.2cm] (r) at (0,0) {};
    \node at (-4,2.7) {$a$};
    \node at (-4,-2.7) {$c$};
    \node at (-2,2.7) {$b$};
    \node at (-2,-2.7) {$d$};
    \node at (0,0.7) {$r$};
    \path[->]
        (a) edge [draw=red, line width=1.5pt] (b)
        (b) edge (c)
        (b) edge [draw=red, line width=1.5pt] (r)
        (c) edge (a)
        (c) edge (d)
        (d) edge (b)
        (d) edge [draw=red, line width=1.5pt] (r)
        (r) edge [draw=red, line width=1.5pt] (a)
        (r) edge [draw=red, line width=1.5pt] (c);
\end{tikzpicture}
\subcaption{$G$}
\end{subfigure}
\begin{subfigure}{.7\textwidth}
\begin{tikzpicture}[scale=0.65]
    \node at (-9,1.5) {$a$};
    \node at (-9,0.5) {$b$};
    \node at (-9,-0.5) {$c$};
    \node at (-9,-1.5) {$d$};
    \node at (0,-1) {$(r,0)$};
    \node[fill,circle, inner sep=0pt, minimum size=0.2cm] (r0) at (0,0) {};
    \foreach \x in {1,...,4}
    \foreach \y in {1,...,4}{
        \node[fill,circle, inner sep=0pt, minimum size=0.2cm] (l\x\y) at (-2*\x,-2.5 + \y) {};
        \node[fill,circle, inner sep=0pt, minimum size=0.2cm] (r\x\y) at (2*\x,-2.5 + \y) {};
    }
    \foreach \x in {1,...,4}{
        \node[draw,circle, inner sep=0pt, minimum size=0.2cm] (lo\x4) at (-2*\x-1, 3) {};
        \node[draw,circle, inner sep=0pt, minimum size=0.2cm] (lo\x3) at (-2*\x-1, 2.3) {};
        \node[draw,circle, inner sep=0pt, minimum size=0.2cm] (lo\x2) at (-2*\x-1, -2.3) {};
        \node[draw,circle, inner sep=0pt, minimum size=0.2cm] (lo\x1) at (-2*\x-1, -3) {};
        \path[->]
            (lo\x4) edge[dashed] (l\x4)
            (lo\x3) edge[dashed] (l\x3)
            (lo\x2) edge[dashed] (l\x2)
            (lo\x1) edge[dashed] (l\x1);
        \node[draw,circle, inner sep=0pt, minimum size=0.2cm] (ro\x4) at (2*\x+1, 3) {};
        \node[draw,circle, inner sep=0pt, minimum size=0.2cm] (ro\x3) at (2*\x+1, 2.3) {};
        \node[draw,circle, inner sep=0pt, minimum size=0.2cm] (ro\x2) at (2*\x+1, -2.3) {};
        \node[draw,circle, inner sep=0pt, minimum size=0.2cm] (ro\x1) at (2*\x+1, -3) {};
        \path[->]
            (r\x4) edge[dashed] (ro\x4)
            (r\x3) edge[dashed] (ro\x3)
            (r\x2) edge[dashed] (ro\x2)
            (r\x1) edge[dashed] (ro\x1);
    }
    \foreach \x in {1,...,3}{
        \pgfmathtruncatemacro{\xnext}{\x+1}
        \draw[->] (l\xnext1) -- (l\x3);
        \draw[->] (l\xnext2) -- (l\x1);
        \draw[->] (l\xnext3) -- (l\x2);
        \draw[->] (l\xnext2) -- (l\x4);
        \draw[->] (l\xnext4) -- (l\x3);
        \draw[->] (r\x1) -- (r\xnext3);
        \draw[->] (r\x2) -- (r\xnext1);
        \draw[->] (r\x3) -- (r\xnext2);
        \draw[->] (r\x2) -- (r\xnext4);
        \draw[->] (r\x4) -- (r\xnext3);
    }
    \draw[->] (l11) -- (r0);
    \draw[->] (l13) -- (r0);
    \draw[->] (r0) -- (r12);
    \draw[->] (r0) -- (r14);
    \node[draw,circle, inner sep=0pt, minimum size=0.2cm] (lr0) at (-0.5, 2) {};
    \node[draw,circle, inner sep=0pt, minimum size=0.2cm] (rr0) at (0.5, 2) {};
    \path[->]
        (lr0) edge[dashed] (r0)
        (r0) edge[dashed] (rr0);
     \path[->]
        (lo24) edge [dashed, draw=red, line width=1.5pt] (l24)
        (l24) edge [draw=red, line width=1.5pt] (l13)
        (l13) edge [draw=red, line width=1.5pt] (r0)
        (r0) edge [draw=red, line width=1.5pt] (r12)(r12) edge [dashed, draw=red, line width=1.5pt] (ro12)
        (lo11) edge [dashed, draw=red, line width=1.5pt] (l11)
        (l11) edge [draw=red, line width=1.5pt] (r0)
        (r0) edge [draw=red, line width=1.5pt] (r14)
        (r14) edge [draw=red, line width=1.5pt] (r23)
        (r23) edge [dashed, draw=red, line width=1.5pt] (ro23);
\end{tikzpicture}
\subcaption{$G_r$}
\end{subfigure}
\caption{Suppose $(s_1,t_1) = (a,c)$, $d_1 = 3$, $(s_2,t_2) = (d,b)$, and $d_2 = 3$. We can find a junction tree rooted at $r$ consisting of the red thick edges in $G$ and $G_r$ that connects the two terminal pairs and satisfies the distance requirements. Note that edge $a \to b$ is used in the in-arborescence and out-arborescence in $G$.}
\label{fig:lg-ex}
\end{figure}

\paragraph{Constructing the Steiner label cover instance on the forest $H$.} To find an approximate solution for the Steiner label cover instance, we further process the graph $G'$ by the height reduction technique \cite{chekuri2011set,helvig2001improved}. The height reduction technique allows us to generate the forest $H$. The structure of $H$ is later useful in Section~\ref{sec:online-slc} for deriving an LP formulation for the Steiner label cover problem.

\begin{lemma} \label{lem:hr} (\cite{chlamtavc2020approximating})
Let $G=(V,E)$ be a directed graph where each edge is associated with a non-negative weight $w: E \to \R_{\geq 0}$. Let $r \in V$ be a root and $\sigma > 0$ be some parameter. Then we can efficiently construct an undirected tree $T_r$ rooted at $r'$ of height $\sigma$ and size $|V|^{O(\sigma)}$ together with edge weights $w': E(T_r) \to \R_{\geq 0}$ and a vertex mapping $\Psi: V(T_r) \to V$, such that
\begin{enumerate}
    \item For any in-arborescence (out-arborescence) $T$ in $G$ rooted at $r$, there exists a subtree $T'$ in $T_r$ rooted at $r'$ such that letting $L(T)$ and $L(T')$ be the set of leaves of $T$ and $T'$, respectively, we have $\Psi(L(T')) = L(T)$. Moreover, 
    \[\sum_{e \in T} w(e) \leq O(\sigma |L(T)|^{1 / \sigma}) \sum_{e \in T'} w'(e).\]
    \item Given any subtree $T'$ in $T_r$ rooted at $r'$, we can efficiently find an in-arborescence (out-arborescence) $T$ in $G$ rooted at $r$, such that $\Psi(L(T')) = L(T)$ and $\sum_{e \in T'} w'(e) \leq \sum_{e \in T} w(e)$.
\end{enumerate}
\end{lemma}

For our purpose, for each $r \in V$, we construct a tree $H_r$ from $G_r$ by the height reduction technique. We recall that $G_r$ is a layered graph with a center vertex $(r,0)$. Let $G^-_r$ denote the subgraph of $G_r$ induced by the vertex set
\[\tuple{\bigcup_{v \in V \setminus \{r\}}\{v, v^-\} \times \{-j \mid j \in [n-1]\}} \cup \{(r,0), (r^-,0)\},\]
and similarly, let $G^+_r$ denote the subgraph of $G_r$ induced by the vertex set
\[\tuple{\bigcup_{v \in V \setminus \{r\}}\{v, v^+\} \times [n-1]} \cup \{(r,0), (r^+,0)\}.\]
By Lemma~\ref{lem:hr}, with $(r,0)$ being the root, we can construct a tree $T_r^-$ rooted at $r^-$ for $G_r^-$ which approximately preserves the cost of any in-arborescence rooted at $(r,0)$ in $G_r^-$ by a subtree in $T_r^-$, and similarly a tree $T_r^+$ rooted at $r^+$ for $G_r^+$ which approximately preserves the cost of any out-arborescence rooted at $(r,0)$ in $G_r^+$ by a subtree in $T_r^+$. We further add a super root $r'$ and edges $\{r',r^-\}$ and $\{r',r^+\}$ both with weight 0. This concludes the construction of the weighted tree $H_r$ (see Figure~\ref{fig:hr} for an illustration).

\begin{figure}[h]
\centering
\begin{tikzpicture}
    \draw (-2,0) -- (2,0) -- (0,2) --cycle;
    \node[draw=none] at (0,0.5) {$T_r^-$};
    \draw (4,0) -- (8,0) -- (6,2) --cycle;
    \node[draw=none] at (6,0.5) {$T_r^+$};
    \node at (0,2.5) {$r^-$};
    \node at (6,2.5) {$r^+$};
    \node at (3,3) {$r'$};
    \node[fill,circle, inner sep=0pt, minimum size=0.2cm] (t1) at (0,2) {};
    \node[fill,circle, inner sep=0pt, minimum size=0.2cm] (t2) at (6,2) {};
    \node[fill,circle, inner sep=0pt, minimum size=0.2cm] (t3) at (3,2.5) {};
    \path[-]
        (t1) edge (t3)
        (t2) edge (t3);
\end{tikzpicture}
\caption{Construction of $H_r$}
\label{fig:hr}
\end{figure}

Let $H$ be the disjoint union of $H_r$ for $r \in V$. To prove Theorem~\ref{thm:sqrt-k-ps}, we show that we can achieve an $\tO(k^\delta)$-approximation for the Steiner label cover problem on graph $H$ in an online manner. We set $\sigma=\lceil 1/\delta \rceil$ and apply Lemma~\ref{lem:hr} to obtain the weight $w':E(H) \to \R_{\geq 0}$ and the mapping $\Psi_r: V(H_r) \to V(G_r)$ for all $r \in V$. Let $\Psi$ be the union of the mappings for all $r \in V$. For each pair $(s_i,t_i)$ in the original graph $G$, we recall that we focus on the vertex subset pair $(S_i,T_i)$ of $G'$ and its relation $R_i$ which captures the distance requirement.

To establish the correspondence between the Steiner label cover instances, we clarify the mapping between the leaves of the arborescences in $G'$ and $H$. Given a vertex $(s^-_i,-j)_r$ in $V(G')$ with a non-negative $j$, let $\Psi^{-1}((s^-_i,-j)_r)$ denote the set of leaves in $T^-_r$ that maps to $(s^-_i,-j)_r$ by $\Psi$. Similarly, given a vertex $(t^+_i,j)_r$ in $V(G')$ with a non-negative $j$, let $\Psi^{-1}((t^+_i,j)_r)$ denote the set of leaves in $T^+_r$ that maps to $(t^+_i,j)_r$ by $\Psi$. The mapping $\Psi^{-1}$ naturally defines the terminal sets of interest $\hat{S_i}:=\Psi^{-1}(S_i)=\{\hat{s} \in V(H) \mid \Psi(\hat{s}) \in S_i\}$ and $\hat{T_i}:=\Psi^{-1}(T_i)=\{\hat{t} \in V(H) \mid \Psi(\hat{t}) \in T_i\}$. The relation is also naturally defined: $\hat{R}_i:= \{(\hat{s},\hat{t}) \in \hat{S}_i \times \hat{T}_i \mid (\Psi(\hat{s}),\Psi(\hat{t})) \in R_i\}$. We note that for $\hat{s} \in \hat{S}_i$ and $\hat{t} \in \hat{T}_i$, $(\hat{s},\hat{t})$ belongs to $\hat{R}_i$ only when $\hat{s}$ and $\hat{t}$ belong to the same tree $H_r$ (see Figure \ref{fig:sc} for an illustration).

\begin{figure}[h]
\centering
\begin{tikzpicture}
    \draw (-2,0) -- (2,0) -- (0,2) --cycle;
    \node[draw=none] at (0,0.5) {$T_r^-$};
    \draw (4,0) -- (8,0) -- (6,2) --cycle;
    \node[draw=none] at (6,0.5) {$T_r^+$};
    \node at (0,2.5) {$r^-$};
    \node at (6,2.5) {$r^+$};
    \node at (3,3) {$r'$};
    \node[fill,circle, inner sep=0pt, minimum size=0.2cm] (t1) at (0,2) {};
    \node[fill,circle, inner sep=0pt, minimum size=0.2cm] (t2) at (6,2) {};
    \node[fill,circle, inner sep=0pt, minimum size=0.2cm] (t3) at (3,2.5) {};
    \path[-]
        (t1) edge (t3)
        (t2) edge (t3);
    \node at (-1,0.3) {$\hat{S}^r_i$};
    \node at (7,0.3) {$\hat{T}^r_i$};
    \draw (-1.6,0) -- (-1.6,0.2) -- (-1.2,0.2);
    \draw (-0.8,0.2) -- (-0.5,0.2) -- (-0.5,0);
    \draw (8-1.6,0) -- (8-1.6,0.2) -- (8-1.2,0.2);
    \draw (8-0.8,0.2) -- (8-0.5,0.2) -- (8-0.5,0);
\end{tikzpicture}
\caption{From pairwise spanner to Steiner label cover. In $G$, consider connecting $s_i$ and $t_i$ by using a junction tree rooted at $r$. In $H_r$, we consider terminal sets $\hat{S}^r_i := \hat{S}_i \cap H_r$ and $\hat{T}^r_i := \hat{T}_i \cap H_r$. A junction tree in $G$ rooted at $r$ corresponds to a subtree with leaves that belongs to $\hat{S}^r_i \cup \hat{T}^r_i$. The set $\hat{R}^r_i := \hat{R}_i \cap (\hat{S}^r_i \times \hat{T}^r_i)$ defines the admissible pairs, which maintains the distance requirement $d_i$ in the layered graph $G_r$.}
\label{fig:sc}
\end{figure}

\paragraph{Proving Theorem~\ref{thm:oslc-ps} and Theorem~\ref{thm:sqrt-k-ps}.}

We first show that the online Steiner label cover problem on $H$ is indeed an ordered one. If we can obtain an $\alpha$-competitive solution, then we can recover an $O(\alpha k^{1/2 + \delta})$-competitive solution for the online pairwise spanner problem on $G$.

Having $H$, $(\hat{S}_i,\hat{T}_i)$, and $\hat{R}_i$ as aforementioned, we show how the orderings $\succ_{i,r}$ are generated so that we have the complete input for the ordered Steiner label cover problem. The ordering construction technique closely follows the one in \cite{chlamtavc2020approximating}. The difference is that we generate an ordering for each tree for global approximation purpose while \cite{chlamtavc2020approximating} focuses on one tree for greedy purpose.

In round $i$, let $\hat{S}^r_i := \hat{S}_i \cap H_r$ be the source leaves in $V(T^-_r)$ and $\hat{T}^r_i := \hat{T}_i \cap H_r$ be the sink leaves in $V(T^+_r)$. We have that $\hat{S}_i = \bigcup_{r \in V} \hat{S}^r_i$ and $\hat{T}_i = \bigcup_{r \in V} \hat{T}^r_i$ because $H_r$ are disjoint for different $r \in V$ so $\{\hat{S}^r_i \mid r \in V\}$ forms a partition of $\hat{S}_i$ and a similar argument holds for $\hat{T}_i$. For the definition of the relation, let $\hat{R}^r_i:=\hat{R}_i \cap (\hat{S}^r_i \times \hat{T}^r_i)$. Intuitively, we partition $\hat{S}_i$, $\hat{T}_i$, and $\hat{R}_i$ into $\hat{S}^r_i$, $\hat{T}^r_i$, and $\hat{R}^r_i$ based on the vertex $r \in V$, respectively.

We define the total ordering $\prec_{i,r}$ on $\hat{S}^r_i \cup \hat{T}^r_i$ as follows.

\begin{itemize}
    \item For any $\hat{s} \in \hat{S}^r_i$ and any $\hat{t} \in \hat{T}^r_i$, we have that $\hat{s} \prec_{i,r} \hat{t}$.
    \item For any $\hat{s},\hat{s}' \in \hat{S}^r_i$ where $\Psi(\hat{s}) = (s^-_i, -j)_r$ and $\Psi(\hat{s}') = (s^-_i, -j')_r$, we have that $\hat{s} \prec_{i,r} \hat{s}'$ if and only if $j > j'$, i.e., under the mapping $\Psi$, the one that is farther from $(r,0)$ in $G'$ has a lower rank according to $\prec_{i,r}$. Ties are broken arbitrarily.
    \item For any $\hat{t}',\hat{t} \in \hat{T}^r_i$ where $\Psi(\hat{t}') = (t^+_i, j')_r$ and $\Psi(\hat{t}) = (t^+_i, j)_r$, we have that $\hat{t'} \prec_{i,r} \hat{t}$ if and only if $j' < j$, i.e., under the mapping $\Psi$, the one that is farther from $(r,0)$ in $G'$ has a higher rank according to $\prec_{i,r}$. Ties are broken arbitrarily.
\end{itemize}

For ease of notion, we use $\preceq_{i,r}$ to denote inclusive lower order, i.e., when $\hat{s} \preceq_{i,r} \hat{s}'$, it is either $\hat{s} = \hat{s}'$ or $\hat{s} \prec_{i,r} \hat{s}'$. If the pair $(\hat{s},\hat{t}) \in \hat{R}^r_i$ corresponds to an interval, then any subinterval pair $(\hat{s}',\hat{t}')$ where $\hat{s} \preceq_{i,r} \hat{s}'$ and $\hat{t'} \preceq_{i,r} \hat{t}$ also belongs to the relation $\hat{R}^r_i$. This is by the construction of the distance-based relations $R_i$ and $\hat{R}_i$ which implies that the Steiner label cover problem is ordered by $\prec_{i,r}$.

\begin{claim} \label{cl:ord}
If $(\hat{s},\hat{t}) \in \hat{R}^r_i$, then for any $\hat{s}' \in \hat{S}^r_i$ and $\hat{t}' \in \hat{T}^r_i$ such that $\hat{s} \preceq_{i,r} \hat{s}' \prec_{i,r} \hat{t}' \preceq_{i,r} \hat{t}$, we have that $(\hat{s}',\hat{t}') \in \hat{R}^r_i$.
\end{claim}

Now we are ready to show Theorem~\ref{thm:oslc-ps}.

\thmoslcps*

\begin{proof}

Suppose we have an ordered Steiner label cover instance on $H$, $(\hat{S}_i,\hat{T}_i)$, $\hat{R}_i$, and $\succ_{i,r}$ as aforementioned and we have an $\alpha$-competitive solution. From the structure of $H$, the online solution is guaranteed to be a forest. By Lemma~\ref{lem:hr}, we can recover an online Steiner label cover solution in $G'$ from that in $H$, by losing a factor of $O(k^{\delta})$. These imply that we have an $O(\alpha k^{\delta})$-competitive solution for the online Steiner label cover problem in $G'$.

We recall that by Claim~\ref{cl:layer-ps-opt}, the Steiner label cover solution in $G'$ pays at most twice of the optimal junction tree solution for pairwise spanners in $G$ with value $\opt_{junc}$ (which also counts the duplicate edges of junction trees rooted at different vertices). By Lemma~\ref{lem:jt-apx}, we have that $\opt_{junc}$ is $O(\sqrt{k}) \opt$ where $\opt$ is the optimal value of the pairwise spanner in $G$.

By Claim~\ref{cl:layer-ps}, we can recover a pairwise spanner solution in $G$ from an online Steiner label cover solution in $G'$. Putting the results together, we obtain an online Steiner label cover solution in $G'$ with cost at most $O(\alpha k^{\delta}) \opt_{junc}$ which recovers an online pairwise spanner solution with cost at most $O(\alpha k^{1/2+\delta}) \opt$.
\end{proof}

The remaining is to show that there exists a randomized polylogarithmically competitive online algorithm for the ordered Steiner label cover problem on the forest $H$, with $(\hat{S}_i,\hat{T}_i)$ and $\hat{R}_i$ arriving online. We show the following lemma in Section \ref{sec:online-slc}.

\lemonlineslc*

Combining Lemma~\ref{lem:online-slc} and Theorem~\ref{thm:oslc-ps} implies Theorem~\ref{thm:sqrt-k-ps}.

\subsection{Online ordered Steiner label cover on a forest} \label{sec:online-slc}

This section is devoted to proving Lemma~\ref{lem:online-slc}.

We start with the high level sketch of the proof. In the original work for the offline setting \cite{chlamtavc2020approximating}, an LP-based approach is used for searching a minimum density junction tree, which picks the best solution among the roots $r \in V$ and repeats. For the online setting, we use an LP that captures the global approximate solution. The constraints of the LP involve the terminal pairs $(\hat{S}_i,\hat{T}_i)$, which is a non-cross-product relation that requires a more meticulous rounding scheme. Fortunately, we are able to use a similar technique to \cite{chlamtavc2020approximating} that generates the terminal vertex orderings, extracts representatives of the sets of terminals according to the orderings, and reduces the ordered Steiner label cover problem to the undirected group Steiner forest problem. We employ the online buy-at-bulk framework \cite{cekp} to solve the final undirected group Steiner forest instance on the forest with specific structure.

\subsubsection{An LP-based approach for Steiner label cover}
We recall that $(\hat{S_i}, \hat{T_i})$ is the terminal pair that consists of leaves in the forest $H$, where $H$ is the disjoint union of trees $H_r$ for $r \in V$. 
Given $(\hat{s},\hat{t}) \in \hat{R}_i$, let $r \in V$ be such that $\hat{s}$ and $\hat{t}$ are leaves of $H_r$. Let $r(\hat{s})=r(\hat{t})$ denote the root $r'$ of this tree $H_r$. 
The goal is to formulate an LP that fractionally picks the edges to pack the $\hat{s}$-$\hat{t}$ paths where $(\hat{s},\hat{t})$ belongs to relation $\hat{R}_i$. We recall that $H$ has $\sigma=\lceil 1/\delta \rceil$ layers, so $E(H)=n^{O(1/\delta)}$ and $|\hat{R}_i|=n^{O(1/\delta)}$.
We use a natural LP relaxation for the problem described in Lemma~\ref{lem:online-slc}:

\begin{equation} \label{opt:slc} 
\begin{aligned}
& \min_{x,y,z} & & \sum_{e \in E(H)}{w'(e) x_e} \\
& \text{subject to}
& & \sum_{(\hat{s},\hat{t}) \in \hat{R}_i}y_{\hat{s},\hat{t}} \geq 1 & \forall i \in [k],\\
& & & \sum_{\hat{t} \mid (\hat{s},\hat{t}) \in \hat{R}_i} y_{\hat{s},\hat{t}} \leq z_{\hat{s}} & \forall i \in [k], \hat{s} \in \hat{S}_i,\\
& & & \sum_{\hat{s} \mid (\hat{s},\hat{t}) \in \hat{R}_i} y_{\hat{s},\hat{t}} \leq z_{\hat{t}} & \forall i \in [k], \hat{t} \in \hat{T}_i,\\
& & & \text{$x$ supports an $\hat{s}$-$r(\hat{s})$ flow of value $z_{\hat{s}}$} & \forall i \in [k], \hat{s} \in \hat{S}_i,\\
& & & \text{$x$ supports an $\hat{t}$-$r(\hat{t})$ flow of value $z_{\hat{t}}$} & \forall i \in [k], \hat{t} \in \hat{T}_i,\\
& & & x_e \geq 0 & \forall e \in E(H),\\
& & & y_{\hat{s},\hat{t}} \geq 0 & \forall i \in [k], (\hat{s},\hat{t}) \in \hat{R}_i.
\end{aligned}
\end{equation}
This LP is a relaxation by considering an integral solution $x,y,z$ where $x_e, y_{\hat{s},\hat{t}}, z_{\hat{s}}, z_{\hat{t}} \in \{0,1\}$. $x_e$ is an indicator of edge $e$, $y_{\hat{s},\hat{t}}$ is an indicator of the $\hat{s}$-$\hat{t}$ path, while $z_{\hat{s}}$ and $z_{\hat{t}}$ denote the $\hat{s}$-$r(\hat{s})$ and $\hat{t}$-$r(\hat{t})$ \emph{flow} value, respectively. Suppose we have an integral solution for Steiner label cover, then for each $(\hat{S}_i,\hat{T}_i)$, there must exists a $y_{\hat{s},\hat{t}}=1$ where $(\hat{s},\hat{t}) \in \hat{R}_i$, which also indicates that all edges $e$ of the $\hat{s}$-$\hat{t}$ path must satisfy $x_e=1$ and there is an $\hat{s}$-$\hat{t}$ flow with value 1. The first constraint ensures that for each pair $(\hat{S}_i,\hat{T}_i)$, there must be an $\hat{s}$-$\hat{t}$ path that belongs to the relation $\hat{R}_i$,  i.e., $\hat{s} \in \hat{S}_i$, $\hat{t} \in \hat{T}_i$, and $(\hat{s},\hat{t}) \in \hat{R}_i$. In the second constraint, for each $i \in [k]$, $z_{\hat{s}}$ denotes an upper bound for the total number of paths that have $\hat{s}$ as the source, where each such $\hat{s}$-$\hat{t}$ path satisfies $(\hat{s},\hat{t}) \in \hat{R}_i$. In the fourth constraint, the path upper bound $z_{\hat{s}}$ is subject to the capacity $x$, i.e., the edge indicator $x$ is naturally an upper bound that packs the $\hat{s}$-$\hat{t}$ paths which belong to the relation. The third and the fifth constraints are defined similarly.\footnote{We thank the anonymous reviewer for suggesting to us the following simpler covering LP formulation similar to the one in \cite{shen2020online}:
\begin{align*}
  \begin{aligned}
    \text{minimize } & \sum_{e \in E(H)}{w'(e) x_e}
    \text{ over } x_e \ge 0 & \forall e \in E(H) 
    \text{ s.t.\ } \sum_{(\hat{s}, \hat{t}) \in \hat{R}_i}\min\{MC(\hat{s},r(\hat{s})),MC(\hat{t},r(\hat{t}))\} \geq 1 \quad \forall i \in [k]
  \end{aligned}
\end{align*}
where $MC(u,v)$ denotes the minimum $u$-$v$ cut value. This is a covering LP because one can consider all possible combinations by picking either $\hat{s}$ or $\hat{t}$ for each $(\hat{s},\hat{t}) \in \hat{R}_i$ as the constraints. A subroutine that computes the sum of the lowest min-cut combination can be used as a separation oracle. In LP~\eqref{opt:slc}, $y_{\hat{s},\hat{t}}$ was later used to extract representative terminal vertices for the undirected group Steiner forest problem. One can also use the above LP to get $x$ then use LP \eqref{opt:int-slc} to get $y$ to extract the representative terminal vertices.
}

In round $i$, $(\hat{S}_i,\hat{T}_i)$ and $\hat{R}_i$ arrive. The goal is to update $x$ in a non-decreasing manner so that the constraints are satisfied with some underlying variables $y_{\hat{s},\hat{t}}$ (which is initially set to 0), and $z_{\hat{s}}$ and $z_{\hat{t}}$ where $(\hat{s},\hat{t}) \in \hat{R}_i$, such that the objective is approximately optimal.

We convert LP~\eqref{opt:slc} into a covering LP. Suppose in round $i$, we are given the edge capacity $x$. 
Let $P(\hat{s})$ denote the set of edges in the $\hat{s}$-$r(\hat{s})$ path and $P(\hat{t})$ denote the set of edges in the $\hat{t}$-$r(\hat{t})$ path.
By setting the first constraint as the objective, merging the second and fourth constraint, and merging the third and fifth constraint, we derive the following internal offline LP, 
\begin{equation} \label{opt:int-slc}
\begin{aligned}
& \max_{y} & & \sum_{(\hat{s},\hat{t}) \in \hat{R}_i}{y_{\hat{s},\hat{t}}} \\
& \text{subject to}
& & \sum_{\hat{t} \mid (\hat{s},\hat{t}) \in \hat{R}_i} y_{\hat{s},\hat{t}} \leq x_e & \forall \hat{s} \in \hat{S}_i, \forall e \in P(\hat{s}),\\
& & & \sum_{\hat{s} \mid (\hat{s},\hat{t}) \in \hat{R}_i} y_{\hat{s},\hat{t}} \leq x_e & \forall \hat{t} \in \hat{T}_i, \forall e \in P(\hat{t}),\\
& & & y_{\hat{s},\hat{t}} \geq 0 & \forall (\hat{s},\hat{t}) \in \hat{R}_i,
\end{aligned}
\end{equation}
and its dual
\begin{equation} \label{opt:int-slc-dual} 
\begin{aligned}
& \min_{\alpha} & & \sum_{e \in E(H)}{x_e \tuple{\sum_{\hat{s} \mid e \in P(\hat{s})}\alpha_{\hat{s},e} + \sum_{\hat{t} \mid e \in P(\hat{t})}\alpha_{\hat{t},e} }} \\
& \text{subject to}
& & \sum_{e \in P(\hat{s})}{\alpha_{\hat{s},e}} + \sum_{e \in P(\hat{t})}{\alpha_{\hat{t},e}} \geq 1 & \forall (\hat{s},\hat{t}) \in R_i,\\
& & & \alpha_{\hat{s},e} \geq 0 & \forall \hat{s} \in \hat{S}_i, \forall e \in P(\hat{s}),\\
& & & \alpha_{\hat{t},e} \geq 0 & \forall \hat{t} \in \hat{T}_i, \forall e \in P(\hat{t}).
\end{aligned}
\end{equation}
We can solve LP~\eqref{opt:int-slc} and LP~\eqref{opt:int-slc-dual} directly because there are only polynomially many constraints. 
If the objective value is at least 1, then $x$ is \emph{good}, i.e., there exists an $\hat{s} \leadsto \hat{t}$ path. 
Otherwise, $x$ is \emph{bad}.

To solve LP~\eqref{opt:slc}, we check if $x$ is good or bad by solving LP~\eqref{opt:int-slc-dual}. If $x$ is good, then there exists $y$ and $z$ such that all the constraints of LP~\eqref{opt:slc} are satisfied, so we move on to the next round. Otherwise, $x$ is bad, so we increment $x$ until it becomes good, which implies that 
\[\sum_{e \in E(H)}{x_e \tuple{\sum_{\hat{s} \mid e \in P(\hat{s})}\alpha_{\hat{s},e} + \sum_{\hat{t} \mid e \in P(\hat{t})}\alpha_{\hat{t},e} }} \geq 1\]
for all feasible $\alpha$ in LP~\eqref{opt:int-slc-dual}. Let $A_i$ be the feasible polyhedron of LP~\eqref{opt:int-slc-dual} in round $i$. We derive the following LP which is equivalent to LP~\eqref{opt:slc}, by considering all the constraints of LP~\eqref{opt:int-slc-dual} from round 1 to round $k$.
\begin{equation} \label{opt:slc-covering} 
\begin{aligned}
& \min_{x} & & \sum_{e \in E(H)}{w'(e) x_e} \\
& \text{subject to}
& & \sum_{e \in E(H)}{x_e \tuple{\sum_{\hat{s} \mid e \in P(\hat{s})}\alpha_{\hat{s},e} + \sum_{\hat{t} \mid e \in P(\hat{t})}\alpha_{\hat{t},e} }} \geq 1 & \forall i \in [k], \alpha \in A_i,\\
& & & x_e \geq 0 & \forall e \in E(H).
\end{aligned}
\end{equation}

In round $i \in [k]$, the subroutine that solves LP~\eqref{opt:int-slc-dual} and checks if the optimum is good or not, is the separation oracle used for solving LP~\eqref{opt:slc-covering} online. Here we use Theorem~\ref{thm:covering} (the formal version of Theorem~\ref{thm:inf-covering}) to show that LP  \eqref{opt:slc-covering} can be solved online in polynomial time by paying an $O(\log n)$ factor. This requires that $\log(1/ \alpha_{\hat{s},e})$, $\log(1/ \alpha_{\hat{t},e})$, and $\log \lp^*$ where $\lp^*$ is the optimum of LP \eqref{opt:slc-covering}, can be represented by polynomial number of bits used for the edge weights $w'$. $\log \lp^*$ can be represented by polynomial number of bits. For $\log(1/ \alpha_{\hat{s},e})$ and $\log(1/ \alpha_{\hat{t},e})$, the subroutine that solves LP \eqref{opt:int-slc-dual} returns a solution $\alpha$ which is represented by polynomial number of bits. By Theorem~\ref{thm:covering}, we have the following Lemma.
\begin{lemma} \label{lem:online-unit-spanner-lp}
There exists a polynomial-time $O(\log n)$-competitive online algorithm for LP~\eqref{opt:slc}.
\end{lemma}

\subsubsection{Online rounding}
Now we are ready to show how to round the solution of LP~\eqref{opt:slc} online. The rounding scheme consists of two main steps. First we extract the representative sets according to the orderings $\succ_{i,r}$ and the online solution of LP~\eqref{opt:slc}, then reduce the online ordered Steiner label cover instance to an online undirected group Steiner forest instance.

We recall that in round $i$, we have terminal sets $\hat{S}^r_i := \hat{S}_i \cap H_r$ and $\hat{T}^r_i := \hat{T}_i \cap H_r$, and the relations $\hat{R}^r_i:=\hat{R}_i \cap (\hat{S}^r_i \times \hat{T}^r_i)$. The terminal sets and relations are partitioned based on the trees $H_r$ where $r \in V$. We also have orderings $\succ_{i,r}$ on $\hat{S}^r_i \cup \hat{T}^r_i$.

\paragraph{Extracting representative sets.} Let $\opt$ be the optimum of the Steiner label cover instance on $H$. Our goal is to find an integral solution of LP~\eqref{opt:slc} online with an objective value at most $\polylog(n) \opt$. We first use the online LP solution to find the representative vertex sets, then use the representatives to reduce to an online undirected group Steiner forest instance.

For each $i \in [k]$ and $r \in V$, we need to find the representative sets $\tilde{S}^r_i \subseteq \hat{S}^r_i$ and $\tilde{T}^r_i \subseteq \hat{T}^r_i$ such that $\tilde{S}^r_i \times \tilde{T}^r_i \subseteq \hat{R}^r_i$. Suppose in round $i$, we have a fractional solution $x,y,z$ of LP~\eqref{opt:slc}. The set pruning is accomplished by focusing on each $r \in V$ and its ordering $\prec_{i,r}$, and taking all terminal vertices according to their contribution to the objective value until we reach the one that cumulatively contributes half of the objective. Formally, for each $r \in V$:
\begin{itemize}
    \item Let $\gamma^r_i := \sum_{(\hat{s},\hat{t}) \in \hat{R}^r_i}y_{\hat{s},\hat{t}}$.
    \item  Choose the median sets, i.e., define the boundary vertices
    \[\hat{s}_{i,r}:= \max_{\prec_{i,r}}\bigg\{\hat{s} \in \hat{S}^r_i \: \bigg| \: \sum_{\hat{s}' \succeq_{i,r} \hat{s} }{\sum_{\hat{t} \mid (\hat{s}',\hat{t}) \in \hat{R}^r_i}{y_{\hat{s}',\hat{t}}}} \geq \gamma^r_i/2 \bigg\}\]
    and
    \[\hat{t}_{i,r}:= \min_{\prec_{i,r}}\bigg\{\hat{t} \in \hat{T}^r_i \: \bigg| \: \sum_{\hat{t}' \preceq_{i,r} \hat{t} }{\sum_{\hat{s} \mid (\hat{s},\hat{t}') \in \hat{R}^r_i}{y_{\hat{s},\hat{t}'}}} \geq \gamma^r_i/2 \bigg\},\]
    and let $\tilde{S}^r_i := \{\hat{s} \mid  \hat{s} \succeq_{i,r} \hat{s}_{i,r}\}$ and $\tilde{T}^r_i := \{\hat{t} \mid  \hat{t} \preceq_{i,r} \hat{t}_{i,r}\}$.
\end{itemize}

The choice of the median sets guarantees that at least half of the LP value is preserved. We have to verify that the union of cross-products of these sets does not contain any pairs that are disallowed by the relation $\hat{R}_i$. The following lemma closely follows the pairwise spanner framework in \cite{chlamtavc2020approximating}. The main difference is that the framework focuses on one tree for greedy purposes, while we focus on all trees for global approximation purpose.

\begin{lemma} \label{lem:rep-set}
The union of the cross-products of the representative sets $\tilde{S}^r_i$ and $\tilde{T}^r_i$ is a subset of $\hat{R}_i$, i.e.,
\[\bigcup_{r \in V}(\tilde{S}^r_i \times \tilde{T}^r_i) \subseteq \hat{R}_i.\]
\end{lemma}

\begin{proof}
The special case of a single tree (i.e., a single $r$) was proven in \cite{chlamtavc2020approximating}.
That is, for fixed $r \in V$, \cite{chlamtavc2020approximating} implies that
    $\tilde{S}^r_i \times \tilde{T}^r_i \subseteq \hat{R}^{r}_i$. Taking the union gives the desired claim. For the sake of completeness, the proof of a single tree case is provided in Appendix~\ref{sec:pf-rep-set}.
\end{proof}

\paragraph{From ordered Steiner label cover to group Steiner forest.} Given the representative sets, we reduce the online ordered Steiner label cover problem to the online undirected group Steiner forest problem as follows. For round $i$, let the representative sets be $\tilde{S}_i:=\bigcup_{r \in V}\tilde{S}^r_i$ and $\tilde{T}_i:=\bigcup_{r \in V}\tilde{T}^r_i$. We connect $\tilde{S}_i$ and $\tilde{T}_i$. It suffices to consider the cross product of $\tilde{S}_i$ and $\tilde{T}_i$ since $H$ is partitioned by $H_r$, i.e., representative vertices from different trees in $H$ cannot be connected.

\begin{figure}[h]
\centering
\begin{tikzpicture}
    \draw (-2,0) -- (2,0) -- (0,2) --cycle;
    \node[draw=none] at (0,0.5) {$T_{r_1}^-$};
    \draw (4,0) -- (8,0) -- (6,2) --cycle;
    \node[draw=none] at (6,0.5) {$T_{r_1}^+$};
    \node at (0,2.5) {$r_1^-$};
    \node at (6,2.5) {$r_1^+$};
    \node at (3,3) {$r'_1$};
    \node[fill,circle, inner sep=0pt, minimum size=0.2cm] (t1) at (0,2) {};
    \node[fill,circle, inner sep=0pt, minimum size=0.2cm] (t2) at (6,2) {};
    \node[fill,circle, inner sep=0pt, minimum size=0.2cm] (t3) at (3,2.5) {};
    \path[-]
        (t1) edge (t3)
        (t2) edge (t3);
    
    \draw (-2,-1) -- (2,-1) -- (0,-3) --cycle;
    \node[draw=none] at (0,-1.5) {$T_{r_2}^-$};
    \draw (4,-1) -- (8,-1) -- (6,-3) --cycle;
    \node[draw=none] at (6,-1.5) {$T_{r_2}^+$};
    \node at (0,-3.5) {$r_2^-$};
    \node at (6,-3.5) {$r_2^+$};
    \node at (3,-4) {$r'_2$};
    \node[fill,circle, inner sep=0pt, minimum size=0.2cm] (s1) at (0,-3) {};
    \node[fill,circle, inner sep=0pt, minimum size=0.2cm] (s2) at (6,-3) {};
    \node[fill,circle, inner sep=0pt, minimum size=0.2cm] (s3) at (3,-3.5) {};
    \path[-]
        (s1) edge (s3)
        (s2) edge (s3);
        
    \node at (-1.8+6,-0.5) {$\tilde{T}_1$};
    \node at (-0.8+6,-0.5) {$\tilde{T}_2$};
    \node at (1.8+6,-0.5) {$\tilde{T}_3$};
    \draw[-] (-2+6,0) -- (-2+6,-0.2) -- (-1.6+6,-0.2) -- (-1.6+6,0);
    \draw[-] (-1+6,0) -- (-1+6,-0.2) -- (-0.6+6,-0.2) -- (-0.6+6,0);
    \draw[-] (2+6,0) -- (2+6,-0.2) -- (1.6+6,-0.2) -- (1.6+6,0);
    \draw[-] (-2+6,-1) -- (-2+6,-0.8) -- (-1.6+6,-0.8) -- (-1.6+6,-1);
    \draw[-] (-1+6,-1) -- (-1+6,-0.8) -- (-0.6+6,-0.8) -- (-0.6+6,-1);
    \draw[-] (2+6,-1) -- (2+6,-0.8) -- (1.6+6,-0.8) -- (1.6+6,-1);
    
    \node at (-1.8,-0.5) {$\tilde{S}_1$};
    \node at (-0.8,-0.5) {$\tilde{S}_2$};
    \node at (1.8,-0.5) {$\tilde{S}_k$};
    \draw[-] (-2,0) -- (-2,-0.2) -- (-1.6,-0.2) -- (-1.6,0);
    \draw[-] (-1,0) -- (-1,-0.2) -- (-0.6,-0.2) -- (-0.6,0);
    \draw[-] (2,0) -- (2,-0.2) -- (1.6,-0.2) -- (1.6,0);
    \draw[-] (-2,-1) -- (-2,-0.8) -- (-1.6,-0.8) -- (-1.6,-1);
    \draw[-] (-1,-1) -- (-1,-0.8) -- (-0.6,-0.8) -- (-0.6,-1);
    \draw[-] (2,-1) -- (2,-0.8) -- (1.6,-0.8) -- (1.6,-1);
\end{tikzpicture}
\caption{Connecting $\tilde{S}_i$ and $\tilde{T}_i$ in the undirected group Steiner forest problem on $H$}
\label{fig:h-prime}
\end{figure}

The goal is to select edges from $E(H)$, such that there is an $\hat{s}$-$\hat{t}$ path for $\hat{s} \in \hat{S}_i$ and $\hat{t} \in \hat{T}_i$ for all $i \in [k]$ and the total edge weight is approximately optimal. 

The online buy-at-bulk framework \cite{cekp} considers a height reduction forest and uses the single-source (sink) group Steiner tree problem \cite{alon2006general} as a black-box to solve the online directed buy-at-bulk problem. Here, we apply this approach to our online undirected group Steiner forest instance.

\begin{theorem}(\cite{cekp})
\label{thm:bab}
Given the graph $H$ and the terminal sets $\tilde{S}_i$ and $\tilde{T}_i$ for $i \in [k]$ as the undirected group Steiner forest instance, there is a randomized polynomial-time algorithm with competitive ratio $\polylog(n)$.
\end{theorem}

We recall that $\opt$ is the optimum of the Steiner label cover instance on $H$ with terminal set pairs $(\hat{S}_i,\hat{T}_i)$ and relations $\hat{R}_i$. Using Theorem~\ref{thm:bab}, we show the following which completes the proof for Lemma~\ref{lem:online-slc}.

\begin{lemma} \label{lem:slc-bab}
There exists a randomized polynomial-time online algorithm for LP~\eqref{opt:slc} which returns an integral solution such that the objective value is at most $\polylog(n) \opt$.
\end{lemma}

\begin{proof}
Let $x$, $y$, and $z$ be a feasible solution for LP~\eqref{opt:slc} with objective value $\lp$.
By Claim~\ref{cl:ord} and Lemma~\ref{lem:rep-set}, 
we can obtain a feasible solution $x'$, $y'$, and $z'$ for LP~\eqref{opt:slc} with objective value at most $2\lp$ as follows. For any $r \in V$ and $(\hat{s},\hat{t}) \in \hat{R}^r_i$, there is a corresponding variable $y_{\hat{s},\hat{t}}$ for the $\hat{s}$-$r'$-$\hat{t}$ path. 
If $\hat{s} \notin \tilde{S}^r_i$, then we replace the $\hat{s}$-$r'$ path by an $\hat{s'}$-$r'$ path where $\hat{s'} \in \tilde{S}^r_i$; if $\hat{t} \notin \tilde{T}^r_i$, then we replace the $r'$-$\hat{t}$ path by an $r'$-$\hat{t'}$ path where $\hat{t'} \in \tilde{T}^r_i$. This shifts the path selection $y$ towards the paths with representative endpoints $\hat{s} \in \tilde{S}^r_i$ and $\hat{t} \in \tilde{T}^r_i$ and addresses the new path assignment $y'$ (see Figure \ref{fig:2lp} for an illustration). We obtain $x'$ by adjusting the capacity. For each edge $e$ that belongs to an $\hat{s}$-$\hat{t}$ path where $(\hat{s},\hat{t}) \in \tilde{S}^r_i \times \tilde{T}^r_i$, we have that $x'_e \leq 2x_e$ because the representatives contribute at least half of the original LP 
path selection. If $e$ does not belong the these paths, then $x'_e = 0$. We also adjust the flow value and obtain $z'$ accordingly.

\begin{figure}[ht]
\centering
\begin{subfigure}{1\textwidth}
\centering
\begin{tikzpicture}
\coordinate (O) at (0,0); 

\node at (2, 0)  (1) {$\hat{s}_1$};
\node at (4, 0)  (2) {$\hat{s}_2$};
\node at (6, 0)  (3) {$\hat{s}_3$};
\node at (8, 0)  (4) {$\hat{s}_4$};
\node at (10, 0)  (5) {$\hat{t}_1$};
\node at (12, 0)  (6) {$\hat{t}_2$};
\node at (14, 0)  (7) {$\hat{t}_3$};
\node at (16, 0)  (8) {$\hat{t}_4$};

\node at (3, 0) {$\prec$};
\node at (5, 0) {$\prec$};
\node at (7, 0) {$\prec$};
\node at (9, 0) {$\prec$};
\node at (11, 0) {$\prec$};
\node at (13, 0) {$\prec$};
\node at (15, 0) {$\prec$};

\draw (2,0.2) -- (2,0.5) -- (10,0.5) -- (10,0.2);
\draw (4,0.6) -- (4,0.9) -- (12,0.9) -- (12,0.6);
\draw (6,1) -- (6,1.3) -- (14,1.3) -- (14,1);
\draw (8,1.4) -- (8,1.7) -- (16,1.7) -- (16,1.4);

\draw[dashed] (6,-0.2) -- (6,-0.5) -- (8,-0.5) -- (8,-0.2);
\draw[dashed] (10,-0.2) -- (10,-0.5) -- (12,-0.5) -- (12,-0.2);

\draw[dashed] (9,2) -- (9,0.2);

\node at (7,-1) {$\hat{S}^r_i$};
\node at (11,-1) {$\hat{T}^r_i$};

\end{tikzpicture}
\caption{The original path selection by LP~\eqref{opt:slc}}
\end{subfigure}
\newline

\hspace{0.5pt}

\begin{subfigure}{1\textwidth}
\centering
\begin{tikzpicture}
\coordinate (O) at (0,0); 

\node at (2, 0)  (1) {$\hat{s}_1$};
\node at (4, 0)  (2) {$\hat{s}_2$};
\node at (6, 0)  (3) {$\hat{s}_3$};
\node at (8, 0)  (4) {$\hat{s}_4$};
\node at (10, 0)  (5) {$\hat{t}_1$};
\node at (12, 0)  (6) {$\hat{t}_2$};
\node at (14, 0)  (7) {$\hat{t}_3$};
\node at (16, 0)  (8) {$\hat{t}_4$};

\node at (3, 0) {$\prec$};
\node at (5, 0) {$\prec$};
\node at (7, 0) {$\prec$};
\node at (9, 0) {$\prec$};
\node at (11, 0) {$\prec$};
\node at (13, 0) {$\prec$};
\node at (15, 0) {$\prec$};

\draw (6,0.2) -- (6,0.5) -- (10,0.5) -- (10,0.2);
\draw (8,0.6) -- (8,0.9) -- (12,0.9) -- (12,0.6);
\draw (6,1) -- (6,1.3) -- (10,1.3) -- (10,1);
\draw (8,1.4) -- (8,1.7) -- (12,1.7) -- (12,1.4);

\draw[dotted] (2,0.2) -- (2,0.5) -- (6,0.5);
\draw[dotted] (4,0.6) -- (4,0.9) -- (8,0.9);
\draw[dotted] (10,1.3) -- (14,1.3) -- (14,1);
\draw[dotted] (12,1.7) -- (16,1.7) -- (16,1.4);

\draw[dashed] (6,-0.2) -- (6,-0.5) -- (8,-0.5) -- (8,-0.2);
\draw[dashed] (10,-0.2) -- (10,-0.5) -- (12,-0.5) -- (12,-0.2);

\draw[dashed] (9,2) -- (9,0.2);

\node at (7,-1) {$\hat{S}^r_i$};
\node at (11,-1) {$\hat{T}^r_i$};

\end{tikzpicture}
\caption{The adjusted path selection with representative endpoints for LP~\eqref{opt:slc}}
\end{subfigure}
\caption{Shifting paths towards paths with representative endpoints}
\label{fig:2lp}
\end{figure}

Now we consider the following LP for undirected group Steiner forest on $H$.
\begin{equation} \label{opt:slc-flow} 
\begin{aligned}
& \min_{x,\gamma} & & \sum_{e \in E(H)}{w'(e) x_e} \\
& \text{subject to}
& & \sum_{r \in V}\gamma^r_i \geq 1 & \forall i \in [k],\\
& & & \text{$x$ supports a $\tilde{S}^r_i$-$r'$ flow of value $\gamma^r_i$} & \forall i \in [k], \forall r \in V,\\
& & & \text{$x$ supports a $\tilde{T}^r_i$-$r'$ flow of value $\gamma^r_i$} & \forall i \in [k], \forall r \in V,\\
& & & x_e \geq 0 & \forall e \in E(H'),\\
& & & \gamma^r_i \geq 0 & \forall i \in [k], r \in V.
\end{aligned}
\end{equation}

Let $\tilde{x}_e=x'_e$ for each $e \in E(H)$ and $\tilde{\gamma}^r_i=\sum_{(\hat{s},\hat{t}) \in \tilde{R}^r_i}y'_{\hat{s},\hat{t}}$.
We observe that $\tilde{x}$ and $\tilde{\gamma}$ is a feasible solution for LP~\eqref{opt:slc-flow} with objective value at most $2\lp$.
This implies that for any solution of LP~\eqref{opt:slc}, there exists a corresponding solution of LP~\eqref{opt:slc-flow} with at most twice of the objective value.

Given a solution $\tilde{x}$ and $\tilde{\gamma}$ for LP~\eqref{opt:slc-flow}, we can obtain a solution $x'$, $y'$, and $z'$ for LP~\eqref{opt:slc} as follows. We note that the objective values of these two solutions are the same.
\begin{itemize}
    \item Let $x'_e = \tilde{x}_e$ for each $e \in E(H)$.
    \item Let $z'_{\hat{s}}=\min_{e \in P(\hat{s})}\{x'_e\}$, i.e., $z'_{\hat{s}}$ is the minimum capacity among the edges that belong to the path from $\hat{s}$ to its tree root $r(\hat{s})$.
    \item Let $z'_{\hat{t}}=\min_{e \in P(\hat{t})}\{x'_e\}$, i.e., $z'_{\hat{t}}$ is the minimum capacity among the edges that belong to the path from the tree root $r(\hat{t})$ to $\hat{t}$.
    \item Let $f^r_{s_i}$ be an $\tilde{S}^r_i$-$r'$ flow of value $\gamma^r_i$ and $f^r_{t_i}$ be an $\tilde{T}^r_i$-$r'$ flow of value $\gamma^r_i$, both supported by edge capacity $\tilde{x}$.
    \item For each $i \in [k]$ and $(\hat{s},\hat{t}) \in \tilde{R}_i$, let $y'_{\hat{s},\hat{t}}$ be the flow value along the $\hat{s}$-$\hat{t}$ path according to the flow $f$.
\end{itemize}

Now we use Theorem~\ref{thm:bab} to obtain an integral solution for LP~\eqref{opt:slc-flow}, then transform this solution to the one for LP~\eqref{opt:slc} as described above. Suppose we have a $\polylog(n)$-competitive solution for the undirected group Steiner forest instance. For each edge $e \in E(H)$ that is selected, set $\tilde{x}_e=1$, otherwise let $\tilde{x}_e=0$. Let $\gamma^r_i=1$ whenever there is an $\hat{s}$-$r'$-$\hat{t}$ path in $H$ where $\hat{s} \in \tilde{S}_i$ and $\hat{t} \in \tilde{T}_i$, otherwise let $\gamma^r_i=0$. We observe that every integral solution of LP~\eqref{opt:slc-flow} corresponds to an undirected group Steiner forest solution for $H$ and $(s_i,t_i)$.
By the structure of $H$, the integrality gap of LP~\eqref{opt:slc-flow} is $\polylog(n)$ \cite{garg2000polylogarithmic,cekp}.
This implies that the objective values obtained from the online undirected group Steiner forest solution for both LP~\eqref{opt:slc} and LP~\eqref{opt:slc-flow} are upper-bounded by $\polylog(n)\opt$.
\end{proof}

\subsection{Summary} \label{sec:us-sum}

We summarize in Algorithm~\ref{alg:unit-ps} the overall $\tO(k^{1/2 + \delta})$-competitive randomized polynomial-time algorithm described in Sections~\ref{sec:lcp-hr} and \ref{sec:online-slc} for pairwise spanners on graphs with uniform edge lengths.

\begin{algorithm}[H]
\caption{$\tO(k^{1/2 + \delta})$-competitive pairwise spanner} \label{alg:unit-ps}
\begin{algorithmic}[1]
\State Given a directed graph $G$, construct the directed graph $G'$, which is a union of disjoint layered graphs. From $G'$, construct the undirected graph $H$, which is a forest.
\For {an arriving pair $(s_i,t_i)$}
    \State Generate $\hat{S}_i$, $\hat{T}_i$, and $\hat{R}_i$, and solve the Steiner label cover LP~\eqref{opt:slc} online.\footnotemark
    \State Generate the orderings $\prec_{i,r}$. From the online solution of LP~\eqref{opt:slc}, extract the representative sets $\tilde{S}_i$ and $\tilde{T}_i$, and the corresponding relation $\tilde{R}_i$.
	\State Solve the online Steiner forest instance for graph $H$ and the terminal vertex set pair $\tilde{S}_i$ and $\tilde{T}_i$. Recover an integral solution for LP~\eqref{opt:slc}, obtain a solution which is a forest for the ordered Steiner label cover problem in $H$, and map it back to the collection of junction trees in $G$.
\EndFor
\end{algorithmic}
\end{algorithm}
\footnotetext{One may derive an LP for the graph $G'$ instead of $H$. However, the integrality gap might be large.}

We conclude that Algorithm~\ref{alg:unit-ps} is $\tO(k^{1/2 + \delta})$-competitive. We pay $O(\sqrt{k})$ for the junction tree approximation, $O(k^{\delta})$ for constructing the forest $H$ by height reduction, and finally $\polylog(n)$ for solving LP~\eqref{opt:slc} and the undirected group Steiner forest instance $H$, $\tilde{S}_i$, and $\tilde{T}_i$ online. The overall competitive ratio is $\tO(k^{1/2+\delta})$.

\section{Online Pairwise Spanners} \label{subsubsec:online-ps}

We recall that in the general {\em pairwise spanner} problem, we are given a directed graph $G=(V, E)$ with edge length $\ell: E \to \R_{\geq 0}$, a general set of $k$ terminals $D = \{(s_i,t_i) \mid i \in [k]\} \subseteq V \times V$, and a target distance $d_i$ for each terminal pair $(s_i,t_i)$, the goal is to output a subgraph $H=(V, E')$ of $G$ such that for every pair $(s_i, t_i) \in D$ it is the case that $d_H(s_i, t_i)\leq d_i$, i.e., the distance of a shortest $s_i \leadsto t_i$ path is at most $d_i$ in the subgraph $H$, and we want to minimize the number of edges in $E'$.

\subsection{An $\tO(n^{4/5})$-competitive online algorithm for pairwise spanners} \label{subsec:general-ps}

In this section, we prove Theorem~\ref{theorem:spanner}. Recall that in the online setting of the problem, the directed graph $G$ is given offline, while the vertex pairs in $D \subseteq V \times V$ arrive online one at a time. In the beginning, $E' = \emptyset$. Suppose $(s_i,t_i)$ and its target distance $d_i$ arrive in round $i$, we select some edges from $E$ and irrevocably add them to $E'$, such that in the subgraph $H=(V,E')$, $d_{H}(s_i,t_i) \leq d_i$. The goal is to approximately minimize the total number of edges added to $E'$.

We start with a high-level sketch of an offline algorithm, which we will build on for the online setting. The randomized rounding framework in \cite{chlamtavc2020approximating,berman2013approximation} has two main steps. One step is to solve and round an LP for the spanner problem. The second is to uniformly sample vertices and add the shortest path in-arborescences and out-arborescences rooted at each of the sampled vertices.
Terminal pairs are classified as either \emph{thin} or \emph{thick} and are addressed by one of the two steps above accordingly. 

In the first step, the rounding scheme based on an LP solution for spanners ensures with high probability that for all thin terminal pairs the distance requirements are met. The second step ensures with high probability that for all thick terminal pairs the distance requirements are met. By selecting an appropriate threshold for classifying the thin and thick pairs, this leads to an $O(n/\sqrt{\opt})$-approximation, where $\opt$ is the number of edges in the optimal solution.

The main challenges in adapting this approach to the online setting are as follows: 1) $\opt$ can be very small, and 2) the LP for spanners is not naturally a pure covering LP, which makes it difficult to solve online.  In the previous work in the offline setting, the small-$\opt$ case is addressed by sophisticated strategies that appear difficult to emulate online.
Instead, we show that the optimal value (however small) is at least the square root of the number of terminal pairs that have arrived. Thus, if $\opt$ is small and not many pairs have arrived, we can greedily add a path with the fewest edges subject to the distance requirement for each pair. To overcome the second challenge, we convert the LP for spanners into an equivalent covering LP as in \cite{dinitz2019lasserre}, where violating covering constraints are generated by an auxiliary LP. Having transformed the LP into a purely covering one, we can solve the LP online, treating the auxiliary LP as a separation oracle. 

\subsubsection{A simple $\tO(n^{4/5})$-approximate offline algorithm based on \cite{chlamtavc2020approximating}}
For ease of exposition, we first design a simpler offline algorithm (slightly weaker than the state-of-the-art) that is more amenable to the online setting. This allows us to establish the main ingredients governing the approximation factor in a simpler setting, and then address the online aspects separately. The algorithm leverages the framework developed in \cite{berman2013approximation,chlamtavc2020approximating}.

Let $\cP_i$ denote the collection of $s_i \leadsto t_i$ paths of distance at most $d_i$ consisting of edges in $E$. Let the \emph{local graph} $G^i=(V^i,E^i)$ be the union of all vertices and edges in $\cP_i$. A pair $(s_i,t_i) \in D$ is \emph{$t$-thick} if $|V^i| \geq t$, otherwise $(s_i,t_i)$ is \emph{$t$-thin}. Consider the following standard LP relaxation (essentially the one in \cite{dinitz2011directed}).

\begin{equation} \label{opt:off-dks-lp}
\begin{aligned}
& \min_{x,y} & & \sum_{e \in E}{x_e} \\
& \text{subject to}
& & \sum_{P \in \cP_i}{y_P} \geq 1 & \forall i \in [k],\\
& & & \sum_{P \mid e \in P \in \cP_i}y_P \leq x_e & \forall e \in E, \forall i \in [k],\\
& & & x_e \geq 0 & \forall e \in E,\\
& & & y_P \geq 0 & \forall P \in \cP_i \quad \forall i \in [k].\\
\end{aligned}
\end{equation}
Herein, $x_e$ is an indicator of edge $e$ and $y_P$ is an indicator of path $P$. Suppose we have an integral feasible solution. Then the first set of constraints ensures that there is at least one $s_i \leadsto t_i$ path of distance at most $d_i$ selected, and the second set of constraints ensures that if a path $P$ is selected, then all its edges are selected.

We say that a pair $(s_i,t_i) \in D$ is \emph{settled} if the selection of edges is such that there exists an $s_i \leadsto t_i$ path of distance at most $d_i$. Applying a simple rounding scheme based on a solution of LP~\eqref{opt:off-dks-lp} settles the thin pairs with high probability, while sampling enough vertices and adding shortest path in-arborescences and out-arborescences rooted at each sampled vertex ensures with high probability that thick pairs are settled. Let $\opt$ be the optimum value of the given pairwise spanner instance. Without loss of generality, we may assume that we know $\opt$ since we can guess every value of $\opt$ in $[|E|]$ in the offline setting. Now we are ready to describe Algorithm~\ref{alg:offline-dks} in \cite{chlamtavc2020approximating}.

\begin{algorithm}[H]
\caption{Offline pairwise spanner} \label{alg:offline-dks}
\begin{algorithmic}[1]
\State $E' \gets \emptyset$ and $t \gets n/\sqrt{\opt}$.
\State Solve LP \eqref{opt:off-dks-lp} and add each edge $e \in E$ to $E'$ with probability $\min\{1,x_e t \ln n\}$ independently.
\State Obtain a vertex set $W \subseteq V$ by sampling $(3n\ln n)/t$ vertices from $V$ independently and uniformly at random. Add the edges of shortest path in-arborescences and out-arborescences rooted at $w$ for each $w \in W$.
\end{algorithmic}
\end{algorithm}

The approximation guarantee of Algorithm~\ref{alg:offline-dks} is stated in the following lemma. The approximation (competitiveness) guarantee of the offline $\tO(n^{4/5})$-approximate algorithm (Algorithm \ref{alg:simple-ps}) and the online $\tO(n^{4/5})$-competitive algorithm (Algorithm \ref{alg:dks}) presented later are based on the proof of this lemma. For this reason and the sake of completeness we include the proof below.

\begin{restatable}{lemma}{lemnsqrtopt}\label{lem:nsqrtopt}(\cite{chlamtavc2020approximating})
Algorithm~\ref{alg:offline-dks} is $\tilde{O}(n/\sqrt{\opt})$-approximate.
\end{restatable}

\begin{proof}
The following claim that settles thin pairs is due to \cite{berman2013approximation}. Essentially, when $n$ is large, for any $t > 2$, with high probability, all the $t$-thin pairs are settled by the rounding scheme suggested below.

\begin{claim}(\cite{berman2013approximation}) \label{lem:otlp}
Let $x,y$ be a feasible solution of LP \eqref{opt:off-dks-lp} with objective value $\lp$. By retaining every edge $e \in E$ independently with probability $\min\{1,x_e \cdot t \ln n\}$, all $t$-thin pairs are settled with probability at least $1-\exp(-t \ln (n/t) + 2 \ln n)$, and the number of edges retained is at most $\tilde{O}(t)\lp$.
\end{claim}

To settle $t$-thick pairs, 
we obtain a vertex set $W \subseteq V$ by sampling $(3n\ln n)/t$ vertices from $V$ independently and uniformly at random, then add edges of the shortest path in-arborescences and out-arborescences rooted at $w$ for each $w \in W$. For each $t$-thick pair $(s_i,t_i)$, the probability that $W \cup V^i = \emptyset$ is $(1-t/n)^{-(3n\ln n)/t} \leq \exp(-3 \ln n) = 1/n^3$. There are at most $n^2$ $t$-thick pairs, so by a union bound, with probability at least $1-1/n$, all $t$-thick pairs in $D$ are settled.

We note that there is a polynomial-time constant factor approximate algorithm for solving LP \eqref{opt:off-dks-lp} according to \cite{dinitz2011directed} even though there are an exponential number of variables and constraints. Let $\lp^*$ denote the optimal objective value of LP \eqref{opt:off-dks-lp}. By Lemma~\ref{lem:otlp}, the $t$-thin pairs are settled in the second step and the number of edges retained is at most $\tO(t)\lp^*$. The $t$-thick pairs are settled by adding shortest path arborescences in the third step, and the number of edges retained is at most $2n \cdot 3n \ln n/t$ since each arborescence has at most $n$ edges. 
It follows that the approximation guarantee is
\[\frac{\tilde{O}(t \lp^*) + \tilde{O}(n^2/t)}{\opt} \leq \frac{\tO(t \opt) + \tO(n^2/t)}{\opt} = \frac{\tO(n\sqrt{\opt})}{\opt} = \tO(\frac{n}{\sqrt{\opt}})\]
where the first inequality follows by $\lp^* \leq \opt$. Namely, the optimal fractional solution is upper-bounded by the optimal integral solution of LP \eqref{opt:off-dks-lp}, and the second equality follows by the definition $t=n/\sqrt{\opt}$.
\end{proof}

In the all-pairs spanner problem where $\opt$ is $\Omega(n)$, Algorithm~\ref{alg:offline-dks} is $\tO(\sqrt{n})$-approximate which matches the state-of-the-art approximation ratio given in \cite{berman2013approximation}. For the pairwise spanner problem, the main challenge is the lack of a nice lower bound for $\opt$. In the offline setting, \cite{chlamtavc2020approximating} achieves an $\tO(n^{3/5 + \ep})$-approximate solution by a careful case analysis when edges have uniform lengths. We give an alternative approach that is amenable to the online setting by considering two cases, where one resolves the issue when $\opt$ does not have a nice lower bound, and the other uses a variant of Algorithm~\ref{alg:offline-dks} given that $\opt$ has a nice lower bound. This approach relies on the following observation.

\begin{lemma} \label{lem:optgeqsqrt}
$\opt \geq \sqrt{k}$.
\end{lemma}

\begin{proof}
We observe that when the spanner has $\ell$
edges, there are at most $\ell$ source vertices and $\ell$ sink vertices, so there are at most $\ell^2$ terminal pairs. Therefore, when the spanner has $\opt$ edges, there are at most $\opt^2$ terminal pairs, so $\opt \geq \sqrt{k}$.
\end{proof}

Now we specify the simple offline algorithm given in Algorithm~\ref{alg:simple-ps}. In the beginning, we set two parameters $T$ and $t$ (which we will describe later), and set $E'=\emptyset$. An $s_i \leadsto t_i$ path is \emph{cheapest feasible} if it meets the distance requirement $d_i$ by using the minimum number of edges from $E$. We note that cheapest feasible paths can be found by Bellman-Ford algorithm.

\begin{algorithm}[H]
\caption{Simple offline pairwise spanner} \label{alg:simple-ps}

\begin{algorithmic}[1]
\If{$k < T$}
    \State Add the edges of a cheapest feasible $s_i \leadsto t_i$ path to $E'$ for each $i \in [k]$.
\Else
    \State Solve LP \eqref{opt:off-dks-lp} and add each edge $e \in E$ to $E'$ with probability $\min\{1,x_e t \ln n\}$ independently.
    \State Obtain a vertex set $W \subseteq V$ by sampling $(3n\ln n)/t$ vertices from $V$ independently and uniformly at random. Add the edges of shortest path in-arborescences and out-arborescences rooted at each vertex $w \in W$ to $E'$.
\EndIf
\end{algorithmic}
\end{algorithm}

\begin{lemma}
Algorithm~\ref{alg:simple-ps} is $\tO(n^{4/5})$-approximate when $T=t=n^{4/5}$.
\end{lemma}

\begin{proof}
If $k < n^{4/5}$, we add the edges of a cheapest feasible $s_i \leadsto t_i$ path for each $(s_i,t_i) \in D$. Each cheapest feasible $s_i \leadsto t_i$ path contains at most $\opt$ edges, so the ratio between this solution and $\opt$ is $n^{4/5}$. The remaining case is $k \geq n^{4/5}$, which implies $\opt \geq n^{2/5}$ by Lemma~\ref{lem:optgeqsqrt}. Let $\lp^*$ be the optimal objective value of LP~\eqref{opt:off-dks-lp}. By Claim~\ref{lem:otlp}, the rounding scheme settles the $t$-thin pairs with high probability. The $t$-thick pairs are settled by a collection of shortest path in-arborescences and out-arborescences using a similar argument in the proof of Lemma~\ref{lem:nsqrtopt}. The approximation guarantee is
\[\frac{\tilde{O}(t \lp^*) + \tilde{O}(n^2/t)}{\opt} \leq \frac{\tO(n^{4/5} \opt) + \tO(n^{6/5})}{\opt}  = \tO(n^{4/5})\] since the number of edges retained from the rounding scheme is at most $\tO(t)\lp^*$ and the number of edges retained by adding arborescences is at most $2n \cdot 3n \ln n/t$ since each arborescence has at most $n$ edges. This summarizes the simple offline $\tilde{O}(n^{4/5})$-approximation algorithm.
\end{proof}

\subsubsection{An $\tO(n^{4/5})$-competitive online algorithm}

It remains to convert the simple offline algorithm to an online algorithm. We address the two main modifications.
\begin{enumerate}
    \item We have to (approximately) solve LP \eqref{opt:off-dks-lp} online, which is not presented as a covering LP.
    \item We have to round the solution of LP \eqref{opt:off-dks-lp} online.
\end{enumerate}

For the first modification, LP \eqref{opt:off-dks-lp} is converted to an equivalent covering LP \eqref{opt:cs-LP3} (which we will show later) and approximately solved in an online fashion. For the second modification, we use an online version of the rounding scheme in Algorithm~\ref{alg:simple-ps}, such that the overall probability (from round 1 to the current round) for the edge selection is consistent with the probability based on the online solution of LP \eqref{opt:off-dks-lp}, by properly scaling the probability based on a conditional argument.

\begin{algorithm}[H]
\caption{Online pairwise spanner} \label{alg:dks}

\begin{algorithmic}[1]
\For {an arriving pair $(s_i,t_i)$}
    \State Convert the spanner LP \eqref{opt:off-dks-lp} to the covering LP \eqref{opt:cs-LP3} and solve LP \eqref{opt:cs-LP3} online (see the proof of Theorem~\ref{theorem:spanner} for the details).
	\If{$i < T$}
	    \State Add the edges of a cheapest feasible $s_i \leadsto t_i$ path to $E'$. \label{line:alg-spanner-greedy}
	\ElsIf{$i = T$}
	    \State Obtain a vertex set $W \subseteq V$ by sampling $(3n\ln n)/t$ vertices from $V$ independently and uniformly at random. Add the edges of shortest path in-arborescences and out-arborescences rooted at each vertex $w \in W$ to $E'$.
	    \State Add each edge $e$ to $E'$ independently with probability $p^i_e$ for each edge $e \in E \setminus E'$.
	\Else \Comment{$i > T$}
	\State Add each edge $e$ to $E'$ independently with probability $(p^i_e - p^{i-1}_e)/(1-p^{i-1}_e)$ for each edge $e \in E \setminus E'$.
	\EndIf
\EndFor
\end{algorithmic}
\end{algorithm}

The online algorithm in round $i$ is given in Algorithm~\ref{alg:dks}. The same structure will be used later for other variants of the online pairwise spanner problem. All these algorithms rely on the LP rounding scheme (Claim~\ref{lem:otlp}) to settle the $t$-thin pairs with high probability and a collection of shortest path in-arborescences and out-arborescences to settle $t$-thick pairs (similar to Lemma~\ref{lem:nsqrtopt}).

In the beginning, we pick a threshold parameter $T$ and a thickness parameter $t$ (which we will show later), and set $E'=\emptyset$. Let $x^i_e$ denote the value of $x_e$ in the approximate solution of LP \eqref{opt:cs-LP3} obtained in round $i$. Let $p^i_e:=\min\{1,x^i_e t \ln n\}$. Algorithm~\ref{alg:dks} is the online version of Algorithm~\ref{alg:simple-ps}. A key insight is that when we add the arborescences in round $T$, it also settles the \emph{future} thick terminal pairs with high probability. With the outline structure of the online algorithm, we are ready to prove Theorem~\ref{theorem:spanner}.

\thmpdks*

\begin{proof}
Suppose in the online setting, there are $k$ rounds where $k$ may be unknown. In round $i \in [k]$, the pair $(s_i,t_i)$ and the distance requirement $d_i$ arrive and we select some new edges from $E$ to settle $(s_i,t_i)$. We run Algorithm~\ref{alg:dks} by setting $T=t=\lfloor n^{4/5} \rfloor$. It suffices to show that 1) LP \eqref{opt:off-dks-lp} can be solved online by losing a polylogarithmic factor, and 2) the overall probability of edge selection is consistent with the probability based on the online solution of LP \eqref{opt:off-dks-lp}.

\paragraph{Converting and solving LP \eqref{opt:off-dks-lp} online.} 
The goal is to update $x$ in a non-decreasing manner upon the arrival of the pair $(s_i,t_i)$ to satisfy all its corresponding constraints, so that the objective value is still approximately optimal. We convert LP \eqref{opt:off-dks-lp} into a covering LP as follows.

First, we check in round $i$, given the \emph{edge capacity} $x$, if there is a (fractional) $s_i \leadsto t_i$ path of distance at most $d_i$. This can be captured by checking the optimum of the following LPs,
\begin{equation} \label{opt:cs-LP1} 
\begin{aligned}
&\max_{y} & & \sum_{P \in \cP_i}{y_P}\\
&\text{subject to}
& &\sum_{P \mid e \in P \in \cP_i}y_P \leq x_e & \forall e \in E,\\
& & &y_P \geq 0 & \forall P \in \cP_i,\\
\end{aligned}
\end{equation}
and its dual
\begin{equation} \label{opt:cs-LP2} 
\begin{aligned}
&\min_{z} & & \sum_{e \in E}{x_e z_e}\\
&\text{subject to}
& &\sum_{e \in P}z_e \geq 1 & \forall P \in \cP_i,\\
& & &z_e \geq 0 & \forall e \in E.\\
\end{aligned}
\end{equation}

We say that $x$ is \emph{good} if the optimum of LP \eqref{opt:cs-LP1} and LP \eqref{opt:cs-LP2} is at least 1, and it is \emph{bad} otherwise. Namely, $x$ is good if there is at least one (fractional) $s_i \leadsto t_i$ path of distance at most $d_i$. In LP \eqref{opt:cs-LP2}, the feasibility problem is equivalent to the following problem. Given the local graph $G^i$ and edge weight $z$, is there an $s_i \leadsto t_i$ path of distance at most $d_i$ and weight less than $1$? We note that with uniform lengths, this problem can be solved by Bellman-Ford algorithm with $d_i$ iterations, which computes the smallest weight among all $s_i \leadsto t_i$ paths of distance at most $d_i$ in the local graph $G^i$.

Although this bicriteria path problem in general is NP-hard \cite{arkin1991bicriteria}, an FPTAS is known to exist \cite{hassin1992approximation,lorenz2001simple}, which gives an approximate separation oracle. We can verify in polynomial time that if there is a path of distance at most $d_i$ and weight less than $1-\ep$. We obtain a solution $z'$ for LP \eqref{opt:cs-LP2} where each constraint is satisfied by a factor of $1-\ep$ and set $z:=z'/(1-\ep)$ as the solution.

To solve LP \eqref{opt:off-dks-lp}, suppose in round $i$, we are given $x$. First, we check if $x$ is good or bad by approximately solving LP \eqref{opt:cs-LP2}. If $x$ is good, then there exists $y$ such that $\sum_{P \in \cP_j}{y_P} \geq 1$, i.e., the solution is feasible for LP \eqref{opt:off-dks-lp}, so we move on to the next round. Otherwise, $x$ is bad, so we increment $x$ 
until it becomes good, which implies $\sum_{e \in E}{x_e z_e} \geq 1$ for all feasible $z$ in LP \eqref{opt:cs-LP2}. Let $Z_i$ be the feasible polyhedron of LP \eqref{opt:cs-LP2} in round $i$. We derive the following LP (essentially the one in \cite{dinitz2019lasserre,dinitz2011directed}) which is equivalent to LP \eqref{opt:off-dks-lp}, by considering all the constraints of LP \eqref{opt:cs-LP2} from round 1 to round $k$.

\begin{equation} \label{opt:cs-LP3}
\begin{aligned}
& \min_{x} & & \sum_{e \in E}{x_e} \\
& \text{subject to}
& & \sum_{e \in E}z_e x_e \geq 1 & \forall z \in Z_i \quad \forall i \in [k] ,\\
& & & x_e \geq 0 & \forall e \in E.\\
\end{aligned}
\end{equation}

In round $i \in [k]$, the subroutine that approximately solves LP \eqref{opt:cs-LP2} and checks if the optimum is good or not, is the separation oracle used for solving LP \eqref{opt:cs-LP3} online. Here we use Theorem~\ref{thm:covering} (the formal version of Theorem~\ref{thm:inf-covering}) to show that LP  \eqref{opt:cs-LP3} can be solved online in polynomial time by paying an $O(\log n)$ factor. This requires that both $\log(1/ z_e)$ and $\log \lp^*$ are polynomial in the number of bits used for the edge lengths, where $\lp^*$ is the optimum of LP \eqref{opt:cs-LP3}. Clearly, $\log \lp^* \leq \log |E|$ is in $\poly(n)$. For $\log(1/ z_e)$, the subroutine that approximately solves LP \eqref{opt:cs-LP2} returns an approximate solution $z$ which is represented by polynomial number of bits used for the edge lengths \cite{hassin1992approximation,lorenz2001simple}. By Theorem~\ref{thm:covering}, we have the following Lemma.

\begin{lemma} \label{lem:online-spanner-lp}
There exists a polynomial-time $O(\log n)$-competitive online algorithm for LP \eqref{opt:off-dks-lp}.
\end{lemma}

\paragraph{Conditional edge selection.} After having a fractional solution of LP \eqref{opt:off-dks-lp} in round $i$ where $i \geq T = \lfloor n^{4/5} \rfloor$, we independently pick $e \in E \setminus E'$ with some scaled probability so that the edge selection is consistent with the probability based on the online solution of LP \eqref{opt:off-dks-lp}. More specifically, let $p_e := \min\{1,x_e t \ln n\}$ and let $p^i_e$ be the value of $p_e$ in round $i$. Let $\tilde{E}$ be the set of edges where each edge is neither selected while adding cheapest feasible paths prior to round $T$ nor selected while adding in-arborescences and out-arborescences in round $T$. We show that each edge $e \in \tilde{E}$ has already been selected with probability $p^i_e$ in round $i$. This can be proved by induction. According to Algorithm \ref{alg:dks}, the base case is round $T$, where $e \in \tilde{E}$ is selected with probability $p^T_e$. Now suppose $i > T$, if $e \in \tilde{E}$ has been selected, it is either selected prior to round $i$ or in round $i$. For the former case, $e$ must had already been selected in round $i-1$, with probability $p^{i-1}_e$ by inductive hypothesis. For the later case, conditioned on $e$ has not been selected in round $i-1$, $e$ is selected with probability $(p^i_e-p^{i-1}_e)/(1-p^{i-1}_e)$. Therefore, in round $i$, $e$ has been selected with probability $p^{i-1}_e + (1-p^{i-1}_e) \cdot \frac{p^i_e-p^{i-1}_e}{1-p^{i-1}_e} = p^i_e$, 
which completes the proof. Intuitively, when $i > T$, conditioned on $e \in \tilde{E}$ was not picked from round 1 to round $i-1$, we pick $e$ with probability $(p^i_e-p^{i-1}_e)/(1-p^{i-1}_e)$ at round $i$, so that the overall probability that $e$ is picked from round 1 to round $i$ is $p^i_e$.

\paragraph{Summary.} We conclude the proof as follows. The overall algorithm is given in Algorithm \ref{alg:dks}. For the initialization, $x$ is a zero vector, $E'$ is an empty set, and $T = t = \lfloor n^{4/5} \rfloor$. The set $E'$ is the solution. We pay an extra logarithmic factor for solving LP \eqref{opt:off-dks-lp} online by Lemma~\ref{lem:online-spanner-lp}. The competitive ratio remains $\tilde{O}(n^{4/5})$.
\end{proof}

\subsection{Online pairwise spanners with maximum allowed distance $d$} \label{subsec:maxd}

Suppose the given graph has uniform edge lengths, and the diameter is bounded or it is guaranteed that the distances between the terminal pairs are bounded. Let $d = \max_{i \in [k]} \{d_i\}$ 
be the maximum
allowed distance of any pair of terminals in the input. This setting is equivalent to the \emph{$d$-diameter spanning subgraph} problem introduced in \cite{bhattacharyya2012transitive}. We assume that $d$
is known offline and show the following result.
  
\begin{restatable}{theorem}{thmdus}\label{theorem:unit-spanner}
For the online pairwise spanner
  problem with uniform edge lengths and maximum allowed distance $d$, there
  is a randomized polynomial-time algorithm with competitive ratio
  \begin{math}
    \tO(d^{1/3} n^{2/3}).
  \end{math}
\end{restatable}

\begin{proof}
When the maximum allowed distance is $d$ for the terminal pairs, we employ Algorithm~\ref{alg:dks} with $T = \lfloor d^{-4/3} n^{4/3} \rfloor$ and $t = d^{1/3}n^{2/3}$. 
If $k < T$, then for each $(s_i,t_i) \in D$, we add the edges of a shortest $s_i \leadsto t_i$ path. Each shortest $s_i \leadsto t_i$ path contains at most $d$ edges. By Lemma~\ref{lem:optgeqsqrt}, $\opt \geq \sqrt{k}$, 
so the ratio between this solution and $\opt$ is
$\frac{d k}{\opt} \leq d \sqrt{k} \leq d \sqrt{T} = O(d^{1/3}n^{2/3}).$
If $k \geq T$, then $\opt \geq d^{-2/3}n^{2/3}$. Let $\lp$ be the value of the online integral $O(\log n)$-competitive solution of LP \eqref{opt:off-dks-lp} obtained by Algorithm~\ref{alg:dks}. 
The approximation guarantee is
\[\frac{\tilde{O}(t \lp) + \tilde{O}(n^2/t)}{\opt} \leq \frac{\tO(d^{1/3}n^{2/3} \opt) + \tO(d^{-1/3}n^{4/3})}{\opt}  = \tO(d^{1/3}n^{2/3}).\]
\end{proof}

\subsection{Online quasimetric spanners and all-server spanners} \label{subsec:as}

Consider the case where the edge lengths form a \emph{quasimetric}, i.e., for any two edges $u \to v$ and $v \to w$,
there is also an edge $u \to w$ such that
\begin{math}
  \ell(u,w) 
  \leq 
  \ell(u,v) + \ell(v,w).
\end{math}
This setting includes the class of transitive-closure graphs with uniform edge lengths, in which each pair or vertices connected by a directed path is also connected by a directed edge.
The offline version of the \emph{transitive-closure
  spanner} problem was formally defined in  \cite{bhattacharyya2012transitive}.
We obtain the following result. 

\begin{restatable}{theorem}{thmqs}\label{theorem:quasimetric-spanner}
  For the online pairwise spanner problem where 
  edge lengths are quasimetric, there is a randomized polynomial-time algorithm with
  competitive ratio $\tO(n^{2/3})$.
\end{restatable}

In the setting of all-server spanners, the given graph has uniform edge lengths. For each terminal pair $(s_i,t_i)$, there is also an edge $s_i \to t_i$ in $E$. This setting is equivalent to the \emph{all-server spanner} problem introduced in \cite{elkin1999client}. We show the following theorem for this case.

\begin{restatable}{theorem}{thmasds}\label{thm:asds}
  For the online all-server spanner
  problem with uniform edge lengths, there is a randomized polynomial-time algorithm with
  competitive ratio $\tO(n^{2/3})$.
\end{restatable}

\begin{proof} [Proof of Theorems \ref{theorem:quasimetric-spanner} and \ref{thm:asds}.]
We employ Algorithm~\ref{alg:dks} with $T = \lfloor n^{4/3} \rfloor$ and $t = n^{2/3}$. For these special cases, we start by adding edges instead of paths. Then, in the initial greedy phase in line \ref{line:alg-spanner-greedy} of Algorithm \ref{alg:dks}, the cost of each step is upper-bounded by 1 instead of $\opt$. This enables us to obtain better competitive ratios. More formally, if $k < T$, then for each $(s_i,t_i) \in D$, we add the edge $s_i \to t_i$. By Lemma~\ref{lem:optgeqsqrt}, $\opt \geq \sqrt{k}$, 
so the ratio between this solution and $\opt$ is
\[\frac{k}{\opt} \leq \sqrt{k} \leq \sqrt{T} = O(n^{2/3}).\]
If $k \geq T$, then $\opt \geq n^{2/3}$. Let $\lp$ be the value of the online integral $O(\log n)$-competitive solution of LP \eqref{opt:off-dks-lp} obtained by Algorithm~\ref{alg:dks}. The approximation guarantee is
\[\frac{\tilde{O}(t \lp) + \tilde{O}(n^2/t)}{\opt} \leq \frac{\tO(n^{2/3} \opt) + \tO(n^{4/3})}{\opt}  = \tO(n^{2/3}).\]
\end{proof}

\subsection{Online pairwise spanners with uniform edge lengths} \label{subsec:us}

In this section, we prove Theorem~\ref{thm:us}.

\thmus*

\begin{proof}
We employ Algorithm~\ref{alg:dks} with a slight tweak and set $T = \lfloor n^{4/3 - 4\ep} \rfloor$ and $t = n^{2/3 + \ep}$.

If $k < T$, instead of adding edges of a shortest $s_i \leadsto t_i$ path, we use Theorem~\ref{thm:sqrt-k-ps} to find an $\tO(n^{2/3 + \ep})$ competitive solution.

\thmsqrtkps*

For any $\ep \in (0,1/3)$, there exists $\delta$ such that $4\delta/(9+12\delta) = \ep$. By picking this $\delta$, we have that
\begin{align*}
(\frac{4}{3} - 4\ep)(\frac{1}{2}+\delta) &=
(\frac{4}{3} - \frac{16\delta}{9+12\delta})(\frac{1}{2}+\delta) = \frac{2}{3} + \frac{4\delta}{3} - \frac{8\delta + 16\delta^2}{9 + 12\delta} \\
&= \frac{2}{3} + \frac{12\delta+16\delta^2-8\delta-16\delta^2}{9+12\delta} = \frac{2}{3} + \frac{4\delta}{9+12\delta} = \frac{2}{3} + \ep.
\end{align*}
Hence, the ratio between the solution obtained by Theorem~\ref{thm:sqrt-k-ps} and $\opt$ is
\[k^{1/2 + \delta} \leq n^{(4/3 - 4\ep)(1/2+\delta)} = n^{2/3+\ep} = \tO(n^{2/3+\ep}).\] 

If $k \geq T$, then $\opt \geqslant n^{2/3 - 2\ep}$. Let $\lp$ be the online integral solution of LP \eqref{opt:off-dks-lp} obtained by Algorithm~\ref{alg:dks}. 
The approximation guarantee is
\begin{equation} \label{eq:ds}
\frac{\tilde{O}(t \lp) + \tilde{O}(n^2/t)}{\opt} \leq \frac{\tO(n^{2/3 + \ep} \opt) + \tO(n^{4/3 - \ep})}{\opt}  = \tO(n^{2/3 + \ep}).
\end{equation}
\end{proof}

\subsection{Online directed Steiner forests with uniform costs} \label{sec:dsf}

This problem is a special case of pairwise spanners with uniform edge lengths and infinite target distances. We show the following theorem.

\thmsteiner*

\begin{proof}
The structure of the online algorithm is the same as that for online pairwise spanners with uniform edge lengths. We employ Algorithm \ref{alg:dks} and set $T = \lfloor n^{4/3 - 4\ep} \rfloor$ and $t = n^{2/3 + \ep}$. If $k < T$, instead of adding edges of a shortest $s_i \leadsto t_i$ path, we use Theorem~\ref{thm:sqrt-k-ps} to find an $\tO(n^{2/3 + \ep})$ competitive solution.\footnote{One can also use the $\tO(k^{1/2+\delta})$-competitive online algorithm in \cite{cekp}.} If $k \geq T$, then the algorithm is $\tO(n^{2/3+\ep})$-competitive by \eqref{eq:ds}.
\end{proof}

\section{Online Covering in Polynomial Time} \label{sec:covering}

This section is devoted to proving the formal version of Theorem~\ref{thm:inf-covering}. We recall that the problem of interest is to solve the covering LP \eqref{equation:covering} online:
\begin{align*}
  \begin{aligned}
    \text{minimize } & \rip{\mathbf{c}}{x} 
    \text{ over } x \in \nnreals^n 
    \text{ s.t.\ } A x \geq \mathbf{1}
  \end{aligned}
\end{align*}
where $A \in \R_{\geq 0}^{m \times n}$ consists of $m$ covering
constraints, $\mathbf{1} \in \R_{> 0}^m$ is a vector of all ones treated as the lower bound of the covering constraints, and $\mathbf{c} \in \R_{> 0}^n$ denotes the positive coefficients of the linear cost function.

In the online covering problem, the cost vector $\mathbf{c}$ is given offline, and each of these covering constraints is presented one by one in an online fashion, that is, $m$ can be unknown. In round $i \in [m]$, $\{a_{ij}\}_{j \in [n]}$ (where $a_{ij}$ denotes the $i$-th row $j$-th column entry of $A$) is revealed, and we have to monotonically update $x$ so that the constraint $\sum_{j \in [n]}{a_{ij}x_j} \geq 1$ is satisfied. We always assume that there is at least one positive entry $a_{ij}$ in each round $i$, otherwise constraint $i$ cannot be satisfied since all the row entries are zeros. The goal is to update $x$ in a non-decreasing manner and approximately minimize the objective value $\rip{\mathbf{c}}{x}$.

We recall that an important idea in this
line of work is to simultaneously consider the dual packing problem LP \eqref{equation:packing}:
\begin{align*}
  \begin{aligned}
    \text{maximize } & \rip{\mathbf{1}}{y} 
    \text{ over } y \in \nnreals^m 
    \text{ s.t.\ } A^T y \leq \mathbf{c}
  \end{aligned}
\end{align*}
where $A^T$ consists of $n$ packing constraints with an upper bound $\mathbf{c}$ given offline.

The primal-dual framework in \cite{buchbinder2009online} simultaneously solves
both LP \eqref{equation:covering} and LP \eqref{equation:packing}, and crucially
uses LP-duality and strong connections between the two solutions to
argue that they are both nearly optimal. The modified framework closely follows the \emph{guess-and-double} scheme in \cite{buchbinder2009online}. Specifically, the scheme runs in phases where each phase estimates a lower bound for the optimum. When the first constraint arrives, the scheme generates the first lower bound
\[\alpha(1) \gets \min_{j \in [n]}\{\frac{c_j}{a_{1j}} \mid a_{1j} > 0\} \leq \opt\]
where $c_j$ is the $j$-th entry of $\mathbf{c}$ and $\opt$ is the optimal value of LP \eqref{equation:covering}.

During phase $r$, we always assume that the lower bound of the optimum is $\alpha(r)$ until the online objective $\rip{\mathbf{c}}{x}$ exceeds $\alpha(r)$. Once the online objective exceeds $\alpha(r)$, we start the new phase $r+1$ from the current violating constraint (let us call it constraint $i_{r+1}$, in particular, $i_1 = 1$), and double the estimated lower bound, i.e., $\alpha(r+1) \gets 2\alpha(r)$.\footnote{In \cite{buchbinder2009online}, the scheme starts all over again from the first constraint. We start from the current violating constraint because it is more amenable when violating constraints are generated by a separation oracle. There is no guarantee for the order of arriving violating constraints in such settings.} We recall that $x$ must be updated in a non-decreasing manner, so the algorithm maintains $\{x^r_j\}$, which denotes the value of each variable $x_j$ in each phase $r$, and the value of each variable $x_j$ is actually set to $\max_r\{x^r_j\}$.

In Algorithm~\ref{alg:covering}, we describe one round of the modified scheme in phase $r$. When a covering constraint $i$ arrives, we introduce a packing variable $y_i=0$. If the constraint is violated, we increment each $x_j$ according to an exponential function of $y_i$ until the constraint is satisfied \emph{by a factor of 2}. This is the main difference between the modified framework and \cite{buchbinder2009online}, which increments the variables until the constraint is satisfied.

\begin{algorithm}[H]
\caption{Online Covering} \label{alg:covering}

\begin{algorithmic}[1]

\For {arriving covering constraint $i$}
    \State $y_i \gets 0$. \Comment{the packing variable $y_i$ is used for the analysis}
    \If {$\sum_{j=1}^n{a_{ij}x^r_j} < 1$} \Comment{if constraint $i$ is not satisfied}
        \While {$\sum_{j=1}^n a_{ij}x_j^r < 2$} \Comment{update until constraint $i$ is satisfied by a factor of 2}
            \State Increase $y_i$ continuously.
            \State Increase each variable $x^r_j$ by the following increment function:
            \[x^r_j \gets \frac{\alpha(r)}{2nc_j}\exp \tuple{\frac{\ln (2n)}{c_j}\sum_{k=i_r}^{i}a_{kj}y_k}.\]
        \EndWhile
    \EndIf
\EndFor
\end{algorithmic}
\end{algorithm}

Although the augmentation is in a continuous fashion, it is not hard to implement it in a discrete way for any desired precision by binary search. Therefore, to show that the modified framework is efficient, it suffices to bound the number of violating constraints it will encounter. The performance of the modified scheme is analyzed in Theorem~\ref{thm:covering} (the formal version of Theorem~\ref{thm:inf-covering}).

\begin{theorem} \label{thm:covering}
There exists an $O(\log n)$-competitive online algorithm for the covering LP \eqref{equation:covering} which encounters $\poly(n, \log \opt, \log (1/\alpha(1)))$ violating constraints.
\end{theorem}

\begin{proof}
The proof for the $O(\log n)$-competitiveness closely follows the one in \cite{buchbinder2009online}. Let $X(r)$ and $Y(r)$ be the covering and packing objective values, respectively, generated during phase $r$. The following claims are used to show that Algorithm~\ref{alg:covering} is $O(\log n)$-competitive. For the sake of completeness, we provide the proof in Appendix \ref{sec:pf-covering}.
\begin{enumerate}[label=(\roman*)]
    \item $x$ is feasible.
    \item For each \emph{finished} phase $r$, $\alpha(r) \leq 4 \ln (2n) \cdot Y(r)$.
    \item $y$ generated during phase $r$ is feasible.
    \item The sum of the covering objective generated from phase 1 to $r$ is at most $2 \alpha(r)$.
    \item Let $r'$ be the last phase, then the covering objective $\rip{\mathbf{c}}{x} \leq 2 \alpha(r')$.
\end{enumerate}
From these five claims together with weak duality, we conclude that
\[\rip{\mathbf{c}}{x} \leq 2 \alpha(r') = 4 \alpha(r'-1) \leq 16 \ln(2n) \cdot Y(r'-1) \leq 16 \ln(2n) \cdot \opt.\]

Now we show that Algorithm~\ref{alg:covering} encounters $\poly(n, \log \opt, \log (1/\alpha(1)))$ violating constraints. 
We first show that there are $O(\log\log n + \log \opt + \log (1/\alpha(1)))$ phases. The estimated lower bound $\alpha$ doubles when we start a new phase. Suppose there are $r'$ phases, then $\alpha(1) \cdot 2^{r'-1} = O(\log n) \opt$ because Algorithm~\ref{alg:covering} is $O(\log n)$-competitive. This implies that $r' = O(\log 
\log n + \log \opt + \log (1/\alpha(1)))$.

In each phase, when a violating constraint arrives, we increment $x$ so that the constraint is satisfied by a factor of $2$. This implies that at least one variable $x_j$ is doubled. $x_j = O(\log n) \opt / c_j$ because $c_j x_j \leq \rip{\mathbf{c}}{x} = O(\log n) \opt$. At the start of phase $r$, $x_j = \alpha(r)/(2nc_j) \geq \alpha(1)/(2nc_j)$. Suppose $x_j$ has been doubled $t$ times in phase $r$, then
\[\frac{\alpha(1)}{2nc_j} \cdot 2^t \leq \frac{\alpha(r)}{2nc_j} \cdot 2^t \leq x_j = O(\log n) \frac{\opt}{c_j}\]
which indicates that $t = O(\log n + \log \opt + \log (1/\alpha(1)))$.

There are $n$ variables and $r'$ phases, and in each phase, each variable is doubled at most $t$ times. Therefore, Algorithm~\ref{alg:covering} encounters $\poly(n, \log \opt, \log (1/\alpha(1)))$ violating constraints.
\end{proof}

As stated in \cite{buchbinder2009online}, the online scheme naturally extends to the setting with unbounded number of constraints where covering constraints do not appear explicitly, but are detected by a separation oracle. However, in the previous work, there was no guarantee about the efficiency, i.e., how fast the algorithm reaches an approximately optimal solution. In the worst case scenario, it is possible that a detected constraint is slightly violated, thus hindering the growth of the covering variables since they are incremented until the constraint is \emph{just} satisfied. The modified framework, on the other hand, increments the variables until the constraint is satisfied by a factor of 2. The sufficient growth of the covering variables ensures a polynomial upper bound for the number of violating constraints that the framework will encounter.

\section{Conclusions and Open Problems} \label{sec:conclusion}

In this work, we present the first online algorithm for pairwise spanners with competitive ratio $\tO(n^{4/5})$ for general lengths and $\tO(n^{2/3 + \ep})$ for uniform lengths. We also improve the competitive ratio for the online directed Steiner forest problem with uniform costs to $\tO(n^{2/3+\ep})$ when $k=\omega(n^{4/3})$. We also show an efficient modified framework for online covering and packing. Our work raises several open questions that we state below.

\paragraph{Online pairwise spanners.} An intriguing open problem is improving the competitive ratio for online pairwise spanners. For graphs with uniform edge lengths, there is a small polynomial gap between the state-of-the-art offline approximation ratio $\tO(n^{3/5 + \ep})$ and the online competitive ratio $\tO(n^{2/3 + \ep})$. For graphs with general edge lengths, we are not aware of any studies about the pairwise spanner problem. Our $\tO(n^{4/5})$-competitive online algorithm intrinsically suggests an $\tO(n^{4/5})$-approximate offline algorithm. As the approach in \cite{chlamtavc2020approximating} achieves an $\tO(n/\sqrt{\opt})$-approximation, we believe that the approximation ratio can be improved for the offline pairwise spanner problem, by judicious case analysis according to the cardinality of $\opt$.

\paragraph{Approximating Steiner forests online in terms of $n$.} The state-of-the-art online algorithm for Steiner forests with general costs is $\tO(k^{1/2 + \ep})$-competitive \cite{cekp}. A natural open question is designing an $o(n)$-competitive online algorithm when $k$ is large, and potentially extend this result to the more general buy-at-bulk network design problem. The currently best known offline approximation for Steiner forests with general costs is $O(n^{2/3 + \ep})$ \cite{berman2013approximation}, by case analysis that settles thick and thin terminal pairs separately. However, the approach in \cite{berman2013approximation} for settling thin pairs is essentially a greedy procedure which is inherently offline. Our approach utilizes the uniform cost assumption to obtain a useful lower bound for the optimal solution, which is incompatible with general costs. It would be interesting to resolve the aforementioned obstacles and have an $o(n)$-competitive online algorithm for directed Steiner forests with general edge costs. One open problem for uniform costs is to improve the competitive ratio, as there is a polynomial gap between the state-of-the-art offline approximation ratio $\tO(n^{4/7 + \ep})$ and the online competitive ratio $\tO(n^{2/3 + \ep})$.

\section{Acknowledgements}
We thank the anonymous reviewers for comments and suggestions that helped improve the presentation. We thank Anupam Gupta and Greg Bodwin for bringing to our attention references that we missed in previous versions of the write-up.
\bibliographystyle{acm}
\bibliography{reference}

\begin{thebibliography}{10}

\bibitem{abboud2018reachability}
{\sc Abboud, A., and Bodwin, G.}
\newblock Reachability preservers: New extremal bounds and approximation
  algorithms.
\newblock In {\em Proceedings of the Twenty-Ninth Annual ACM-SIAM Symposium on
  Discrete Algorithms\/} (2018), SIAM, pp.~1865--1883.

\bibitem{ahmed2019graph}
{\sc Ahmed, R., Bodwin, G., Sahneh, F.~D., Hamm, K., Jebelli, M. J.~L.,
  Kobourov, S., and Spence, R.}
\newblock Graph spanners: A tutorial review, 2019.

\bibitem{alon2006general}
{\sc Alon, N., Awerbuch, B., Azar, Y., Buchbinder, N., and Naor, J.}
\newblock A general approach to online network optimization problems.
\newblock {\em ACM Transactions on Algorithms (TALG) 2}, 4 (2006), 640--660.

\bibitem{aaabn-set-cover}
{\sc Alon, N., Awerbuch, B., Azar, Y., Buchbinder, N., and Naor, J.}
\newblock The online set cover problem.
\newblock {\em {SIAM} J. Comput. 39}, 2 (2009), 361--370.

\bibitem{alon1987optimalpreprocessing}
{\sc Alon, N., and Schieber, B.}
\newblock {\em Optimal preprocessing for answering on-line product queries}.
\newblock Citeseer, 1987.

\bibitem{antonakopoulos2010approximating}
{\sc Antonakopoulos, S.}
\newblock Approximating directed buy-at-bulk network design.
\newblock In {\em International Workshop on Approximation and Online
  Algorithms\/} (2010), Springer, pp.~13--24.

\bibitem{arkin1991bicriteria}
{\sc Arkin, E.~M., Mitchell, J.~S., and Piatko, C.~D.}
\newblock Bicriteria shortest path problems in the plane.
\newblock In {\em Proc. 3rd Canad. Conf. Comput. Geom\/} (1991), Citeseer,
  pp.~153--156.

\bibitem{AwasthiJMR16}
{\sc Awasthi, P., Jha, M., Molinaro, M., and Raskhodnikova, S.}
\newblock Testing lipschitz functions on hypergrid domains.
\newblock {\em Algorithmica 74}, 3 (2016), 1055--1081.

\bibitem{Awerbuch}
{\sc Awerbuch, B.}
\newblock Communication-time trade-offs in network synchronization.
\newblock In {\em Proceedings of the Fourth Annual ACM Symposium on Principles
  of Distributed Computing\/} (New York, NY, USA, 1985), PODC ’85,
  Association for Computing Machinery, p.~272–276.

\bibitem{awerbuch1997buy}
{\sc Awerbuch, B., and Azar, Y.}
\newblock Buy-at-bulk network design.
\newblock In {\em Proceedings 38th Annual Symposium on Foundations of Computer
  Science\/} (1997), IEEE, pp.~542--547.

\bibitem{awerbuch2004line}
{\sc Awerbuch, B., Azar, Y., and Bartal, Y.}
\newblock On-line generalized steiner problem.
\newblock {\em Theoretical Computer Science 324}, 2-3 (2004), 313--324.

\bibitem{awerbuch1993throughput}
{\sc Awerbuch, B., Azar, Y., and Plotkin, S.}
\newblock Throughput-competitive on-line routing.
\newblock In {\em Proceedings of 1993 IEEE 34th Annual Foundations of Computer
  Science\/} (1993), IEEE, pp.~32--40.

\bibitem{azar2013online}
{\sc Azar, Y., Bhaskar, U., Fleischer, L., and Panigrahi, D.}
\newblock Online mixed packing and covering.
\newblock In {\em Proceedings of the twenty-fourth annual ACM-SIAM symposium on
  Discrete algorithms\/} (2013), SIAM, pp.~85--100.

\bibitem{azar2016online}
{\sc Azar, Y., Buchbinder, N., Chan, T.~H., Chen, S., Cohen, I.~R., Gupta, A.,
  Huang, Z., Kang, N., Nagarajan, V., Naor, J., et~al.}
\newblock Online algorithms for covering and packing problems with convex
  objectives.
\newblock In {\em 2016 IEEE 57th Annual Symposium on Foundations of Computer
  Science (FOCS)\/} (2016), IEEE, pp.~148--157.

\bibitem{bansal2008randomized}
{\sc Bansal, N., Buchbinder, N., and Naor, J.}
\newblock Randomized competitive algorithms for generalized caching.
\newblock In {\em Proceedings of the fortieth annual ACM symposium on Theory of
  computing\/} (2008), pp.~235--244.

\bibitem{bansal2012primal}
{\sc Bansal, N., Buchbinder, N., and Naor, J.}
\newblock A primal-dual randomized algorithm for weighted paging.
\newblock {\em Journal of the ACM (JACM) 59}, 4 (2012), 1--24.

\bibitem{Baswana_streamingalgorithm}
{\sc Baswana, S.}
\newblock Streaming algorithm for graph spanners - single pass and constant
  processing time per edge.
\newblock {\em Inf. Process. Lett\/} (2008).

\bibitem{BaswanaK10}
{\sc Baswana, S., and Kavitha, T.}
\newblock Faster algorithms for all-pairs approximate shortest paths in
  undirected graphs.
\newblock {\em {SIAM} J. Comput. 39}, 7 (2010), 2865--2896.

\bibitem{bateni2012euclidean}
{\sc Bateni, M., and Hajiaghayi, M.}
\newblock Euclidean prize-collecting steiner forest.
\newblock {\em Algorithmica 62}, 3-4 (2012), 906--929.

\bibitem{berman2013approximation}
{\sc Berman, P., Bhattacharyya, A., Makarychev, K., Raskhodnikova, S., and
  Yaroslavtsev, G.}
\newblock Approximation algorithms for spanner problems and directed steiner
  forest.
\newblock {\em Information and Computation 222\/} (2013), 93--107.

\bibitem{berman1997line}
{\sc Berman, P., and Coulston, C.}
\newblock On-line algorithms for steiner tree problems.
\newblock In {\em Proceedings of the twenty-ninth annual ACM symposium on
  Theory of computing\/} (1997), pp.~344--353.

\bibitem{bhattacharyya2012transitive}
{\sc Bhattacharyya, A., Grigorescu, E., Jung, K., Raskhodnikova, S., and
  Woodruff, D.~P.}
\newblock Transitive-closure spanners.
\newblock {\em SIAM Journal on Computing 41}, 6 (2012), 1380--1425.

\bibitem{BodwinW16}
{\sc Bodwin, G., and Williams, V.~V.}
\newblock Better distance preservers and additive spanners.
\newblock In {\em SODA\/} (2016), R.~Krauthgamer, Ed., {SIAM}, pp.~855--872.

\bibitem{borradaile2015polynomial}
{\sc Borradaile, G., Klein, P.~N., and Mathieu, C.}
\newblock A polynomial-time approximation scheme for euclidean steiner forest.
\newblock {\em ACM Transactions on Algorithms (TALG) 11}, 3 (2015), 1--20.

\bibitem{buchbinder2006improved}
{\sc Buchbinder, N., and Naor, J.}
\newblock Improved bounds for online routing and packing via a primal-dual
  approach.
\newblock In {\em 2006 47th Annual IEEE Symposium on Foundations of Computer
  Science (FOCS'06)\/} (2006), IEEE, pp.~293--304.

\bibitem{buchbinder2009online}
{\sc Buchbinder, N., and Naor, J.}
\newblock Online primal-dual algorithms for covering and packing.
\newblock {\em Mathematics of Operations Research 34}, 2 (2009), 270--286.

\bibitem{buchbinder2009design}
{\sc Buchbinder, N., and Naor, J.~S.}
\newblock The design of competitive online algorithms via a primal--dual
  approach.
\newblock {\em Foundations and Trends{\textregistered} in Theoretical Computer
  Science 3}, 2--3 (2009), 93--263.

\bibitem{cekp}
{\sc Chakrabarty, D., Ene, A., Krishnaswamy, R., and Panigrahi, D.}
\newblock Online buy-at-bulk network design.
\newblock {\em {SIAM} J. Comput. 47}, 4 (2018), 1505--1528.

\bibitem{charikar1999approximation}
{\sc Charikar, M., Chekuri, C., Cheung, T.-y., Dai, Z., Goel, A., Guha, S., and
  Li, M.}
\newblock Approximation algorithms for directed steiner problems.
\newblock {\em Journal of Algorithms 33}, 1 (1999), 73--91.

\bibitem{chawla2006optimal}
{\sc Chawla, S., Roughgarden, T., and Sundararajan, M.}
\newblock Optimal cost-sharing mechanisms for steiner forest problems.
\newblock In {\em International Workshop on Internet and Network Economics\/}
  (2006), Springer, pp.~112--123.

\bibitem{Chechik15}
{\sc Chechik, S.}
\newblock Approximate distance oracles with improved bounds.
\newblock In {\em STOC\/} (2015), R.~A. Servedio and R.~Rubinfeld, Eds., {ACM},
  pp.~1--10.

\bibitem{chekuri2011set}
{\sc Chekuri, C., Even, G., Gupta, A., and Segev, D.}
\newblock Set connectivity problems in undirected graphs and the directed
  steiner network problem.
\newblock {\em ACM Transactions on Algorithms (TALG) 7}, 2 (2011), 1--17.

\bibitem{chekuri2010approximation}
{\sc Chekuri, C., Hajiaghayi, M.~T., Kortsarz, G., and Salavatipour, M.~R.}
\newblock Approximation algorithms for nonuniform buy-at-bulk network design.
\newblock {\em SIAM Journal on Computing 39}, 5 (2010), 1772--1798.

\bibitem{chlamtavc2020approximating}
{\sc Chlamt{\'a}{\v{c}}, E., Dinitz, M., Kortsarz, G., and Laekhanukit, B.}
\newblock Approximating spanners and directed steiner forest: Upper and lower
  bounds.
\newblock {\em ACM Transactions on Algorithms (TALG) 16}, 3 (2020), 1--31.

\bibitem{CowenW04}
{\sc Cowen, L., and Wagner, C.~G.}
\newblock Compact roundtrip routing in directed networks.
\newblock {\em J. Algorithms 50}, 1 (2004), 79--95.

\bibitem{DerbelGP07}
{\sc Derbel, B., Gavoille, C., and Peleg, D.}
\newblock Deterministic distributed construction of linear stretch spanners in
  polylogarithmic time.
\newblock In {\em DISC\/} (2007), A.~Pelc, Ed., vol.~4731 of {\em Lecture Notes
  in Computer Science}, Springer, pp.~179--192.

\bibitem{DerbelGPV08}
{\sc Derbel, B., Gavoille, C., Peleg, D., and Viennot, L.}
\newblock On the locality of distributed sparse spanner construction.
\newblock In {\em PODC\/} (2008), R.~A. Bazzi and B.~Patt{-}Shamir, Eds.,
  {ACM}, pp.~273--282.

\bibitem{dinitz2011directed}
{\sc Dinitz, M., and Krauthgamer, R.}
\newblock Directed spanners via flow-based linear programs.
\newblock In {\em STOC\/} (2011), pp.~323--332.

\bibitem{dinitz2019lasserre}
{\sc Dinitz, M., Nazari, Y., and Zhang, Z.}
\newblock Lasserre integrality gaps for graph spanners and related problems.
\newblock {\em arXiv preprint arXiv:1905.07468\/} (2019).

\bibitem{dinitz2016approximating}
{\sc Dinitz, M., and Zhang, Z.}
\newblock Approximating low-stretch spanners.
\newblock In {\em Proceedings of the twenty-seventh annual ACM-SIAM symposium
  on Discrete algorithms\/} (2016), SIAM, pp.~821--840.

\bibitem{DorHZ00}
{\sc Dor, D., Halperin, S., and Zwick, U.}
\newblock All-pairs almost shortest paths.
\newblock {\em {SIAM} J. Comput. 29}, 5 (2000), 1740--1759.

\bibitem{Elkin05}
{\sc Elkin, M.}
\newblock Computing almost shortest paths.
\newblock {\em {ACM} Trans. Algorithms 1}, 2 (2005), 283--323.

\bibitem{Elkin11}
{\sc Elkin, M.}
\newblock Streaming and fully dynamic centralized algorithms for constructing
  and maintaining sparse spanners.
\newblock {\em {ACM} Trans. Algorithms 7}, 2 (2011), 20:1--20:17.

\bibitem{elkin1999client}
{\sc Elkin, M., and Peleg, D.}
\newblock The client-server 2-spanner problem with applications to network
  design.
\newblock In {\em {SIROCCO} 8, Proceedings of the 8th International Colloquium
  on Structural Information and Communication Complexity, Vall de N{\'{u}}ria,
  Girona-Barcelona, Catalonia, Spain, 27-29 June, 2001\/} (2001), F.~Comellas,
  J.~F{\`{a}}brega, and P.~Fraigniaud, Eds., vol.~8 of {\em Proceedings in
  Informatics}, Carleton Scientific, pp.~117--132.

\bibitem{ElkinP07}
{\sc Elkin, M., and Peleg, D.}
\newblock The hardness of approximating spanner problems.
\newblock {\em Theory Comput. Syst. 41}, 4 (2007), 691--729.

\bibitem{feldman2012improved}
{\sc Feldman, M., Kortsarz, G., and Nutov, Z.}
\newblock Improved approximation algorithms for directed steiner forest.
\newblock {\em Journal of Computer and System Sciences 78}, 1 (2012), 279--292.

\bibitem{FernandezW020}
{\sc Fernandez, M., Woodruff, D.~P., and Yasuda, T.}
\newblock Graph spanners in the message-passing model.
\newblock In {\em ITCS\/} (2020), T.~Vidick, Ed., vol.~151 of {\em LIPIcs},
  Schloss Dagstuhl - Leibniz-Zentrum f{\"{u}}r Informatik, pp.~77:1--77:18.

\bibitem{filtser2020graph}
{\sc Filtser, A., Kapralov, M., and Nouri, N.}
\newblock Graph spanners by sketching in dynamic streams and the simultaneous
  communication model.
\newblock {\em arXiv preprint arXiv:2007.14204\/} (2020).

\bibitem{fleischer2006simple}
{\sc Fleischer, L., K{\"o}nemann, J., Leonardi, S., and Sch{\"a}fer, G.}
\newblock Simple cost sharing schemes for multicommodity rent-or-buy and
  stochastic steiner tree.
\newblock In {\em Proceedings of the thirty-eighth annual ACM symposium on
  Theory of computing\/} (2006), pp.~663--670.

\bibitem{garg2000polylogarithmic}
{\sc Garg, N., Konjevod, G., and Ravi, R.}
\newblock A polylogarithmic approximation algorithm for the group steiner tree
  problem.
\newblock {\em Journal of Algorithms 37}, 1 (2000), 66--84.

\bibitem{GoemansW95}
{\sc Goemans, M.~X., and Williamson, D.~P.}
\newblock A general approximation technique for constrained forest problems.
\newblock {\em {SIAM} J. Comput. 24}, 2 (1995), 296--317.

\bibitem{gupta2003approximation}
{\sc Gupta, A., Kumar, A., P{\'a}l, M., and Roughgarden, T.}
\newblock Approximation via cost-sharing: a simple approximation algorithm for
  the multicommodity rent-or-buy problem.
\newblock In {\em 44th Annual IEEE Symposium on Foundations of Computer
  Science, 2003. Proceedings.\/} (2003), IEEE, pp.~606--615.

\bibitem{gupta2014approximating}
{\sc Gupta, A., and Nagarajan, V.}
\newblock Approximating sparse covering integer programs online.
\newblock {\em Mathematics of Operations Research 39}, 4 (2014), 998--1011.

\bibitem{gupta2017last}
{\sc Gupta, A., Ravi, R., Talwar, K., and Umboh, S.~W.}
\newblock Last but not least: Online spanners for buy-at-bulk.
\newblock In {\em Proceedings of the Twenty-Eighth Annual ACM-SIAM Symposium on
  Discrete Algorithms\/} (2017), SIAM, pp.~589--599.

\bibitem{gupta2014changing}
{\sc Gupta, A., Talwar, K., and Wieder, U.}
\newblock Changing bases: Multistage optimization for matroids and matchings.
\newblock In {\em International Colloquium on Automata, Languages, and
  Programming\/} (2014), Springer, pp.~563--575.

\bibitem{hassin1992approximation}
{\sc Hassin, R.}
\newblock Approximation schemes for the restricted shortest path problem.
\newblock {\em Mathematics of Operations research 17}, 1 (1992), 36--42.

\bibitem{helvig2001improved}
{\sc Helvig, C.~S., Robins, G., and Zelikovsky, A.}
\newblock An improved approximation scheme for the group steiner problem.
\newblock {\em Networks: An International Journal 37}, 1 (2001), 8--20.

\bibitem{imase1991dynamic}
{\sc Imase, M., and Waxman, B.~M.}
\newblock Dynamic steiner tree problem.
\newblock {\em SIAM Journal on Discrete Mathematics 4}, 3 (1991), 369--384.

\bibitem{KapralovW14}
{\sc Kapralov, M., and Woodruff, D.~P.}
\newblock Spanners and sparsifiers in dynamic streams.
\newblock In {\em PODC\/} (2014), M.~M. Halld{\'{o}}rsson and S.~Dolev, Eds.,
  {ACM}, pp.~272--281.

\bibitem{kmmo}
{\sc Karlin, A.~R., Manasse, M.~S., McGeoch, L.~A., and Owicki, S.~S.}
\newblock Competitive randomized algorithms for non-uniform problems.
\newblock {\em Algorithmica 11}, 6 (1994), 542--571.

\bibitem{khurana2017genome}
{\sc Khurana, V., Peng, J., Chung, C.~Y., Auluck, P.~K., Fanning, S., Tardiff,
  D.~F., Bartels, T., Koeva, M., Eichhorn, S.~W., Benyamini, H., et~al.}
\newblock Genome-scale networks link neurodegenerative disease genes to
  $\alpha$-synuclein through specific molecular pathways.
\newblock {\em Cell systems 4}, 2 (2017), 157--170.

\bibitem{konemann2005primal}
{\sc K{\"o}nemann, J., Leonardi, S., Sch{\"a}fer, G., and van Zwam, S.}
\newblock From primal-dual to cost shares and back: a stronger lp relaxation
  for the steiner forest problem.
\newblock In {\em International Colloquium on Automata, Languages, and
  Programming\/} (2005), Springer, pp.~930--942.

\bibitem{konemann2008group}
{\sc K{\"o}nemann, J., Leonardi, S., Sch{\"a}fer, G., and van Zwam, S.~H.}
\newblock A group-strategyproof cost sharing mechanism for the steiner forest
  game.
\newblock {\em SIAM Journal on Computing 37}, 5 (2008), 1319--1341.

\bibitem{Kortsarz2001OnTH}
{\sc Kortsarz, G.}
\newblock On the hardness of approximating spanners.
\newblock {\em Algorithmica 30\/} (2001), 432--450.

\bibitem{lorenz2001simple}
{\sc Lorenz, D.~H., and Raz, D.}
\newblock A simple efficient approximation scheme for the restricted shortest
  path problem.
\newblock {\em Operations Research Letters 28}, 5 (2001), 213--219.

\bibitem{PachockiRSTW18}
{\sc Pachocki, J., Roditty, L., Sidford, A., Tov, R., and Williams, V.~V.}
\newblock Approximating cycles in directed graphs: Fast algorithms for girth
  and roundtrip spanners.
\newblock In {\em SODA\/} (2018), A.~Czumaj, Ed., {SIAM}, pp.~1374--1392.

\bibitem{PatrascuR14}
{\sc Patrascu, M., and Roditty, L.}
\newblock Distance oracles beyond the {T}horup-{Z}wick bound.
\newblock {\em {SIAM} J. Comput. 43}, 1 (2014), 300--311.

\bibitem{PelegS89}
{\sc Peleg, D., and Sch{\"{a}}ffer, A.~A.}
\newblock Graph spanners.
\newblock {\em Journal of Graph Theory 13}, 1 (1989), 99--116.

\bibitem{PelegU89a}
{\sc Peleg, D., and Ullman, J.~D.}
\newblock An optimal synchronizer for the hypercube.
\newblock {\em {SIAM} J. Comput. 18}, 4 (1989), 740--747.

\bibitem{pirhaji2016revealing}
{\sc Pirhaji, L., Milani, P., Leidl, M., Curran, T., Avila-Pacheco, J., Clish,
  C.~B., White, F.~M., Saghatelian, A., and Fraenkel, E.}
\newblock Revealing disease-associated pathways by network integration of
  untargeted metabolomics.
\newblock {\em Nature methods 13}, 9 (2016), 770--776.

\bibitem{RodittyTZ08}
{\sc Roditty, L., Thorup, M., and Zwick, U.}
\newblock Roundtrip spanners and roundtrip routing in directed graphs.
\newblock {\em {ACM} Trans. Algorithms 4}, 3 (2008), 29:1--29:17.

\bibitem{roughgarden2007optimal}
{\sc Roughgarden, T., and Sundararajan, M.}
\newblock Optimal efficiency guarantees for network design mechanisms.
\newblock In {\em International Conference on Integer Programming and
  Combinatorial Optimization\/} (2007), Springer, pp.~469--483.

\bibitem{shen2020online}
{\sc Shen, X., and Nagarajan, V.}
\newblock Online covering with $l_q$-norm objectives and applications to
  network design.
\newblock {\em Mathematical Programming 184\/} (2020).

\bibitem{Yao1982SpacetimeTF}
{\sc Yao, A. C.-C.}
\newblock Space-time tradeoff for answering range queries (extended abstract).
\newblock In {\em STOC '82\/} (1982).

\bibitem{young1994thek}
{\sc Young, N.}
\newblock The $k$-server dual and loose competitiveness for paging.
\newblock {\em Algorithmica 11}, 6 (1994), 525--541.

\bibitem{zelikovsky1997series}
{\sc Zelikovsky, A.}
\newblock A series of approximation algorithms for the acyclic directed steiner
  tree problem.
\newblock {\em Algorithmica 18}, 1 (1997), 99--110.

\end{thebibliography}

\appendix
\section{Missing Proofs in Section~\ref{sec:us}}

\subsection{Missing proof for Lemma~\ref{lem:sqrt-k-den}} \label{pf:lem:sqrt-k-den}

\lemsqrtkden*

\begin{proof}
Let $G^*$ (a subgraph of $G$) be the optimal pairwise spanner solution with $\opt$ edges. The proof proceeds by considering the following two cases: 1) there exists a vertex $r \in V$ that belongs to at least $\sqrt{k}$ $s_i \leadsto t_i$ paths of distance at most $d_i$ in $G^*$ for distinct $i$, and 2) there is no such vertex $r \in V$.

For the first case, we consider the union of the $s_i \leadsto t_i$ paths in $G^*$, each of distance at most $d_i$, that passes through $r$. This subgraph in $G^*$ contains an in-arborescence and an out-arborescence both rooted at $r$, whose union forms a junction tree. This junction tree has at most $\opt$ edges and connects at least $\sqrt{k}$ terminal pairs, so its density is at most $\opt / \sqrt{k}$.

For the second case, each vertex $r \in V$ appears in at most $\sqrt{k}$ $s_i \leadsto t_i$ paths in $G^*$. More specifically, each edge $e \in E$ also appears in at most $\sqrt{k}$ $s_i \leadsto t_i$ paths in $G^*$. By creating $\sqrt{k}$ copies of each edge, all terminal pairs can be connected by edge-disjoint paths. Since the overall duplicate cost is at most $\sqrt{k} \cdot \opt$, at least one of these paths has cost at most $\sqrt{k} \cdot \opt / k $. This path constitutes a junction tree whose density is at most $\opt / \sqrt{k}$.
\end{proof}

\subsection{Missing proof in Lemma~\ref{lem:rep-set}} \label{sec:pf-rep-set}
\begin{lemma} \label{lem:rep-set-single}
The cross-product of the representative sets $\tilde{S}^r_i$ and $\tilde{T}^r_i$ is a subset of $\hat{R}^r_i$, i.e.,
\[\tilde{S}^r_i \times \tilde{T}^r_i \subseteq \hat{R}^r_i.\]
\end{lemma}

\begin{proof}
Let $U^r_i := \{\hat{s} \mid \hat{s} \preceq_{i,r} \hat{s}_{i,r}\}$ be the prefix set according to $\prec_{i,r}$. By the choice of $\hat{s}_{i,r}$ and the definition of $\gamma^r_i$, we have that
\[\sum_{\hat{s} \in U^r_i}{\sum_{\hat{t} \mid (\hat{s}, \hat{t}) \in \hat{R}^r_i}{y_{\hat{s},\hat{t}}}} = \gamma^r_i - \sum_{\hat{s} \succ_{i,r} \hat{s}_{i,r} }{\sum_{\hat{t} \mid (\hat{s}, \hat{t}) \in \hat{R}^r_i}{y_{\hat{s},\hat{t}}}} > \gamma^r_i - \frac{\gamma^r_i}{2} = \frac{\gamma^r_i}{2}.\]
Let $\hat{t}_{\max} := \max_{\prec_{i,r}}\{\hat{t} \mid (\hat{s}_{i,r}, \hat{t}) \in \hat{R}^r_i\}$.
We have that for any $\hat{s} \in U^r_i$, if $\hat{t}$ is such that $(\hat{s},\hat{t}) \in \hat{R}^r_i$, then $\hat{t} \preceq_{i,r} \hat{t}_{\max}$. Otherwise, suppose $\hat{t} \succ_{i,r} \hat{t}_{\max}$, then by Claim~\ref{cl:ord}, $(\hat{s}, \hat{t}) \in \hat{R}^r_i$ implies that $(\hat{s}_{i,r}, \hat{t}) \in \hat{R}^r_i$, this contradicts to the definition of $\hat{t}_{\max}$. We have that
\[\sum_{\hat{t} \preceq_{i,r} \hat{t}_{\max}}{\sum_{\hat{s} \mid (\hat{s}, \hat{t}) \in \hat{R}^r_i}{y_{\hat{s},\hat{t}}}} \geq \sum_{\hat{t} \preceq_{i,r} \hat{t}_{\max}}{\sum_{\substack{\hat{s} \in U^r_i \\ (\hat{s}, \hat{t}) \in \hat{R}^r_i}}{y_{\hat{s},\hat{t}}}} = \sum_{\hat{s} \in U^r_i}{\sum_{\hat{t} \mid (\hat{s}, \hat{t}) \in \hat{R}^r_i}{y_{\hat{s},\hat{t}}}} >  \frac{\gamma^r_i}{2}\]
which implies that $\hat{t}_{\max} \succeq_{i,r} \hat{t}_{i,r}$ and $\tilde{T}^r_i \subseteq \{\hat{t} \mid \hat{t} \prec_{i,r} \hat{t}_{\max}\}$. By Claim~\ref{cl:ord}, $(\hat{s}_{i,r}, \hat{t}_{\max}) \in \hat{R}^r_i$ implies that $\tilde{S}^r_i \times \tilde{T}^r_i \subseteq \tilde{R}^r_i$.
\end{proof}

\section{Missing Proofs in Theorem~\ref{thm:covering}} \label{sec:pf-covering}
We recall Algorithm~\ref{alg:covering}.

\setcounter{algorithm}{4}

\begin{algorithm}[H]
\caption{Online covering}

\begin{algorithmic}[1]

\For {arriving covering constraint $i$}
    \State $y_i \gets 0$. \Comment{the packing variable $y_i$ is used for the analysis}
    \If {$\sum_{j=1}^n{a_{ij}x^r_j} < 1$} \Comment{if constraint $i$ is not satisfied}
        \While {$\sum_{j=1}^n a_{ij}x_j^r < 2$} \Comment{update until constraint $i$ is satisfied by a factor of 2}
            \State Increase $y_i$ continuously.
            \State Increase each variable $x^r_j$ by the following increment function:
            \[x^r_j \gets \frac{\alpha(r)}{2nc_j}\exp \tuple{\frac{\ln (2n)}{c_j}\sum_{k=i_r}^{i}a_{kj}y_k}.\]
        \EndWhile
    \EndIf
\EndFor
\end{algorithmic}
\end{algorithm}

We recall that the following claims are used to show that Algorithm~\ref{alg:covering} is $O(\log n)$-competitive. 
\begin{enumerate}[label=(\roman*)]
    \item \label{pd1} $x$ is feasible.
    \item \label{pd2} For each \emph{finished} phase $r$, $\alpha(r) \leq 4 \ln (2n) \cdot Y(r)$.
    \item \label{pd3} $y$ generated during phase $r$ is feasible.
    \item \label{pd4} The sum of the covering objective generated from phase 1 to $r$ is at most $2 \alpha(r)$.
    \item \label{pd5} Let $r'$ be the last phase, then the covering objective $\rip{\mathbf{c}}{x} \leq 2 \alpha(r')$.
\end{enumerate}

\begin{proof}[Proof of \ref{pd1}.] In phase $r$, suppose constraint $i$ arrives, then either it is already satisfied by $x^r$, or it is violated and we update $x^r$ until constraint $i$ is satisfied by a factor of 2 (or start another phase until we reach a phase that satisfies constraint $i$ by a factor of 2 without exceeding the estimated lower bound $\alpha$). $x$ satisfies all the covering constraints from $1$ to $i$ because it is the coordinate-wise maximum of $\{x^k\}_{k \in [r]}$.
\end{proof}

\begin{proof}[Proof of \ref{pd2}.] In the beginning of phase $r$, $x^r_j = \alpha(r)/(2nc_j)$ so $X(r)$ is initially at most $\alpha(r)/2$. The total increase of $X(r)$ is at most $\alpha(r)/2$ because $X(r) \geq \alpha(r)$ when phase $r$ ends. Therefore, it suffices to show that
\[\frac{\partial X(r)}{\partial y_i} \leq 2 \ln(2n) \cdot \frac{\partial Y(r)}{\partial y_i}.\]
This follows because
\begin{align}
    \frac{\partial X(r)}{\partial y_i} &= \sum_{j=1}^n{c_j \frac{\partial x^r_j}{\partial y_i}} \nonumber \\
    &= \sum_{j=1}^n{\frac{c_j\ln(2n)a_{ij}}{c_j}\frac{\alpha(r)}{2nc_j} \cdot \exp \tuple{\frac{\ln (2n)}{c_j}\sum_{k=i_r}^{i}a_{kj}y_k}} \nonumber \\
    &= \ln(2n) \sum_{j=1}^n{a_{ij} x^r_j} \leq 2 \ln(2n) = 2 \ln(2n) \cdot \frac{\partial Y(r)}{\partial y_i} \label{eq:covering-2}
\end{align}
where \eqref{eq:covering-2} holds because $\sum_{j=1}^n{a_{ij} x^r_j} \leq 2$ and $Y(r) = \rip{\mathbf{1}}{y}$ implies that $\partial Y(r) / \partial y_i = 1$.
\end{proof}

\begin{proof}[Proof of \ref{pd3}.] We observe that during phase $r$ of Algorithm~\ref{alg:covering}, $x^r_j \leq \alpha(r)/c_j$, since otherwise phase $r$ is finished. Therefore,
\[x^r_j = \frac{\alpha(r)}{2nc_j}\exp \tuple{\frac{\ln (2n)}{c_j}\sum_{k=i_r}^{i}a_{kj}y_k} \leq \frac{\alpha(r)}{c_j}\]
which implies that $\sum_{k=i_r}^{i}a_{kj}y_k \leq c_j$.
\end{proof}

\begin{proof}[Proof of \ref{pd4}.] The sum of the covering objective generated from phase 1 to $r$ is at most
\[\sum_{k=1}^r \alpha(k) = \sum_{k=1}^r \frac{\alpha(r)}{2^{k-r}} \leq 2\alpha(r).\]
\end{proof}

\begin{proof}[Proof of \ref{pd5}.] In the last phase $r'$, $x$ is feasible because it is the coordinate-wise maximum of $\{x^r\}_{r \in [r']}$. We have that \[\rip{\mathbf{c}}{x} = \sum_{j=1}^n c_j x_j \leq \sum_{j=1}^n \sum_{r=1}^{r'} c_j x^r_j = \sum_{r=1}^{r'} \tuple{\sum_{j=1}^n c_j x^r_j} \leq \sum_{r=1}^{r'}\alpha(r) \leq 2 \alpha(r')\] where the first inequality holds because $x_j = \max_{r \in [r']}\{x^r_j\} \leq \sum_{r=1}^{r'} x^r_j$, the second inequality is by the fact that the covering objective $\sum_{j=1}^n c_j x^r_j$ cannot exceed the estimated lower bound $\alpha(r)$, while the last inequality is by \ref{pd4}.
\end{proof}

\section{Online Packing in Polynomial Time} \label{sec:packing}

In this section, we prove the formal version of Theorem~\ref{thm:inf-packing}. The problem of interest is to solve the packing LP \eqref{equation:packing} online:
\begin{align*}
  \begin{aligned}
    \text{maximize } & \rip{\mathbf{1}}{y} 
    \text{ over } y \in \nnreals^m 
    \text{ s.t.\ } A^T y \leq \mathbf{c}
  \end{aligned}
\end{align*}
where $A^T \in \R_{\geq 0}^{n \times m}$ consists of $n$ packing constraints with an upper bound $\mathbf{c}$ given offline.

In the online packing problem, the columns $i \in [m]$ of $A^T$ and the corresponding variables $y_i$'s are presented online as zeros, one at a time, where $m$ can be unknown. Let $a_{ij}$ be the $i$-th row $j$-th column entry of the matrix $A$ and $c_j$ be the $j$-th entry of $\mathbf{c}$. In round $i$, $\{a_{ij}\}_{j \in [n]}$, i.e., the $i$-th column of $A^T$ is revealed, and the goal is to either let the arriving variable $y_i$ remain zero, or irrevocably assign a positive value to it, such that 1) the objective value $\rip{\mathbf{1}}{y}$ is approximately optimal, and 2) each constraint $\sum_{k=1}^i a_{kj} y_k \leq c_j$ for each $j \in [n]$ is approximately satisfied.

An important idea for solving the online packing problem is to simultaneously consider the dual online covering problem LP \eqref{equation:covering}:
\begin{align*}
  \begin{aligned}
    \text{minimize } & \rip{\mathbf{c}}{x} 
    \text{ over } x \in \nnreals^n 
    \text{ s.t.\ } A x \geq \mathbf{1}.
  \end{aligned}
\end{align*}
In this problem, the cost vector $\mathbf{c}$ is given offline, and each of these covering constraints is presented one by one in an online fashion. More specifically, in round $i$, $\{a_{ij}\}_{j \in [n]}$ is revealed. The goal is to monotonically update $x$ so that the arriving constraint $\sum_{j=1}^n{a_{ij} x_j} \geq 1$ is satisfied and the objective value $\rip{\mathbf{c}}{x}$ remains approximately optimal.

Similar to Section~\ref{sec:covering}, we employ the primal-dual framework that simultaneously solves
both LP \eqref{equation:covering} and LP \eqref{equation:packing}, and crucially
uses LP-duality and strong connections between the two solutions, to
argue that they are both nearly optimal.

In the beginning of Algorithm~\ref{alg:packing}, a parameter $B > 0$ is given as an input, and $x$ is initialized to a zero vector. In round $i$, $y_i$ is introduced with the $i$-th column of $A^T$, i.e., $\{a_{ij}\}_{j \in [n]}$ is revealed. In the corresponding online covering problem, constraint $i$ arrives and is presented in the form $\sum_{j=1}^n{a_{ij}x_j} \geq 1$. Without loss of generality, we always assume that there is at least one positive entry $a_{ij}$ in round $i$, otherwise the packing objective $\rip{\mathbf{1}}{y}$ is unbounded. If the arriving constraint is violated, we increase the value of the new packing variable $y_i$ and each covering variable $x_j$ simultaneously until the new covering constraint is satisfied \emph{by a factor of 2}. We use an augmenting method in a continuous fashion, which can be implemented in a discrete way with any desired accuracy by binary search. The $x_j$'s are incremented according to an exponential function of $y_i$. We note that $y_i$ is only increased in round $i$ and fixed after then, and the $x_j$'s never decrease. The performance of the modified scheme is analyzed in Theorem~\ref{thm:packing} (the formal version of Theorem~\ref{thm:inf-packing}).

\begin{algorithm}[H]
\caption{Online packing} \label{alg:packing}

\begin{algorithmic}[1]

\For {an arriving covering constraint $i$}
    \State $y_i \gets 0$ and $a^{\max}_j \gets \max_{k \in [i]}\{a_{kj}\}$ for each $j \in [n]$.
    \If {$\sum_{j=1}^n{a_{ij}x_j} < 1$} \Comment{if constraint $i$ is not satisfied}
        \While {$\sum_{j=1}^n a_{ij}x_j < 2$} \Comment{update until constraint $i$ is satisfied by a factor of 2}
            \State Increase $y_i$ continuously.
            \State Increase each variable $x_j$ where $a^{\max}_j > 0$ by the following increment function:
            \[x_j \gets \max\Bigg\{x_j,\frac{1}{na^{\max}_j}\tuple{\exp \tuple{\frac{B}{3c_j}\sum_{k=1}^{i}a_{kj}y_k}-1}\Bigg\}.\]
        \EndWhile
    \EndIf
\EndFor
\end{algorithmic}
\end{algorithm}

\begin{theorem} \label{thm:packing}
For any $B > 0$, there exists a $1/B$-competitive online algorithm for the packing LP \eqref{equation:packing} which updates $y$ $\poly(n, \log B', \log \opt, \log \alpha)$ times. Moreover, for each constraint $j \in [n]$, the following holds
\[\sum_{i=1}^m a_{ij} y_i = c_j \cdot O\tuple{\frac{\log n + \log (a^{\max}_j/a^{\min}_j)}{B}}.\]
Here, $a^{\max}_j:= \max_{i \in [m]}\{a_{ij}\}$ and $a^{\min}_j:= \min_{i \in [m]}\{a_{ij} \mid a_{ij} > 0\}$ for each $j \in [n]$, $\alpha := \max_{j \in [n]}\{a^{\max}_j/c_j\}$, and $B':=3\max_{j \in [n]}\{\ln(2na^{\max}_j/a^{\min}_j+1)\}$.
\end{theorem}

\begin{proof}
The proof for the $1/B$-competitiveness closely follows the one in \cite{buchbinder2009online}. We provide the proof for completeness. 
Let $X(i)$ and $Y(i)$ be the covering and packing objective values, respectively, in round $i \in [m]$. The following claims are used to show that Algorithm~\ref{alg:packing} is $1/B$-competitive. 
\begin{enumerate}[label=(\roman*)]
    \item \label{p1} $x$ is feasible.
    \item \label{p2} In each round $i \in [m]$, $X(i)/B \leq Y(i)$.
    \item \label{p3} For any packing constraint $j \in [n]$,
    \[\sum_{i=1}^m a_{ij}y_i \leq c_j \frac{3\ln (2na^{\max}_j/a^{\min}_j+1)}{B}.\]
\end{enumerate}
From the first two claims together with weak duality, we conclude that Algorithm~\ref{alg:packing} is $B$-competitive, while the third claim directly proves that each packing constraint is approximately satisfied.

\begin{proof}[Proof of \ref{p1}.] This clearly holds because if constraint $i$ is already satisfied, we do nothing, otherwise update variables $x_j$'s until constraint $i$ is satisfied by a factor of 2.
\end{proof}

\begin{proof}[Proof of \ref{p2}.] We observe that in the beginning of each round, when $a^{\max}_j$ increases, $x_j$ and $X(i)$ do not change. $X(i)$ increases only when $Y(i)$ increases. Initially, both $X(i)$ and $Y(i)$ are zero. Consider round $i$ in which $y_i$ is increased continuously. We show that $\partial X(i)/\partial y_i \leq B \cdot \partial Y(i)/\partial y_i$ and conclude that $X(i)/B \leq Y(i)$.
\begin{align}
    \frac{\partial X(i)}{\partial y_i} &= \sum_{j=1}^n c_j \frac{\partial x_j}{\partial y_i} \nonumber \\
    &\leq \sum_{j=1}^n \frac{c_j}{n a^{\max}_j}\frac{Ba_{ij}}{3c_j}\exp \label{eq:packing-partial} \tuple{\frac{B}{3c_j}\sum_{k=1}^{i}a_{kj}y_k} \\
    &= \frac{B}{3}\sum_{j=1}^n \nonumber a_{ij}\tuple{\frac{1}{n a^{\max}_j}\tuple{\exp \tuple{\frac{B}{3c_j}\sum_{k=1}^{i}a_{kj}y_k}-1} + \frac{1}{n a^{\max}_j}} \nonumber \\
    &\leq \frac{B}{3}\sum_{j=1}^n a_{ij}\tuple{x_j + \frac{1}{n a^{\max}_j}} \leq \frac{B}{3}(2 + 1) = B = B \frac{\partial Y(i)}{\partial y_i} \label{eq:packing-last-line}
\end{align}
where \eqref{eq:packing-partial} follows by taking the partial derivative, and \eqref{eq:packing-last-line} follows because
\begin{enumerate}
    \item $\sum_{j=1}^n a_{ij}x_j \leq 2$ while incrementing $x_j$ in round $i$.
    \item $x_j \geq 1/(n a^{\max}_j)(\exp (B/(3c_j)\sum_{k=1}^{i}a_{kj}y_k)-1)$.
    \item $\sum_{j=1}^n a_{ij}/(na^{\max}_j) \leq 1$.
\end{enumerate}
\end{proof}

\begin{proof}[Proof of \ref{p3}.] We recall that $a^{\max}_j:= \max_{i \in [m]}\{a_{ij}\}$ and $a^{\min}_j:= \min_{i \in [m]}\{a_{ij} \mid a_{ij} > 0\}$ for each $j \in [n]$. During the run of Algorithm~\ref{alg:packing}, $x_j \leq 2/a^{\min}_j$, since otherwise each covering constraint $i \in [m]$ with $a_{ij} > 0$ is already satisfied by a factor of 2. Consider the packing constraint $j$, we have
\[\frac{1}{na^{\max}_j}\tuple{\exp \tuple{\frac{B}{3c_j}\sum_{i=1}^m a_{ij}y_i}-1} \leq x_j \leq \frac{2}{a^{\min}_j}\]
which implies that
\[\sum_{i=1}^m a_{ij}y_i \leq c_j \frac{3\ln (2na^{\max}_j/a^{\min}_j+1)}{B}.\]
\end{proof}

Now we show that Algorithm~\ref{alg:packing} updates $y$ $\poly(n, \log B', \log \opt, \log \alpha)$ times where $\opt$ is the objective value of LP \eqref{equation:packing}. It suffices to show that Algorithm~\ref{alg:packing} encounters $\poly(n, \log B', \log \opt, \log \alpha)$ violating covering constraints. We recall that $B' := 3\max_{j \in [n]}\{\log(2na^{\max}_j/a^{\min}_j+1)\}$ and $\alpha:=\max_{j \in [n]}\{a^{\max}_j/c_j\}$. Suppose we scale the packing solution $y$ by $B/B'$ such that all the packing constraints $j \in [n]$ are satisfied. Then by weak duality, since $By/B'$ is feasible, we have that
\[\frac{c_j x_j}{B'} \leq \frac{X(i)}{B'} \leq \frac{BY(i)}{B'} \leq \opt\]
for every $j \in [n]$. Therefore, $x_j = O(B'\opt/c_j)$.

We show that for each arriving covering constraint $i$ that is violated, one of the following cases must hold when the constraint is satisfied by a factor of 2 after the variable update: 1) there exist a \emph{large} variable $x_j \geq 1/(2na^{\max}_j)$ that is updated to at least $3x_j/2$, or 2) there exists a \emph{small} variable $x_j < 1/(2na^{\max}_j)$ that becomes large, i.e., $x_j$ is updated to at least $1/(2na^{\max}_j)$. Let $L$ and $S$ be the set of large and small variable subscript labels before the update, respectively, and $x'_j$ be the value of $x_j$ after the update. If none these two cases holds, then
\[\sum_{j=1}^n a_{ij} x'_j < \frac{3}{2}\sum_{j \in L}a_{ij}x_j + \sum_{j \in S}\frac{a_{ij}}{2na^{\max}_j} < \frac{3}{2} + \frac{1}{2} = 2\]
where the second inequality is by the fact that constraint $i$ is violated and $a_{ij} \leq a^{\max}_j$.
This implies that constraint $i$ is not satisfied by a factor of $2$, a contradiction.

Suppose $x_j$ has been updated $t$ times by a factor of $3/2$ since it was large, then
\[\frac{1}{2na^{\max}_j} (\frac{3}{2})^t =  O(\frac{B'\opt}{c_j})\]
which implies $t=O(\log n + \log \opt + \log B' + \log (a^{\max}_j/c_j))$.

There are $n$ variables, each variable can be updated from small to large once and updated $t$ times by a factor of $3/2$ since it was large. Hence, Algorithm~\ref{alg:packing} encounters $\poly(n, \log B', \log \opt, \log \alpha)$ violating covering constraints.
\end{proof}

\end{document}